\documentclass[%
 reprint,
 superscriptaddress,
 amsmath,amssymb,
 aps,
 prx,
 floatfix,
]{revtex4-2}

\usepackage{silence}
\usepackage[bookmarksnumbered,linktocpage,hypertexnames=false,colorlinks=true,linkcolor=blue,urlcolor=blue,citecolor=blue,anchorcolor=green,breaklinks=true,pdfusetitle]{hyperref}
\usepackage{bm,bbm,bbold,braket,amsthm,amsfonts,mathtools,graphicx,enumitem,float,mathdots,enumitem,mleftright,framed,overpic,booktabs}

\usepackage[dvipsnames]{xcolor}
\usepackage{tikz}
\usepackage{quantikz}
\usepackage{dsfont}
\usepackage{soul} % strike through with \st{}
\usepackage{printlen}
\usepackage{chemformula}

\usepackage[longend,ruled,linesnumbered,nokwfunc]{algorithm2e}
\SetKwInput{Input}{Input}
\SetKwInput{Output}{Output}
\SetKwFor{Loop}{loop}{}{end loop}

\usepackage[capitalize,nameinlink]{cleveref}

\tikzstyle{box} = [rectangle, rounded corners, 
minimum width=3cm, 
minimum height=1cm,
text centered]
\tikzstyle{arrow}=[->]

\newtheorem{thm}{Theorem}[section]\crefname{thm}{Theorem}{Theorems}
\newtheorem*{thm*}{Theorem}
\newtheorem{lem}[thm]{Lemma}\crefname{lem}{Lemma}{Lemmas}
\crefname{prop}{Proposition}{Propositions}
\crefname{cor}{Corollary}{Corollaries}
\theoremstyle{definition}
\crefname{def}{Definition}{Definitions}
\theoremstyle{remark}

\newcommand{\RR}{\mathbbm R}
\newcommand{\CC}{\mathbbm C}

\newcommand{\ZZ}{\mathbbm Z}
\newcommand{\HH}{\mathcal H}
\newcommand{\EE}{\mathbb E}
\newcommand{\prob}{\mathbb P}
\newcommand{\rv}[1]{\bm{#1}}

\newcommand{\ot}{\otimes}
\newcommand{\eps}{\varepsilon}
\newcommand{\id}{I}
\newcommand{\bigO}{\mathcal O}
\newcommand{\der}{\textrm{d}}

\newcommand{\cnot}{\textsc{cnot }}
\newcommand{\trot}{S}
\newcommand{\sloperpe}{D}

\DeclareMathOperator{\real}{Re}

\DeclareMathOperator{\tot}{tot}
\DeclareMathOperator{\even}{even}
\DeclareMathOperator*{\argmin}{arg\,min}
\DeclareMathOperator{\stage}{stage}
\DeclareMathOperator{\qpe}{qpe}
\DeclareMathOperator{\trotter}{trot}
\DeclareMathOperator{\rand}{rand}
\DeclareMathOperator{\synth}{synth}
\DeclareMathOperator{\rot}{rot}
\DeclareMathOperator{\op}{op}
\DeclareMathOperator{\gs}{gs}

\DeclarePairedDelimiter\abs{\lvert}{\rvert}
\DeclarePairedDelimiter\norm{\lVert}{\rVert}
\DeclarePairedDelimiter\ceil{\lceil}{\rceil}

\begin{document}

\title{Phase estimation with partially randomized time evolution}

\date{\today}
\author{Jakob G\"{u}nther}
\affiliation{Department of Mathematical Sciences, University of Copenhagen, Universitetsparken 5, 2100 Copenhagen, Denmark}
\author{Freek Witteveen}
\affiliation{Department of Mathematical Sciences, University of Copenhagen, Universitetsparken 5, 2100 Copenhagen, Denmark}
\affiliation{QuSoft and CWI, Science Park 123, 1098 XG Amsterdam}
\author{Alexander Schmidhuber}
\affiliation{Center for Theoretical Physics, MIT, Cambridge, MA, 02139, USA}
\author{Marek Miller}
\affiliation{Department of Mathematical Sciences, University of Copenhagen, Universitetsparken 5, 2100 Copenhagen, Denmark}
\author{Matthias Christandl}
\affiliation{Department of Mathematical Sciences, University of Copenhagen, Universitetsparken 5, 2100 Copenhagen, Denmark}
\author{Aram W. Harrow}
\affiliation{Center for Theoretical Physics, MIT, Cambridge, MA, 02139, USA}

\begin{abstract}
    Quantum phase estimation combined with Hamiltonian simulation is the most promising algorithmic framework to computing ground state energies on quantum computers.
    Its main computational overhead derives from the Hamiltonian simulation subroutine.
    In this paper we use randomization to speed up product formulas, one of the standard approaches to Hamiltonian simulation.
    We propose new partially randomized Hamiltonian simulation methods in which some terms are kept deterministically and others are randomly sampled.
    We perform a detailed resource estimate for single-ancilla phase estimation using partially randomized product formulas for benchmark systems in quantum chemistry and obtain orders-of-magnitude improvements compared to other simulations based on product formulas.
    When applied to the hydrogen chain, we have numerical evidence that our methods exhibit asymptotic scaling with the system size that is competitive with the best known qubitization approaches.
\end{abstract}

\maketitle

\section{Introduction}

Among the most promising applications of quantum computing is the calculation of ground state energies in quantum many-body physics and quantum chemistry.
While in general it is hard for a quantum computer to estimate the ground state energy, quantum phase estimation (QPE) provides an efficient algorithm when a state with significant overlap to the ground state is given.
The complexity of the resulting algorithms for chemistry depends largely on the complexity of simulating the Hamiltonian time evolution as a quantum circuit, as QPE requires time simulation of time $\bigO(\eps^{-1})$ to compute the energy to accuracy $\eps$.
The dominant cost of performing QPE is then given by the cost of Hamiltonian evolution and its scaling with time and system size.
State-of-the-art resource estimates for realistic examples from chemistry show that the computational cost of performing QPE is still formidable \cite{childs2018toward,von2021quantum,su2021fault,lee2021even}. This means that for realizing useful simulations of quantum chemistry, further reductions in the circuit size required to perform Hamiltonian simulation are necessary.
This is particularly relevant for devices with a limited number of logical qubits and limited circuit depth.
See \cite{bauer2020quantum} for a general discussion of QPE for quantum chemistry.

Deterministic product formulas, also known as Trotterization, have the advantage of being conceptually simple and to require little or no ancilla qubits for their implementation as a quantum circuit.
In many situations, their performance can be better than what can be expected from rigorous bounds \cite{babbush2015chemical,childs2018toward} or can even be proven to have good scaling, e.g. in systems with spatially local interactions \cite{childsTheoryTrotterError2021,low2023complexity,su2021nearly}.
There also exist powerful post-Trotter Hamiltonian simulation methods \cite{low2017optimal,low2019hamiltonian}; while they have better scaling, methods based on product formulas are still of great interest due to the small number of ancilla qubits required and their good performance for Hamiltonians of interest \cite{childs2018toward}.

This work is motivated by quantum chemical Hamiltonians.
For such a Hamiltonian, defined on an active space of $N$ spatial orbitals (corresponding to $2N$ qubits in second quantization), the number of terms in the Hamiltonian scales as $L = \bigO(N^4)$.
Trotter product formulas have a computational cost which scales with $L$.
While this leads to polynomial scaling in $N$, the resulting circuits rapidly grow impractically large.

An alternative is provided by \emph{randomized product formulas}, in particular qDRIFT \cite{campbell2019random}.
Here, the number of gates no longer depends on $L$, but rather on $\lambda$, the sum of the interaction strengths in the Hamiltonian.
If one wants to simulation Hamiltonian evolution for time $t$, this requires $\bigO(\lambda^2 t^2)$ gates.
So, while this gets rid of the dependence on $L$, the downside is that the scaling is quadratic in the time (and $\lambda$).
In particular, since QPE with target accuracy $\eps$ requires Hamiltonian simulation for times up to $\bigO(\eps^{-1})$, using randomized product formulas leads to phase estimation costs scaling with $\eps^{-2}$ rather than the optimal Heisenberg-limited scaling of $\eps^{-1}$.
Prior work uses an analysis which suggests that combining QPE with qDRIFT leads to an unfavorable scaling of the total cost with the accuracy of the computation \cite{campbell2019random}.
In particular, a naive analysis leads to a scaling of order $\eps^{-4}$ or an undesirable dependence on the probability of error \cite{lee2021even}.
Using a different randomization scheme, however, as well as a different phase estimation scheme, one can restore the scaling to $\eps^{-2}$ \cite{wan2022randomized}.

A natural generalization of randomized product formulas is to consider \emph{partially randomized} product formulas \cite{ouyang2020compilation,jin2023partially,hagan2023composite}, which are effective if there is a relatively small number of dominant terms in the Hamiltonian, and a long tail of small terms (which are nevertheless too large to be ignored altogether). Such behavior is typical for quantum chemistry Hamiltonians; see \cref{sec:electronic structure main text} and \cref{sec:electronic structure problem} for a detailed discussion.

\subsection{Results and organization}

\begin{figure}
    \centering
    \includegraphics[width=\linewidth]{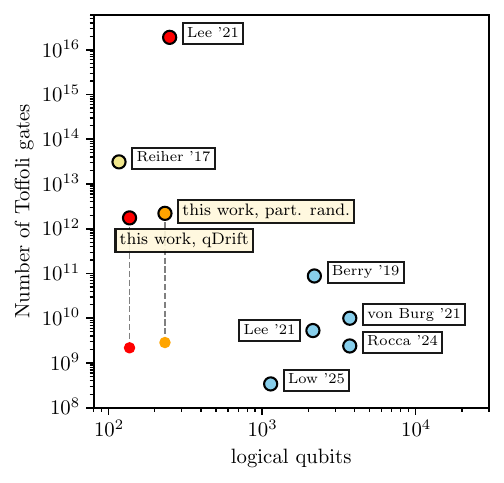}
    \caption{Total Toffoli gate count and logical qubit estimates for ground state energy estimation for FeMoco using the active space of \cite{reiher2017elucidating}, with target error $\eps = 0.0016$ and assuming an initial state equal to the ground state. The blue-colored data points use qubitization \cite{berry2019qubitization,von2021quantum,lee2021even,rocca2024reducing,low2025fast}, the yellow data point is the most optimistic estimate from \cite{reiher2017elucidating} using a deterministic product formula, the red data points use randomized product formulas from \cite{lee2021even} and this work, the orange data point uses partial randomization.
        We also show the Toffoli depth (i.e.~the maximum number of Toffoli gates per circuit), assuming a high ground state overlap state and choosing $\xi = 0.1$, using (partially) randomized product formulas. The vertical dashed lines connect the Toffoli depth to the total Toffoli count, to illustrate how multiple circuit runs are used.
    }
    \label{fig:femoco intro}
\end{figure}

We start by setting up notation for the relevant class of Hamiltonians and single-ancilla phase estimation in \cref{sec:qpe} and introduce background material on the Trotter formalism in \cref{sec:trotter}.
We discuss randomized product formulas in \cref{sec:random product}.
Our first contribution is an improved analysis for the cost of randomized product formulas for single-ancilla QPE, achieving reduced cost compared to previous estimates \cite{lee2021even,wan2022randomized,campbell2019random} and reduced maximal circuit depth. This result is stated as \cref{thm:rpe qDRIFT main text}.
In \cref{fig:femoco intro} we show the resulting resource estimates for the challenge benchmark problem FeMoco \cite{reiher2017elucidating} and compare with previous work. We see multiple order of magnitude improvements in the performance of the randomized product formulas.
Our improvement compared to \cite{lee2021even} by more than two orders of magnitude breaks down as follows.
We find a value of $\lambda$ which is a factor 5 smaller, reducing the cost by a factor 25. The overhead of the phase estimation routine is reduced by a factor of 38 (while improving accuracy guarantees). Indeed, \cite{lee2021even} finds a requirement of approximately $305 \lambda^2 \eps^{-2}$ Pauli rotations, while we require approximately $8\lambda^2 \eps^{-2}$ Pauli rotations. Finally, as explained in \cref{sec:gate count}, we give a compilation from Pauli rotations into Toffoli gates using a factor 10 fewer gates by combining techniques from \cite{gidney2018halving,sanders2020compilation,lee2021even,wan2022randomized}.

\cref{fig:femoco intro} assumes an exact eigenstate, and aims to minimize total resources. More realistically, we should assume we start with a state with (high) ground state overlap.  Our methods can accommodate such states, but using an exact eigenstate facilitates comparison with the previous methods.
If one can prepare a state with high ground state overlap, using robust phase estimation in combination with randomized product formulas has the effect of using multiple circuits of smaller size. This can be useful for early fault-tolerant devices and paralellization.
To make this concrete, let us again consider the example of FeMoco.
If one indeed can prepare exact eigenstates, or states with sufficiently high ground state overlap, one can trade the gate count per circuit further against using more circuits. If we choose the trade-off parameter $\xi = 0.1$ in \cref{thm:rpe qDRIFT main text}, we need 4741 circuits, with a maximum of $6.4 \times 10^8$ Pauli rotations per circuit (as also indicated in \cref{fig:femoco intro}).
This choice of $\xi$ is also compatible with a state with ground state overlap $\eta > 0.9$, and it has been argued that for a larger active space of FeMoco achieving a squared overlap of approximately $0.95$ is realistic (but the cost of preparing this state is significant) \cite{berry2024rapid}.
As is also clear from \cref{fig:femoco intro}, the partially randomized method (which scales better with the precision $\eps$) benefits less from this reduced circuit depth.

\begin{figure}
    \centering
    \includegraphics[width=\linewidth]{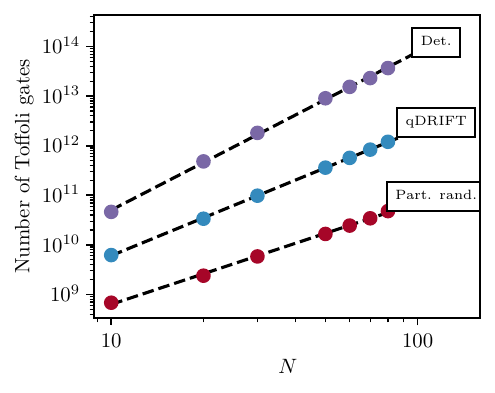}
    \caption{Gate counts for ground state energy phase estimation for precision $\eps = 0.0016$ for an $N$ atom hydrogen chain with a minimal basis set, using deterministic, randomized and partially randomized product formulas.}
    \label{fig:hchain cost}
\end{figure}

The main contribution of this work is the study of partially randomized product formulas.
In \cref{sec:composite} we propose a modification of the scheme of Hagan and Wiebe~\cite{hagan2023composite}. We give a simple error analysis for our scheme when used in single-ancilla QPE, improving on the bounds that would be obtained by analyzing the product formula and the QPE separately.
As an example of a significant cost improvement from partially randomized product formulas, we show in \cref{fig:hchain cost} cost estimates for the hydrogen chain with a minimal basis (a standard benchmark system to explore the scaling of the thermodynamic limit). We observe an improvement of more than an order of magnitude for partial randomization compared to fully randomized product formulas, which in turn outperform deterministic product formulas.
Additionally, when varying both the required accuracy $\eps$ and the number of hydrogen atoms $N$, for the fully randomized method we have a scaling of $\bigO(N^{2.5} \eps^{-2})$, while for the partially randomized product formula we get $\bigO(N^{2} \eps^{-1.7})$.
While the scaling with $\eps$ is suboptimal, the scaling with $N$ is competitive with qubitization using hypertensor contraction \cite{lee2021even}, which has a scaling of $\bigO(N^{2.1} \eps^{-1})$.
% For FeMoco, we observe a modest improvement compared to fully randomized product formulas (see \cref{fig:femoco intro}). The increased number of ancilla qubits is due to the compilation strategy from Pauli rotations to Toffoli gates.

For electronic structure problems, we show that partial randomization can be used to improve the performance of product formulas which decompose the Hamiltonian based on a \emph{factorization} of the Coulomb tensor \cite{motta2021low,von2021quantum}.
% An example is the data point in \cref{fig:femoco intro} which gives a (modest) improvement over the fully randomized product formula.
For our quantum chemistry resource estimates, we perform a numerical estimate of the Trotter error for the hydrogen chain and a benchmark set of small molecules, and optimize of the representation of the Hamiltonian for improved performance.

Our work performs an end-to-end analysis, studying aspects of the modeling of electronic structure Hamiltonians, Hamiltonian simulation and its use as a subroutine in phase estimation.
In \cref{fig:overview} we give an overview of our main technical contributions and their relation.

\begin{figure}
\begin{tikzpicture}
\def\w{2.2}  % rectangle width
\def\h{1.5}  % rectangle height

% draw rectangle centered at x,y with: \draw (x-\w/2,y-\h/2) rectangle (x+\w/2,y+\h/2);
\filldraw[fill=green!20!white,draw=black,rounded corners] (0 -\w/2, 1-\h/2) rectangle (0+\w/2, 1+\h/2);
\node at (0,1) [text width=2cm, text centered] {\textbf{Hamiltonian simulation}};
\filldraw[fill=blue!20!white,draw=black,rounded corners] (6-\w/2, 1-\h/2) rectangle (6+\w/2, 1+\h/2);
\node at (6, 1) [text width=2cm, text centered] {\textbf{quantum phase estimation}};
\filldraw[fill=red!20!white,draw=black,rounded corners] (3-\w/2, -3-\h/2) rectangle (3+\w/2, -3+\h/2);
\node at (3, -3) [text width=2cm, text centered] {\textbf{chemistry model}};

% text at boxes
\node at (0,2) [anchor=south,text width=2.8cm, text centered] {improved compilation (\cref{sec:gate count})};
\node at (5.8,2) [anchor=south,text width=3cm,text centered] {cost of single-ancilla QPE (\cref{sec:qpe,sec:rpe})};
\node at (3,-4) [anchor=north,text width=6cm, text centered] {optimized Hamiltonian representation \& random/deterministic partition (\cref{sec:electronic structure main text,sec:optimize hamiltonian})};

% lines with text
\draw (0+\w/2, 1) -- (6 - \w/2, 1) node[midway,text width=2.7cm, text centered,fill=white] {analysis of part. random product formulas via LCUs (\cref{sec:composite})};
\draw (0, 1-\h/2) -- (3-\w/2, -3) node[midway,pos=0.4,text width=3.3cm, text centered,fill=white] {empirical Trotter error estimates (\cref{sec:electronic structure main text,sec:trotter appendix})};
\draw (3 + \w/2, -3) -- (6, 1-\h/2) node[midway,pos=0.6,text width=3cm,text centered,fill=white] {high-overlap enabled depth reduction (\cref{sec:random product})};

\end{tikzpicture}
\caption{An overview of the topics and contributions in this work. The combined improvements lead to the resource estimates in \cref{fig:femoco intro} and \cref{fig:hchain cost}.}
\label{fig:overview}
\end{figure}

\subsection{Comparison to previous work}
We first discuss previous work on phase estimation using randomized product formulas.
Previous state-of-the-art resource estimates for QPE using qDRIFT were given in \cite{lee2021even}.
That work presents three resource analyses of phase estimation using qDRIFT, based on diamond norm bounds for the qDRIFT error. Directly requiring mean squared error $\eps^2$ lead to a scaling of the number of gates of $\bigO(\lambda^4 \eps^{-4})$, which is excessively expensive.
The data point on \cref{fig:femoco intro} is their improvement which uses a less stringent error measure, namely a $95\%$ confidence interval of width $\eps$.
Finally, \cite{lee2021even} gives a third approach, where the goal is to use a Hodges-Lehmann estimator on multiple samples of the QPE outcome. This reduces the phase estimation overhead by an order of magnitude, but relies on the non-rigorous assumption that the qDRIFT error is symmetric.
For reference, \cref{fig:femoco intro} also compares with state-of-the-art estimates based on qubitization. These have better scaling in the precision and better overall performance, but do require more ancilla qubits in order to implement linear combinations of unitaries.

For single-ancilla phase estimation diamond norm bounds are not needed, and this was used in \cite{wan2022randomized}. That work additionally proposed a modified randomized product formula, which we call Randomized Taylor Expansion (RTE). This has the property that it can be used to give an unbiased estimator of the time evolution unitary in single-ancilla QPE.
Our analysis extends this idea to qDRIFT, which can be seen as only keeping the first order in the Taylor expansion.
We also note that if one can prepare a state with high ground state overlap, one can break up the computation in multiple smaller circuits.
We can compare our estimates for FeMoco to the ones in \cite{wan2022randomized}.
The algorithm in \cite{wan2022randomized} is based on computing the cumulative density function of the spectrum, $C(x)$, and finding the ground state energy as the first jump of this function.
Here, we aim for error $\eps$ with failure probability at most $0.1$.
Assuming an exact eigenstate and finding a crossing point for $C(x) < 0.3$ or $C(x) > 0.7$, leads to an estimate of approximately $4500$ %6 * 750
circuits with an \emph{average} number of $7.9 \times 10^{9}$ % 1.1 * (1511/405)^2 * 10^11
Pauli rotations.
Here, we have used the results from Figure 2 in \cite{wan2022randomized}, and rescaled to use our optimized value of $\lambda$. %1511 -> 405
We see that robust phase estimation requires shorter depth circuits (without increasing the total required number).

Partial randomization schemes have also been studied previously.
There are various different proposals \cite{kivlichan2019phase,ouyang2020compilation,jin2023partially,hagan2023composite,rajput2022hybridized,kiss2023importance}.
Our scheme is a modification of that proposed in \cite{hagan2023composite}.
The difference is that we replace qDRIFT by the RTE technique from \cite{wan2022randomized} to `unbias' the product formula, which is helpful when using single-ancilla QPE.
Indeed, this means that when analyzing QPE, we do not need to use diamond norm bounds (as derived in \cite{hagan2023composite}), which would lead to worse constant factors. Since we find that in practice partial randomization tends to lead to (modest) constant improvements, it is important to have a sharp cost analysis.
% Numerical explorations on partially randomized product formulas \cite{pocrnic2024composite}

\section{Single-ancilla quantum phase estimation}\label{sec:qpe}
We consider Hamiltonians which are a sum of local terms $H = \sum_l H_l$.
A special case of interest is when
\begin{align}\label{eq:pauli hamiltonian}
    H = \sum_{l=1}^L H_l = \sum_{l=1}^L h_l P_l, \qquad h_l \in \RR
\end{align}
where $P_l \in \pm \{\id, X, Y, Z\}^{\ot N_q}$ are Pauli operators on $N_q$ qubits.
Our main motivation derives from molecular electronic structure Hamiltonians in second quantization, as we will discuss in \cref{sec:electronic structure main text} onward.
In that case, we have $N_q = 2N$, where $N$ is the number of molecular orbitals, and $L = \bigO(N^4)$.
We define the \emph{weight} of the Hamiltonian in \cref{eq:pauli hamiltonian} to be
\begin{align}\label{eq:lambda}
    \lambda = \sum_l \abs{h_l}.
\end{align}
The eigenvalues of $H$ and the associated eigenvectors are denoted by $E_k$ and $\ket{\psi_k}$, respectively, where $E_0 \leq E_1 \leq E_2 \leq \dots$.
Our task is to compute the ground state energy $E_0$ up to precision $\eps > 0$.
If we know a \emph{guiding state} $\ket{\psi}$ with significant overlap to the ground state $\ket{\psi_0}$, we can use a quantum computer to approximate $E_0$ by applying QPE to the time evolution operator and the initial state: $e^{-it H} \ket{\psi}$.
For many problems involving the electronic structure of molecules, it is feasible to find good guiding states using established methods see \cite{tubman2018postponing,lee2023evaluating,fomichev2023initial,erakovic2024high,berry2024rapid,morchen2024classification} for elaborate discussion and further references.

We consider versions of QPE based using only a single ancilla qubit, which use the \emph{Hadamard test}, illustrated in \cref{fig:hadamard}.
Performing QPE using only the Hadamard test (without a quantum Fourier transform circuit) has an extensive history \cite{kitaev1995quantum,dobvsivcek2007arbitrary,kimmel2015robust,wiebe2016efficient,somma2019quantum,lin2022heisenberg}.

We start with the guiding state $\ket{\psi}$, and a single ancilla qubit is set to $\ket{+} = H\ket{0}$ .
Having applied $U$, controlled on the ancilla qubit, we measure the ancilla with respect to the basis $\ket{0} \pm e^{i\theta} \ket{1}$.
This gives a random variable $\rv{Z}_{\theta}$ which has outcomes $\pm 1$.
The Hadamard test has expectation values
\begin{align}
    \EE \rv{Z}_{\theta} = \real e^{i\theta} \bra{\psi} U \ket{\psi}.
\end{align}
In particular, letting $\rv{X}$ and $\rv{Y}$ denote the random variables obtained from $\theta = 0, \pi/2$ and setting $\rv{Z} = \rv{X} + i\rv{Y}$, we have
\begin{align}
    \EE \rv{Z} = \bra{\psi} U \ket{\psi}.
\end{align}
We will focus on the case when $U = U(t) := e^{-it H}$.

If we denote by $\rv{Z}_{\theta}(t)$ the outcomes of the Hadamard test using $U = U(t)$, then
\begin{align}\label{eq:signal hadamard test}
    \begin{split}
        g(t) := \EE \rv{Z}(t) & = \bra{\psi} \exp(-iHt) \ket{\psi} \\
                              & = \sum_k c_k \exp(-i E_k t),
    \end{split}
\end{align}
where $c_k = \abs{\braket{\psi | \psi_k}}^2$. We can approximate $g(t)$ by taking multiple independent samples $\rv{Z}^{(n)}(t)$ and take the mean $\overline{\rv{Z}}(t) = \frac{1}{N} \sum_{n=1}^N \rv{Z}^{(n)}(t)$.

If we think of $g(t)$ as a time signal, then the phase estimation routine will constitute a signal processing transformation to compute the lowest frequency of $g(t)$ (corresponding to the energy $E_0$) \cite{somma2019quantum}, provided that we have some guarantee on the overlap of $\ket{\psi}$ with the ground state; we assume a lower bound $c_0 \geq \eta$.
With appropriate signal processing methods, one can find the value of $E_0$ with accuracy $\eps$ using $m$ circuits with time evolution for times $t_1, \dots, t_m$. This can be done such that the maximal time evolution $t_{\max} = \max \{t_1, \dots, t_m\}$ and the total time over all circuit runs $t_{\tot} = t_1 + t_2 + \dots + t_m$ both scale as $\eps^{-1}$.
This \emph{Heisenberg scaling} is known to be optimal \cite{giovannettiQuantumMetrology2006}.
When using a Hamiltonian simulation method to simulate the time evolution, $t_{\max}$ is what determines the maximal number of gates per circuit and $t_{\tot}$ determines the total gate count.
We measure the error of the QPE by its (root) mean square error.
In our resource estimates we consider two different scenarios. In the first scenario, we are given access to the exact eigenstate $\ket{\psi_0}$ (so $c_0 = 1$). In the more realistic second scenario, we consider a state $\ket{\psi}$ with an overlap bound $\eta < 1$.

\begin{figure}
    \centering
    \begin{tikzcd}
        \lstick{$\ket{0}$} & \gate{H} & \ctrl{1} & \gate{R(\theta)} & \gate{H} & \meter{}\\
        \lstick{$\ket{\psi}$} & \qw & \gate{U} & \qw & \qw & \qw
    \end{tikzcd}
    \caption{The Hadamard test for estimating $\rv{Z}(\theta)$. Here $R(\theta) = \proj{0} + e^{i\theta} \proj{1}$ is a phase gate.}
    \label{fig:hadamard}
\end{figure}

There are various ways to extract the ground state energy by sampling from $\rv{Z}(t)$ which achieve Heisenberg scaling \cite{wiebe2016efficient,lin2022heisenberg,ding2023even}.
A particularly elegant method is \emph{robust phase estimation} \cite{higgins2009demonstrating,kimmel2015robust,belliardo2020achieving,ni2023low}.
For now, we assume that the Hadamard test is implemented using \emph{exact} time evolution, and we measure the cost in terms of the time evolution required.
% We also assume that the Hamiltonian $H$ is normalized to have spectrum contained in $[-\pi, \pi]$.
This method assumes that $\eta$ is sufficiently large.
The idea of robust phase estimation is to estimate the time evolution signal at times $t = 2^m$ for $m = 1, \dots, M$ with $M = \ceil{\log \eps^{-1}}$.
For each $m$, the expectation value of the Hadamard test outcome $\rv{Z}(2^m)$ is
\begin{align}
    g(2^m) = \sum_k c_k \exp(-i 2^m E_k).
\end{align}
We repeat the test $N_m$ times and obtain the average $\overline{\rv{Z}}(2^m)$.
The integer $M$ corresponds to the number of precision bits in the estimate of $E_0$.
With each round we update a guess $\theta_m$ for $E_0$.
Given an outcome $\overline{z}(2^m) = r_m\exp(- i \phi_m)$ of the random variable $\overline{\rv{Z}}(2^m)$ and an estimate $\theta_{m-1}$ for $E_0$ from the previous round, the new estimate $\theta_m$ is the number which is compatible with the value $\phi_m$ in the sense that $2^m \theta_m$ equals $\phi_m$ modulo $2\pi$, and is close to the estimate for $\theta_{m-1}$ from the previous round.
That is, $\theta_m = 2^{-m}(\phi_m + 2\pi k)$ for the integer $k$ such that this is closest to $\theta_{m-1}$.
It can be proven \cite{higgins2009demonstrating,kimmel2015robust,belliardo2020achieving,ni2023low} that for $\eta \geq 4 - 2\sqrt{3} \approx 0.54$ this gives a correct algorithm with maximal evolution time $t_{\max} = \bigO(\eps^{-1})$ and total evolution time $t_{\tot} = \bigO(\eps^{-1})$.
In molecular electronic structure problems, in most cases it is possible to reach this ground state overlap regime \cite{tubman2018postponing,fomichev2023initial,erakovic2024high,berry2024rapid}.
In situations with smaller values of $\eta$, there exist various other single-ancilla phase estimation schemes \cite{lin2022heisenberg,obrien2019quantum,ding2023even}. The methods in this work apply directly to these phase estimation schemes as well.
The robust phase estimation scheme has the advantage that if the overlap is close to 1, one can reduce the maximal required time evolution, at the cost of increasing the total time evolution \cite{ni2023low}.
Indeed, if $\eta = 1 - \alpha$ for small $\alpha$, we may take a trade-off parameter $\Omega(\alpha) = \xi \leq 1$ (i.e. we can take $\xi$ small if the ground state overlap is large), and obtain
\begin{align}\label{eq:rpe max time reduction}
    t_{\max} = \bigO\left(\frac{\xi}{\eps}\right), t_{\tot} = \bigO\left(\frac{1}{\xi \eps} \right).
\end{align}
This result is due to \cite{ni2023low}; for completeness we provide a precise statement and proof in \cref{thm:rpe}.

\section{Deterministic product formulas}\label{sec:trotter}
The previous section involved approximating $e^{-iHt}$ as a subroutine.  In this section, we explain deterministic product formulas, also known as the Trotter method, one of the leading approaches to Hamiltonian simulation, on which our work will build.
We refer to \cite{childsTheoryTrotterError2021} for an extensive discussion and references.
Given a Hamiltonian $H = \sum_{l=1}^L H_l$ consisting of $L$ terms, we may approximate for a short time evolution for time $\delta$ as
\begin{align}
    U(\delta) \approx \trot_1(\delta) := e^{-i \delta H_L}e^{-i \delta H_{L-1}} \cdots e^{-i \delta H_2} e^{-i \delta H_1}.
\end{align}
This is a first order product formula. A second order product formula is
\begin{align}
    \trot_2(\delta) := \left(e^{-\frac{i}{2}\delta H_1} \cdots e^{-\frac{i}{2}\delta H_L} \right) \left(e^{-\frac{i}{2}\delta H_L} \cdots e^{-\frac{i}{2}\delta H_1}\right)
\end{align}
and higher order Suzuki-Trotter product formulas can be recursively defined by
\begin{align}
    \trot_{2k}(\delta) = (\trot_{2k-2}(u_k \delta))^2 \trot_{2k-2}((1 - 4u_k)\delta) (\trot_{2k-2}(u_k \delta))^2
\end{align}
for $u_k = (4 - 4^{1 / (2k - 1)})^{-1}$.
For the Trotter-Suzuki product formula of order $2k$, $\trot_{2k}(\delta)$ requires $2L5^{k-1}$ unitaries of the form $e^{-i \varphi H_l}$.
Compiling $e^{-i \varphi H_l}$ into elementary gates gives a quantum circuit. For example, if $H_l = h_l P_l$ are Pauli operators, $e^{-i \varphi H_l}$ is a Pauli rotation, which can straightforwardly be compiled into elementary gates.

In general, we write $\trot_p(\delta)$ for a $p$-order product formula; meaning that $\trot_p(\delta)$ is a product of time evolutions along the $H_{l}$.
The order refers to the order of approximation in the sense that the error between $U(\delta)$ and $\trot_p(\delta)$ is $\bigO(\delta^{p+1})$.
Time evolution for arbitrary time $t$ can be approximated by dividing $t$ into a sufficiently large number of intervals $r$ and discretize $U(t)$ into $r$ \emph{Trotter steps}
\begin{align}
    U(t) = U(t/r)^r \approx \trot_p(t/r)^r
\end{align}
with time step $\delta = t / r$.

We will now discuss the error due to the approximation of $U(\delta)$ by $\trot_p(\delta)$, which we will refer to as the \emph{Trotter error}, in more detail.
If one wants to simulate $U(\delta)$, a good error measure is the operator norm difference, which scales as
\begin{align}
    \norm{U(\delta) - \trot_p(\delta)}_{\op} \leq C_{\op} \delta^{p+1}.
\end{align}
Here we have written the prefactor
\begin{align}
    C_{\op} = C_{\op}(p,\{H_l\})
\end{align}
explicitly. We refer to this constant (and similar parameters, with respect to different error measures, defined below) as the \emph{Trotter constant}. While constant with respect to the choice of step size $\delta$, it does depend on the order $p$ and the Hamiltonian and its decomposition into terms $H_l$ (and in that way on the system size). This dependence is important.

For the purpose of phase estimation for ground state energies, we can consider a different error metric.
For fixed $\delta$, we may write $\trot_{p}(\delta) = \exp(-i \delta \tilde H)$ for some effective Hamiltonian $\tilde H$ which is close to $H$ for small $\delta$.
If we apply QPE to the unitary $S_p(\delta)$, this is equivalent to performing QPE to the effective Hamiltonian $\tilde H$, and we can find an estimate for $\tilde E_0$, the ground state energy of $\tilde H$ (note that $\tilde H$ depends on $\delta$ as well).
The outcome of the QPE now has a \emph{bias} due to the Trotter error, and in order to obtain an estimate of $E_0$ to precision $\eps$, we need that $\abs{E_0 - \tilde E_0} = \bigO(\eps)$.
The scaling of this bias is
\begin{align}\label{eq:ground state error trotter}
    \abs{E_0 - \tilde{E}_0} \leq C_{\gs} \delta^{p}
\end{align}
for a Trotter constant $C_{\gs} = C_{\gs}(p, \{H_l\})$.
This means that $\delta = \bigO((C_{\gs}^{-1}\eps)^{1/p})$ suffices for precision $\eps$.
Phase estimation needs $\bigO(\delta^{-1}\eps^{-1})$ applications of $\trot_{p}(\delta)$ to estimate $\tilde E_0$ to precision $\bigO(\eps)$.
Together, this yields a scaling of $\bigO(C_{\gs}^{\frac{1}{p}}\eps^{-1 - \frac{1}{p}})$ Trotter steps.

\begin{figure*}[t]
    \centering
    \begin{quantikz}
        \lstick{$\ket{0}$} & \gate{H} & \ctrl{1}\gategroup[2,steps=3,style={dashed, color=gray, inner sep=6pt}]{Controlled $\trot_p(\delta)$} & \qw \ \ldots\ & \ctrl{1} & \qw \ \ldots \ & \gate{R(\theta)} & \gate{H} & \meter{}\\
        \lstick{$\ket{\psi}$} & \qw & \gate{e^{-i \delta_1 H_{l_1}}} & \qw \ \dots\ & \gate{e^{-i \delta_I H_{l_I}}} & \qw \ \ldots \ & \qw & \qw & \qw
    \end{quantikz}
    \vspace{0.1cm}

    \caption{Performing the  Hadamard test on $\trot_p(\delta)^s$ allows one to estimate $\bra{\psi} \trot_p(\delta)^{s} \ket{\psi} = \bra{\psi} \exp(-i\delta s \tilde{H}) \ket{\psi}$.}
    \label{fig:trotter circuit}

\end{figure*}

The operator norm error can be bounded in the following way \cite{childsTheoryTrotterError2021}: %Lemma 1 in \cite{childsTheoryTrotterError2021} for $\lambda$
\begin{align}\label{eq:operator norm error trotter}
    \norm{U(\delta) - \trot_{p}(\delta)}_{\op} \leq \alpha_{p}(\{H_l\})\delta^{p+1} = \bigO((\lambda \delta)^{p+1}).
\end{align}
Here $\alpha_{p}(\{H_l\})$ is an expression in terms of operator norms of nested commutators of the terms $H_l$ in the Hamiltonian, and $\lambda$ is the weight defined in \cref{eq:lambda}.
One can then use these to bound
\begin{align}
    \abs{E_0 - \tilde{E}_0} = \bigO\mleft( \delta^{-1} \norm{U(\delta) - \trot_{p}(\delta)}_{\op}\mright) = \bigO(\lambda^{p+1} \delta^{p}).
\end{align}
The commutator bounds are often quite sharp \cite{childs2018toward,childsTheoryTrotterError2021}; a disadvantage of such bounds is that the number of commutator terms scales as $\bigO(L^{p+1})$ and may be practically infeasible to compute even for modest system sizes \cite{poulin2014trotter}.

In \cref{sec:electronic structure main text}, we numerically estimate the Trotter error constant
\begin{align}
C_{\text{gs}} := \Delta E \delta^{-p-1},
\label{eq:Cgs-def}
\end{align}
from the ground-state energy error $\Delta E$, and see that it is typically relatively small.
This is an important aspect of product formula algorithms for phase estimation, since it means that Trotter methods may in practice use larger step sizes $\delta$ than one can rigorously justify and can be implemented using fewer resources than rigorous upper bounds suggest.
Because of the relevance of the Trotter error for Hamiltonian simulation, this has been studied extensively in previous work; see \cite{poulin2014trotter,wecker2014gate,babbush2015chemical,reiher2017elucidating,kivlichan2020improved} for empirical and rigorous Trotter error bounds for phase estimation for quantum chemistry.
Additionally, there is theoretical evidence that for most states the Trotter error is significantly lower than the worst-case bounds \cite{sahinouglu2021hamiltonian,chen2024average,zhao2022hamiltonian}.

\section{Phase estimation with random compilation for time evolution}\label{sec:random product}
Two major downsides to Trotter product formulas are their scaling with $L$, and a difficult to rigorously control Trotter error.
One can avoid both these features using \emph{randomized product formulas} \cite{campbell2019random,wan2022randomized,chen2021concentration,kiss2023importance,granet2024hamiltonian,kiumi2024te}.
Here we discuss two randomized methods of implementing an approximation to the time evolution $U(t)$:
qDRIFT \cite{campbell2019random} and a randomized Taylor series expansion (RTE) \cite{wan2022randomized}.

We now assume that the Hamiltonian consists of a linear combination of Pauli operators $P_l$, since for these randomized methods we will use that $P_l^2 = \id$. We will denote by
\begin{align}
    V_l(\phi) = e^{-i \phi P_l}
\end{align}
a rotation along the Pauli operator $P_l$, and refer to these as Pauli rotations.
Absorbing signs into the Pauli operators (i.e. replacing $P_l \mapsto \mathrm{sgn}(h_l) P_l$), we may write the Hamiltonian as
\begin{align}
    H = \lambda \sum_{l=1}^L p_l P_l, \quad p_l = \frac{\abs{h_l}}{\lambda}
\end{align}
where the numbers $p_l$ form a probability distribution $p$.
For convenience of notation we will now assume that $\lambda = 1$; it can be reinstated by rescaling the time variable by $\lambda$.
To approximate time evolution for time $t$, we again divide $t$ into $r$ steps, with a step size $\tau$ so that $\tau r = t$.
The qDRIFT approximation to time evolution is given by sampling random variables $\rv{l_1}, \dots, \rv{l_r}$ according to $p$, and given outcomes $l_1, \dots, l_r$, apply the unitary
\begin{align}
    V_{l_r}(\varphi) \cdots V_{l_1}(\varphi), \quad \text{where} \, \, \varphi = \arctan(\tau).
\end{align}
That is, we apply a sequence of equiangular Pauli rotations, and which ones we apply is sampled proportionally to the weight of that Pauli operator in the Hamiltonian.
If $\mathcal S_{\tau}$ denotes the quantum channel which with probability $p_l$ applies $V_l(\varphi)$, the above procedure leads to the composition $\mathcal S_{\tau}^r$.

This method was proposed in \cite{campbell2019random}.
In order to measure the accuracy of the approximation of the exact time evolution, we have to compare the exact time evolution unitary with the quantum channel $\mathcal S_{\tau}$.
The natural measure to compare two quantum channels is the diamond norm.
We let $\mathcal U_t$ denotes the quantum channel corresponding to applying the time evolution unitary $U(t)$ mapping $\rho \mapsto U(t)^\dagger \rho U(t)$. It was shown in \cite{campbell2019random} that if we break the total time $t$ into $r$ shorter steps of time $\tau$, so $\tau r = t$, then we have the diamond norm bound
\begin{align}
    \norm{\mathcal U_t - \mathcal S_{\tau}^r}_{\diamond} \leq 2r\tau^2 \exp(2\tau) \approx 2 t^2 / r.
\end{align}
This means that choosing a number of steps $r \geq 2t^2 \eps^{-1}$, or reinstating the dependence on $\lambda$, $r \geq 2\lambda^2 t^2 \eps^{-1}$, suffices to approximate time evolution for time $t$ to error $\eps$ in diamond norm.
In particular, there is no dependence on $L$, but only on $\lambda$. This is potentially beneficial if $L$ is large and there is a large number of terms with small coefficients in the Hamiltonian.

This can then be used to simulate time evolution in quantum phase estimation.
Since phase estimation requires time evolution for time $t = \Theta(\eps^{-1})$, this would lead to a scaling of $r = \bigO(\eps^{-2})$ Pauli rotations.
However, directly applying the diamond norm bounds with qDRIFT in the standard QPE algorithm leads to suboptimal results \cite{campbell2019random,lee2021even}.
In general, one can ask whether using qDRIFT leads to more efficient phase estimation procedures than using Trotterization.
There is no unambiguous answer: the randomized method scales worse with $\eps$, and the relation between $L$ and $\lambda$ depends on the Hamiltonian.
An advantage for qDRIFT is most likely if $L$ is large, which is the case for molecular systems using atomic orbitals.
Here, \cite{campbell2019random} found an advantage, but using very loose upper bounds for the Trotter error. In \cite{lee2021even}, for FeMoco it was found that qDRIFT was much less efficient than the Trotter estimate with optimistic Trotter error estimates from \cite{reiher2017elucidating} (see \cref{fig:femoco intro}).

\begin{figure*}[t]
    \begin{quantikz}
        \lstick{$\ket{0}$} & \gate{H}\slice[style={color=gray}]{Sample $l_1 \sim p$} & \ctrl{1} & \qw \ \ldots\ & \qw\slice[style={color=gray}]{Sample $l_r \sim p$}& \ctrl{1} & \qw & \gate{R(\theta)} & \gate{H} & \meter{}\\
        \lstick{$\ket{\psi}$} & \qw & \gate{V_{l_1}(\varphi)} & \qw \ \dots\ & \qw & \gate{V_{l_r}(\varphi)} & \qw & \qw & \qw & \qw
    \end{quantikz}\\
    \vspace{0.1cm}
    \begin{quantikz}
        \lstick{$\ket{0}$} & \gate{H}\slice[style={color=gray}]{sample $n=0$\\ $l_1 \sim p$} & \ctrl{1} & \qw\slice[style={color=gray}]{sample $n=2$\\ $l_k, l_k' \sim p$} \ \ldots\ & \qw & \ctrl{1} & \ctrl{1} & \qw \ \ldots\ & \qw\slice[style={color=gray}]{sample $n=0$\\ $l_r \sim p$} & \ctrl{1} & \qw & \gate{R(\theta)} & \gate{H} & \meter{}\\
        \lstick{$\ket{\psi}$} & \qw & \gate{V_{l_1}(\varphi)} & \qw \ \dots\ & \qw & \gate{V_{l_k}(\varphi_2)} & \gate{P_{l_{k}'}} & \qw \ \dots\ & \qw & \gate{V_{l_r}(\varphi)} & \qw & \qw & \qw & \qw
    \end{quantikz}

    \caption{Hadamard test for qDRIFT and for RTE. For qDRIFT, $\tau = t/r$ and $\varphi = \arctan(\tau)$. RTE is similar, but now for each gate one samples an (even) order $n$ of the Taylor expansion. If $n=0$, one samples $l$ according to $p$ as for qDRIFT. If $n$ is larger (which happens with probability $\bigO(\tau^2)$), the rotation is over a different angle $\varphi_n$ and there is an additional controlled Pauli operator in the circuit.}
    \label{fig:qdrift}
\end{figure*}

We will now argue that previous work has overestimated the cost of doing phase estimation using qDRIFT.
A key reason is that using an analysis based on diamond norm bounds is suboptimal, as also observed in \cite{wan2022randomized}.
Similar to the discussion of Trotter error, what we care about eventually is not an accurate simulation of the full time evolution (corresponding to small diamond norm), but rather the accuracy of the ground state energy as extracted from the phase estimation procedure.
In particular, for phase estimation methods based on the Hadamard test, what is relevant is the expectation value of the outcome of the Hadamard test.
If we apply the Hadamard test to a quantum channel which applies unitary $U_m$ with probability $p_m$ to initial state $\ket{\psi}$, the resulting random variable $\rv{Z}$ has expectation value
\begin{align}
    \EE \rv{Z} = \sum_m p_m \bra{\psi} U_m \ket{\psi}.
\end{align}
The expectation value is over two sources of randomness: the sampling of the unitary $U_m$ and the randomness from the measurement at the end of the circuit.
This means that what matters for the Hadamard test is the linear combination of unitaries $\sum_k p_m U_m$ rather than the quantum channel $\rho \mapsto \sum_m p_m U_m \rho U_m^\dagger$.
Let $\rv{Z}(\tau, r)$ be the random variable obtained from applying a Hadamard test to $r$ rounds of qDRIFT with interval $\tau$.
Consider the Taylor expansion of a short time evolution:
\begin{align}\label{eq:qdrift lcu}
    \begin{split}
        e^{-i \tau H} & = \id - i\tau \sum_l p_l P_l + \bigO(\tau^2)             \\
                      & = \sum_l p_l(\id - i\tau P_l) + \bigO(\tau^2)            \\
                      & = \sqrt{1 + \tau^2} \sum_l p_l V_l(\phi) + \bigO(\tau^2)
    \end{split}
\end{align}
where we have used that $\id - i \tau P_l = \sqrt{1 + \tau^2} V_l(\phi)$ for $\phi = \arctan(\tau)$.
Alternatively, we can use
\begin{align}
    \sum_l p_l V_l(\varphi) = \frac{\id - i\tau H}{\sqrt{1 + \tau^2}}.
\end{align}
Using this, one can show (see \cref{sec:randomized hadamard} for details) that, when using an initial state $\ket{\psi}$ with squared overlap $c_k$ with the energy $E_k$ eigenstate, the expectation value of $\rv{Z}(\tau, r)$ has the exact expression
\begin{align}\label{eq:time signal qDRIFT}
    g(\tau,r) := \EE \rv{Z}(\tau, r) = \sum_k \frac{c_k}{(1 + \tau^2)^{r/2}} (\id - i\tau E_k)^r.
\end{align}
If we choose a sufficiently small step size $\tau$, we get
\begin{align}
    g(\tau,r) \approx \sum_k c_k e^{-\tau^2 r (1 - E_k^2)/2} e^{-i \tau r E_k}.
\end{align}
This expression is similar to $g(t)$ (the exact time signal of \cref{eq:signal hadamard test}) if we take $t = \tau r$, with the difference that there is a \emph{damping factor} $\exp(-\tau^2r(1 - E_k^2)/2)$.
So, if we want to deduce information about the phases $E_k$, and we do not want to increase the number of required samples, we need to choose $\tau$ and $r$ such that $\tau^2 r$ is $\bigO(1)$.
This leads to a scaling of $r = \Omega(t^2)$.
One can reduce this value of $r$, but this comes at an (exponential) sampling overhead.
The expectation value in \cref{eq:time signal qDRIFT} is an \emph{exact} expression, so it is more amenable to analysis than the result of deterministic product formulas, where one has to bound the Trotter error.
The original proposal in \cite{campbell2019random} uses an angle $\varphi = \tau$; this only changes the expectation values to order $\bigO(\tau^2)$ and does not change the error analysis; here we choose $\varphi = \arctan(\tau)$ for consistency with the Taylor expansion.

An alternative random product formula for Hamiltonian simulation was introduced in \cite{wan2022randomized}, which we call a \emph{randomized Taylor expansion} (RTE).
Whereas qDRIFT, as per \cref{eq:qdrift lcu}, can be seen as a first-order expansion of $\exp(-iH\tau)$ as an LCU, one can also take the full Taylor expansion and normalize the coefficients to a probability distribution.
This gives an LCU
\begin{align}\label{eq:rte lcu}
    e^{-i\tau H} & = B(\tau) \sum_m b_m U_m
\end{align}
where the $b_m$ form a probability distribution, and the $U_m$ are unitaries which are a composition of one Pauli rotation and one Pauli operator and $B(\tau)$ is a normalization factor; see \cref{sec:randomized hadamard} for details and a derivation.
Up to normalization, qDRIFT precisely corresponds to only keeping the lowest order term. The probability of sampling a higher order term is $\bigO(\tau^2)$.
The normalization factor $B(\tau)$ is approximately $\exp(\tau^2)$ for small $\tau$ and can be efficiently computed.
Taking $r$ rounds thus gives a LCU
\begin{align}\label{eq:rte lcu repeated}
    e^{-i\tau r} = B(\tau)^r \sum_{m_1,\dots,m_r} \underbrace{b_{m_1}\dots b_{m_r}}_{q_{\vec k}} \underbrace{U_{m_1} \cdots U_{m_r}}_{W_{\vec k}}
\end{align}
where the $q_{\vec k}$ are a probability distribution over $\vec k = (m_1,\dots,m_r)$, and the $W_{\vec k}$ are unitaries consisting of $r$ Pauli rotations interspersed with Pauli operators.
Performing a Hadamard test leads to a signal
\begin{align}
    \EE \rv{Z}(\tau,r) = \frac{1}{B(\tau)^r} \sum_k c_k e^{-i E_k \tau r}.
\end{align}
In particular, if we want to estimate $g(t)$, one can do so using this random variable for $\tau r = t$.
The factor $B(\tau)^r \leq \exp(\tau^2 r)$ determines the required number of samples. To keep this constant, one should have $\tau^2 r = \bigO(1)$ and hence $r = \Omega(t^2)$, as for qDRIFT.
The advantage of this approach is that, opposed to qDRIFT, it gives an \emph{unbiased} estimator for the signal $g(t)$.

% A detailed derivation can be found in \cref{sec:randomized hadamard}.

Since we have exact expressions for the expectation value of Hadamard tests for both qDRIFT and RTE, we can now use these to analyze single-ancilla QPE schemes.
We do so for robust phase estimation.
The following result is stated and proven as \cref{thm:rpe} in \cref{sec:rpe} and follows straightforwardly from the analysis of robust phase estimation in \cite{ni2023low,kimmel2015robust,belliardo2020achieving}. In contrast to the approach in \cite{campbell2019random,lee2021even}, there is no dependence on a failure probability of the algorithm.
We consider the case with large ground state overlap. If $c_0 \geq \eta = 1 - \alpha$, we can choose a parameter $\xi = \Omega(\alpha)$ to reduce the maximal circuit size.

\begin{thm}\label{thm:rpe qDRIFT main text}
    Given a guiding state $\ket{\psi}$ with ground state overlap lower bounded by $\eta \geq 4 - 2\sqrt{3} \approx 0.54$ and $\xi = \Omega(\alpha)$ for $\eta = 1 - \alpha$, robust phase estimation using random product formulas gives an estimate of $E_0$ with root mean square error $\eps$ using circuits with at most
    \begin{align}
        r_{\max} = \bigO\mleft(\frac{\xi^2 \lambda^2}{\eps^2}\mright)
    \end{align}
    Pauli rotations per circuit and a total of
    \begin{align}
        r_{\tot} = \bigO\mleft(\frac{\lambda^2}{\eps^2}\mright)
    \end{align}
    Pauli rotations.
\end{thm}
We also numerically estimate the constant prefactors, and for $\eta = 1, \xi = 1$ we find $r_{\tot} \leq 8 \lambda^2 \eps^{-2}$ using qDRIFT.
For $\eta = 1$, since we lose Heisenberg scaling, if one does not count the cost of preparing the exact ground state one might as well just measure the ground state energy by using $\bigO(\lambda^2 \eps^{-2})$ copies of the ground state.
However, the robust phase estimation procedure uses only $\bigO(\log(\lambda \eps^{-1})^2)$ copies of the ground state (when choosing $\xi = 1$), and crucially it still works for ground state overlap $\eta < 1$ sufficiently large.
Note additionally that the \emph{total} cost does not depend on the trade-off parameter $\xi$, in contrast to the trade-off in \cref{eq:rpe max time reduction}.
However, choosing small $\xi$ does come at a cost of a higher number of circuits, and in particular of the number of initial state preparation circuits.
For randomized methods it is advantageous to choose $\xi \sim 1 - \eta $ as small as possible (if the state preparation cost is not too high), since the total cost of the Hamiltonian simulation does not increase under the depth reduction.
This is in contrast to the exact time evolution model, where a reduction in maximal evolution time comes at the cost of an increasing total evolution time, as in \cref{eq:rpe max time reduction}, due to the larger number of circuit repetitions.
The suboptimal scaling of $\bigO(\eps^{-2})$ of randomized methods is a significant disadvantage.
However, as we see in molecular electronic structure examples, it can still outperform deterministic product formulas.
This advantage increases further when one uses depth reductions for states with high ground state overlap.

\section{Partially Randomized Product Formulas}\label{sec:composite}
We have seen that randomization can be useful for Hamiltonian simulation in the context of phase estimation.
However, it has the clear disadvantage of scaling quadratically with $\lambda$ and $\eps^{-1}$, whereas deterministic product formulas may have a better dependence on $\eps^{-1}$.
Deterministic product formulas are disadvantageous if $L$, the number of terms in the Hamiltonian, is large.
Random product formulas are helpful if there is a tail consisting of many terms in the Hamiltonian of small weight.

A natural proposal is to treat a subset of terms in a random fashion, and a subset of terms deterministically.
There exist various such proposals, with different ways of implementing the partial randomization \cite{ouyang2020compilation,rajput2022hybridized,jin2023partially,hagan2023composite}.
Here, we modify the scheme of \cite{hagan2023composite} to a version which has the advantage of being particularly easy to analyze for the purpose of single-ancilla phase estimation.
Again, the difference in the analysis is that we do not require a diamond norm bound on the resulting quantum channel, but only need to evaluate the expectation value of the resulting Hadamard test.

The partial randomization is based on a decomposition
\begin{align}\label{eq:partial random decomposition}
    H = \underbrace{\sum_{l = 1}^{L_D} H_l}_{= H_D} + \underbrace{\sum_{m=1}^{M} h_{m} P_{m}}_{= H_R}
\end{align}
where we assume that the $P_m$ square to identity and where $H_D$ contains the terms we treat deterministically, and $H_R$ the terms which are treated by a random product formula.
The idea is to use a deterministic Trotter formula to the decomposition of $H$ into the terms $H_l$ and $H_R$, and to apply a random product formula to the decomposition of $H_R$ into terms $h_m P_m$.
In order for this to yield an improvement, one aims for a decomposition such that
\begin{align}\label{eq:decompose partial random}
    \lambda_R = \sum_m \abs{h_{m}} \ll \lambda
    \qquad \text{and}\qquad
    L_D \ll M.
\end{align}

\begin{figure*}[t]
    \begin{quantikz}
        \lstick{$\ket{0}$} & \gate{H} & \ctrl{1}\gategroup[2,steps=3,style={dashed, color=gray, inner sep=2pt}]{Equal to $\trot_2(\delta)$ in expectation value} & \ctrl{1} & \ctrl{1} & \ctrl{1}\gategroup[2,steps=3,style={dashed, color=gray, inner sep=2pt}]{\small{(Repeat)}} & \ctrl{1} & \ctrl{1} & \qw \ \ldots\ \ & \gate{H} & \meter{}\\
        \lstick{$\ket{\psi}$} & \qw & \gate{U_{1 \rightarrow L_D}} & \gate{W_{k_1}} & \gate{U_{L_D \rightarrow 1}} & \gate{U_{1 \rightarrow L_D}} & \gate{W_{k_2}} & \gate{U_{L_D \rightarrow 1}} & \qw \ \ldots\ \ & \qw & \qw
    \end{quantikz}
    \vspace{0.1cm}
    \begin{flushleft}
        where \\
        \begin{quantikz}  & \gate{\;\;W_k\,\;} & \qw & \end{quantikz} = \begin{quantikz} \qw & \gate{V_{m_1}(\varphi)\,} & \qw \ \dots\ & \gate{V_{m_r}(\varphi)} & \qw \qw & \end{quantikz} \qquad $\vec k = (m_1, \dots, m_r)$, $m_i$ sampled according to \cref{eq:rte lcu repeated}
        \begin{quantikz}  & \gate{U_{1 \rightarrow L_D}} & \qw & \\ & \gate{U_{L_D \rightarrow 1}} & \qw & \end{quantikz} = \begin{quantikz}  & \gate{e^{-\frac{i}{2} \delta H_{1}}} & \qw \ \dots\ & \gate{e^{-\frac{i}{2} \delta H_{L_D}}} &  \qw & \rstick[2,brackets=none]{Evolution of terms in $H_D$.}
            \\  & \gate{e^{-\frac{i}{2} \delta H_{L_D}}} & \qw \ \dots\ & \gate{e^{-\frac{i}{2} \delta H_{1}}} & \qw & \end{quantikz} \\
    \end{flushleft}
    \vspace{0.1cm}
    \caption{Hadamard test using partially randomized product formula. Here, we illustrate the second order deterministic method, based on \cref{eq:second order partial random}.}
    \label{fig:partially randomized circuit}
\end{figure*}

To analyze how such methods will perform, we now briefly summarize two lessons from the previous two sections. Deterministic product formulas lead to a \emph{bias} in the ground state energy $E_0$ due to the Trotter error. A small step size is required in order to keep this bias below the desired precision. For randomized product formulas, one can make the resulting signal such that the ground state energy is unbiased, but the \emph{amplitude} of the signal is damped and one incurs a sampling overhead. Here, a sufficiently small step size is required in order to keep this sampling overhead of constant size.
Partially randomized product formulas will have both a bias (from the deterministic decomposition) and a damping (from implementing the randomized part of the Hamiltonian) in the resulting signal. The error analysis straightforwardly combines these elements.

Let $\trot_{p}(\delta)$ be a deterministic product formula with respect to the decomposition of $H$ into the terms $H_l$ and $H_R$.
This means that $\trot_{p}(\delta)$ is a product of unitaries of the form $e^{-i \gamma H_l}$ (which we assume to have some decomposition into elementary gates) and terms of the form $e^{-i \gamma H_R}$ for some $\gamma$.
As in \cref{eq:ground state error trotter}, this gives an effective Hamiltonian with ground state $\tilde{E}_0$ with error $\abs{E_0 - \tilde E_0} \leq C_{\gs} \delta^{p}$. Here, the Trotter constant $C_{\gs} = C_{\gs}(p, \{H_l\}_{l=1}^{L_D}, H_R)$ depends on the order $p$ as well as on the Hamiltonian, decomposed into terms $H_l$ for $l = 1, \dots, L_D$ and treating $H_R$ as a single term.
Each of the $e^{-i \gamma H_R}$ appearing in the formula we will then simulate using RTE, using $r$ steps per Trotter step.

To make this concrete, we take the second order Suzuki-Trotter formula, in which case we get $\trot_{2}(\delta)$ given by
\begin{align}\label{eq:second order partial random}
    e^{- \frac12 i \delta H_1} \cdots e^{- \frac12 i \delta H_{L_D}} e^{-\delta i H_R} e^{- \frac12 i\delta H_{L_D}} \cdots e^{- \frac12 i \delta H_1}.
\end{align}
If we normalize $H_R / \lambda_R$, we need to do time evolution for time $\delta \lambda_R$, which we break up into $r$ steps.
Using RTE as in \cref{eq:rte lcu repeated} we expand
\begin{align}
    e^{-i\delta H_R} & = (e^{-i\tau H_R/\lambda_R})^{r} \quad \text{for } \tau r = \delta \lambda_R \\
                     & = B \sum_{k} q_k W_k
\end{align}
as a LCU, which gives a LCU for $\trot_{2}(\delta)$.
Here, each $W_k$ is a composition of $r$ Pauli rotations and Pauli operators.

% This leads to the following result, see \cref{lem:composite simulation} for a more detailed statement and derivation.

% \begin{lem}\label{lem:composite simulation main text}
%     For $H = H_D + H_R$ as in \cref{eq:partial random decomposition}, the Hadamard test shown in \cref{fig:partially randomized circuit} has measurement outcomes $\rv{Z}(\tau, s)$ with expectation value
%     \begin{align}
%         \EE \, B \rv{Z}(\tau,s) =  \bra{\psi} \trot_{p}(\tau)^s \ket{\psi}
%     \end{align}
%     where $B = \bigO(1)$, and the circuit consists of at most $PLs$ time evolution unitaries along one of the $H_l$, and $r = \bigO(\lambda_R^2 \tau^2 s^2)$ time evolution unitaries along one of the $P_m$.
% \end{lem}

Since this gives an unbiased LCU for the exact time evolution, it is straightforward to show (see \cref{lem:composite simulation}) that using RTE and repeating $s$ times in a Hadamard test has expectation value
\begin{align}
    \EE \rv{Z}(\delta, r, s) = \frac{1}{B(r, \delta)^s} \bra{\psi} \trot_{p}(\delta)^s \ket{\psi}
\end{align}
where
\begin{align}
    B(r, \delta) \leq \exp\mleft(\bigO\mleft(\frac{\lambda_R^2 \delta^2 s}{r}\mright)\mright).
\end{align}

In order to perform phase estimation, we need to choose $\delta$ such that the Trotter error of $\trot_{p}(\delta)$ is smaller than $\eps$.
That means we should choose a Trotter step size $\delta = \bigO((C_{\gs}\eps)^{\frac{1}{p}})$ and maximal depth $s = \bigO(C_{\gs}^{\frac{1}{p}}\eps^{-1 - \frac{1}{p}})$ for phase estimation.
In order to keep the number of samples constantly that are required to estimate $\bra{\psi} \trot_{p}(\delta)^s \ket{\psi}$ to constant precision, we need
\begin{align}
    B(r,\delta)^s \leq \exp\left(\bigO\left(\frac{\lambda_R^2 \delta^2 s}{r}\right)\right)
\end{align}
to be constant.
Since $\delta s = \bigO(\eps^{-1})$, this can be achieved by taking the total number of rotations in the randomized part as $rs = \bigO(\lambda_R^2 \eps^{-2})$.
Together, this leads to the following result.

\begin{thm}\label{thm:partially randomized main text}
    Let $H$ as in \cref{eq:decompose partial random}, and suppose that the Trotter error constant for $H$ (written as a sum of the $H_l$ and treating $H_R$ as a single term) is $C_{\gs}$, then phase estimation to precision $\eps$ using the partially randomized product formula requires $\bigO(L_D C_{\gs}^{\frac{1}{p}} \eps^{-1 - \frac{1}{p}})$ evolutions along the $H_l$ and $\bigO(\lambda_R^2 \eps^{-2})$ Pauli rotations.
\end{thm}

Gate counts with constant factors are given in \cref{sec:gate count}.
A subtlety here is that the Trotter constant $C_{\gs}$ depends on the choice of partitioning. This complicates the task of finding an optimal distribution of terms to $H_R$ or $H_D$.
One can bound $C_{\gs}$ using a commutator bound in the usual way (see \cite{hagan2023composite} for such a bound).
For our numerical estimates we will take a different approach: we estimate the Trotter error for the fully deterministic method, and use the resulting Trotter error constant. In that case, it is clear that to minimize the cost, the optimal decomposition assigns the $L_D$ largest weight terms to $H_D$ for some choice of $L_D$, and the remainder to $H_R$.
Suppose that $H$ has $L$ terms and weight $\lambda$.
Clearly, partial randomization is most useful if there is a small number $L_D \ll L$ of terms with large weight (so including these in $H_D$ means that $\lambda_R \ll \lambda$); see \cref{fig:weight distribution} for an example of a distribution of weights for a quantum chemistry Hamiltonian.

\section{Electronic structure Hamiltonians}\label{sec:electronic structure main text}

We start with a molecular electronic structure Hamiltonian in second quantization:
\begin{align}\label{eq:electronic hamiltonian}
    \sum_{pq,\sigma} h_{pq} a_{p\sigma}^\dagger a_{q\sigma} + \frac12 \sum_{pqrs,\sigma\tau} h_{pqrs} a_{p\sigma}^\dagger a_{r\tau}^\dagger a_{s\tau} a_{q\sigma},
\end{align}
where $a_{p\sigma}^{\dagger}$, $a_{q\sigma}$ are fermionic creation and annihilation operators and $h_{pq}$, $h_{pqrs}$ are integrals arising from a calculation using the underlying orbital basis (see Appendix \ref{sec:electronic structure problem}).
The indices $p, q, r, s$ denote the spatial degrees of freedom, whereas $\sigma, \tau$ are reserved for spin.
The Hamiltonian acts on the space of $2N$ spin-orbitals, and we fix the number of electrons $n$.
And a choice of fermion-to-qubit mapping transforms it into a linear combination of Pauli operators, acting on $2N$ qubits with $L = \bigO(N^4)$ terms gives a representation of $H$ as a linear combination of Pauli operators
\begin{align}\label{eq:pauli electronic structure}
    H = \sum_l H_l = \sum_l h_l P_l.
\end{align}
We work with active spaces of molecular orbitals. An alternative choice is a plane wave basis, which has the advantage that the Hamiltonian has fewer terms and the Trotter error can be more easily controlled. This is particularly useful for materials \cite{kivlichan2020improved,childsTheoryTrotterError2021,su2021nearly,mcardle2022exploiting,low2023complexity}.
For this work we restrict to active spaces based on molecular orbitals since plane wave bases need a much larger number of orbitals to reach similar accuracy.
While in principle there are $\bigO(N^4)$ nonzero coefficients in \cref{eq:electronic hamiltonian}, many are extremely small and can be truncated. For large spatially extended systems, one expects that the number of significant terms mostly scales quadratically in $N$. Generally, one expects polynomial decay of the weight, with a tail of terms which decays exponentially as in \cref{fig:weight distribution}. See \cref{sec:electronic structure problem} for more discussion.

\begin{figure}
    \centering
    \includegraphics[width=\linewidth]{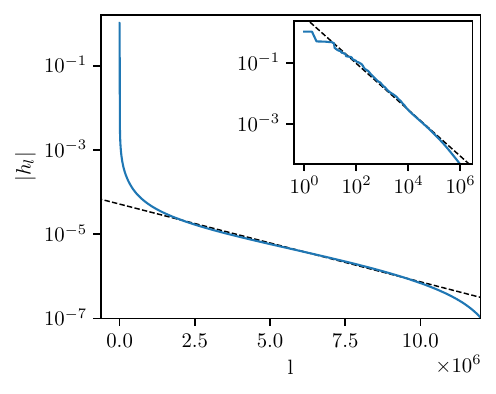}
    \includegraphics[width=\linewidth]{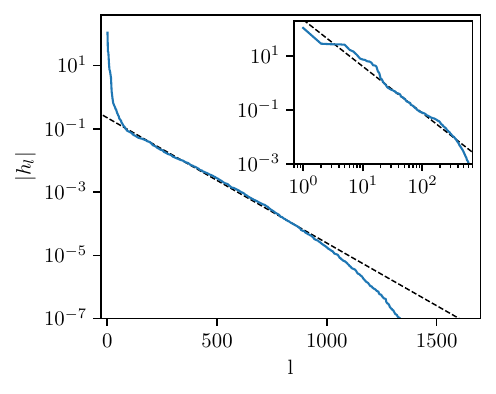}
    \caption{Distribution of the weights of terms of the FeMoco Hamiltonian in the Pauli decomposition (top) and the 1- and 2-electron terms in the double factorization representation (bottom).
        The dashed lines show an exponential fit to the tail (main plot) and a power law for the largest terms (the log-log inset plot).}
    \label{fig:weight distribution}
\end{figure}

As benchmark systems, we use a linear chain of hydrogen atoms, which is a toy model for a system of increasing size, and FeMoco, a standard challenge benchmark problem used for quantum computing resource estimates, as well as a collection of small molecules with varying active space sizes. Detailed descriptions can be found in \cref{sec:electronic structure problem}.

The computational cost of implementing product formulas for $H$ depends on various factors: for deterministic product formulas, it depends on $L$ and the Trotter error; for randomized product formulas it depends on the weight $\lambda$, and for partially randomized methods it additionally depends on a choice of partitioning into deterministic and randomized terms.
In order to assess the performance of phase estimation with partially randomized product formulas, in the following subsections we discuss numerical estimates of the Trotter error, and the decomposition into $H_D$ and $H_R$.

\subsection{Estimating Trotter error}
In order to give a good comparison between deterministic and (partially) randomized product formulas, we start by evaluating the Trotter error, which is crucial for the performance of the deterministic product formula.
We perform a detailed analysis of Trotter error with system size.

Apart from the practical difficulties in evaluating the error bounds based on commutators in \cref{eq:operator norm error trotter}, for the purpose of computing ground state energies bounding the ground state energy error in terms of the operator norm error can be a loose bound. This has been confirmed by numerical chemistry simulations for small molecules \cite{babbush2015chemical}. From a theoretical perspective, the operator norm of our matrices depends on how many orbitals we choose to include, which is a choice of the algorithm and not a fundamental physical feature of the problem; see \cite{burgarth2023strong} for more discussion of this.

We numerically estimate the Trotter error constant for ground state energy errors, and see that it is typically relative small, as shown in \cref{fig:trotter-error-proxy}.
See also \cite{poulin2014trotter,wecker2014gate,babbush2015chemical,reiher2017elucidating,kivlichan2020improved} for previous work on empirical and rigorous Trotter error bounds for phase estimation for quantum chemistry.
This is an important aspect of product formula algorithms for phase estimation, since it means that Trotter methods may in practice use larger step sizes $\delta$ than we can rigorously justify and can be implemented using fewer resources than rigorous upper bounds suggest.

\begin{figure}
    \centering
    \includegraphics[width=\linewidth]%{trotterconst_largedataset_gserror_trot2_anticommfit.pdf}
    {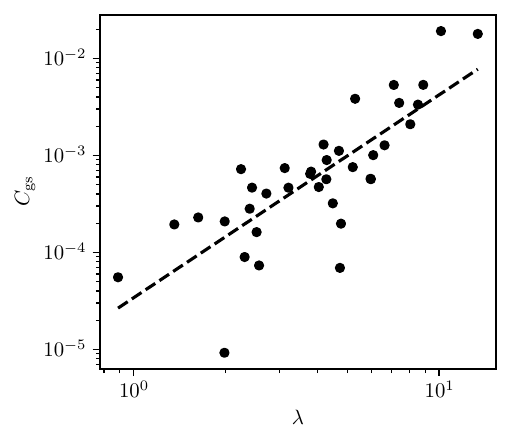}
    \caption{We calculate the Trotter error $\Delta E$ (on the ground state energy) and use this to estimate the prefactor $C_{\gs}$, following \cref{eq:Cgs-def}.  This is done for $p=2$ for a variety of molecules and active spaces, and plotted against $\lambda$.
    }
    \label{fig:trotter-error-proxy}
\end{figure}

As discussed in \cref{sec:trotter}, there are various error measures and bounds one can use. For our purposes, we care about the ground state energy error.
In order to estimate this error, we consider a benchmark set of small molecules with varying active space size and at half filling, so the number of electrons $n$ equals half the number of spin-orbitals $2N$.
We compute the Trotter error constant numerically for the second order method with $p=2$.
As we show in \cref{fig:trotter-error-proxy}, the Trotter error constant correlates well with $\lambda$.
Crucially, even though the Trotter error constant scales with system size, it is typically rather small even for larger systems.
These results corroborate similar findings in \cite{babbush2015chemical, reiher2017elucidating}.

Further numerical results are presented in \cref{sec:trotter appendix}.
We estimate Trotter error scalings for the hydrogen chain, which has a natural scaling limit.
In general we see that using the ground state energy error instead of an operator norm bound on the full unitary leads to about an order of magnitude improvement in the cost estimate.
Additionally, for our chemistry benchmark systems we see no structural improvement by going to higher order Trotter-Suzuki product formulas.

\subsection{Choice of decomposition for partial randomization in the Pauli representation}
For the partially randomized product formulas, one additionally needs to choose a partition of the Hamiltonian into a deterministic and randomized part as in \cref{eq:partial random decomposition}.
We first discuss the case where we consider $H$ in a Pauli representation, where the terms $H_D$ and $H_R$ are all Pauli operators (i.e. $H_l = P_l$ in \cref{eq:partial random decomposition}).
The total cost of the partially randomized product formula circuit is the sum of the deterministic and randomized parts as in \cref{thm:partially randomized main text}. Obviously, we want to choose a decomposition which minimizes this sum.

A subtlety is that the Trotter error of the scheme depends on the choice of decomposition, and which terms are treated randomly; i.e. $C_{\gs} = C_{\gs}(p, \{H_l\}, H_R)$.
A (rigorous) approach to this is to use commutator bounds to estimate the Trotter error and to take this into account in the optimization problem.
Such an approach is pursued in \cite{hagan2023composite} and applies straightforwardly to our method as well.
However, this leads to a complicated optimization problem. Additionally, for the type of problems we consider, the deterministic methods need to be taken with heuristic empirical error scaling in order to be competitive.
As a heuristic, we simply bound the Trotter error of the partially randomized method by our estimate of the Trotter error of the fully deterministic method. This is not a rigorous treatment but appears to be reasonable in practice.
We study the dependence of the Trotter error on the choice of decomposition numerically. The results can be found in \cref{sec:trotter appendix}; the conclusion is that when decreasing $L_D$, the Trotter error does not decrease significantly unless $\lambda_R$ is a significant fraction of $\lambda$.
This means that in practice the heuristic of using the Trotter error of the fully deterministic method should suffice; it is not valid when $\lambda_R$ is close to $\lambda$ but in this regime partially randomized product formulas will anyways not improve much over fully randomized product formulas.

Once one assumes that the Trotter error constant $C_{\gs}$ may be taken constant when finding the optimal decomposition, it is clear that the optimal decomposition should simply include the $L_D$ highest weight terms in $H_D$, and the remainder in $H_R$, for some choice of $L_D$.
One can then easily numerically search for the optimal value of $L_D$ which minimizes the cost in \cref{thm:partially randomized main text}.
In \cref{fig:hchain cost} we compare the gate counts for phase estimation of the hydrogen chain, comparing the deterministic, randomized (qDRIFT) and partially randomized methods.
In terms of cost metric, we consider the number of non-Clifford gates (expressed as number of Toffoli gates) after compilation of the Pauli rotations. We also count the number of 2-qubit gates, see \cref{sec:gate count}.
In both cases, the partially randomized method achieves a reduction of cost, while scaling similar to the randomized method with increasing system size.
We note that the optimal choice of decomposition depends on the cost metric employed (e.g. whether we measure two-qubit gates or non-Clifford gates).

Another interesting use for partial randomization is that one can improve the compilation to elementary gates. The randomized part consists of many rotations over the same angle, which is beneficial when compiling to Toffoli gates \cite{gidney2024magic,kivlichan2020improved}. For the deterministic part $H_D = \sum_l h_l P_l$, one can round the values of the $h_l$ to values more favorable for compilation by transferring a small part of the weight to the randomized part, see \cref{sec:gate count}. This is similar to previously known randomization strategies for compilation \cite{campbell2017shorter,kliuchnikov2023shorter,koczor2024probabilistic}.

Finally, one could ask to what extent the terms in $H_R$ contribute to the value of the ground state energy. In principle, it could be that one can deterministically truncate many terms in the Hamiltonian without perturbing the ground state energy by much.
We study this numerically in \cref{sec:trotter appendix}, and find that at the partition with optimal cost for partial randomization, the error of fully truncating $H_R$ is generally too large.

\section{Different representations of the Hamiltonian}\label{sec:optimize hamiltonian}
The representation in \cref{eq:pauli electronic structure} is not unique.
A significant degree of freedom is the choice of orbital basis.
In \cref{sec:weight reduction} we discuss how this degree of freedom can be optimized to improve the performance of randomized product formulas.
Moreover, for deterministic product formulas it can be beneficial to use different decompositions of the Hamiltonian based on a factorization of the Coulomb tensor. 
In \cref{sec:factorization partial random} we argue that this approach benefits from partial randomization as well.

\subsection{Weight reduction for the Hamiltonian}\label{sec:weight reduction}
The performance of randomized product formulas depends on the value of the weight $\lambda$.
Therefore, minimizing the value of $\lambda$ is important for (partially) randomized product formulas. It is in fact relevant for all methods based on LCU, such as qubitization.
For this reason, it has been investigated extensively \cite{koridonOrbitalTransformationsReduce2021,loaiza2023reducing}.
We employ two techniques to significantly reduce the value of $\lambda$.

The first is that one can choose a different basis for the single-particle Hilbert space (i.e. a unitary transformation of the orbitals) \cite{koridonOrbitalTransformationsReduce2021}.
For choosing a good basis of orbitals, a good starting point are localized orbitals, rather than the canonical orbitals derived from a Hartree-Fock calculation.
While this choice already significantly reduces the value of $\lambda$, one can also use these as a starting point for further minimization \cite{koridonOrbitalTransformationsReduce2021}.
One can directly optimize $\lambda$ over parameterized orbital basis transformations via gradient descent.
Since this a high-dimensional optimization problem it is crucial to start from a good initial guess.
Typically, using localized orbitals rather than the canonical orbitals derived from a Hartree-Fock calculation reduces the value of $\lambda$.
In case orbital localization does not converge or does not yield a lower $\lambda$ value, we find that the orbital rotation into the eigenbasis of the largest Cholesky vector of the Coulomb-tensor decomposition to be an inexpensive alternative.
Besides the choice of orbitals, the Hamiltonian has a particle-number symmetry which can be exploited, and this leads to a second technique to find a different Hamiltonian with the same ground state energy, but with lower weight \cite{loaiza2023block}.

Applying these methods to FeMoco for the Reiher \cite{reiher2017elucidating} and Li \cite{li2019electronic} choice of active space, we see a reduction of the value of $\lambda$ from 1863 to 405 for the Reiher Hamiltonian and a reduction from 1511 to 568 for the Li Hamiltonian.
We note that this also leads to a significant improvement of the performance of the sparse qubitization method of \cite{berry2019qubitization} to FeMoco.
This leads to the results reported in \cref{fig:femoco intro}; see \cref{sec:gate count} for details.

When we discussed the choice of decomposition for the partially randomized method, we assumed a fixed choice of orbitals and hence of the coefficients in the Pauli representation of the Hamiltonian.
A reasonable choice is to take an orbital basis for which $\lambda$ is small.
However, a more complicated optimization is possible: one would want to find an orbital basis where most of the weight is concentrated in a small number of terms.
We leave such optimization to future work.

\subsection{Partial randomization with factorizations of the Hamiltonian}\label{sec:factorization partial random}
For deterministic product formulas, one does not need to restrict to the Pauli representation of the Hamiltonian.
One general approach is based on factorizations of the Hamiltonian, where one can factorize the Coulomb tensor $h_{pqrs}$. This gives a representation of the 2-body terms of the Hamiltonian as
\begin{align}\label{eq:factorized hamiltonian}
    \sum_{l=1}^{L} H_l = \sum_{l=1}^{L} (U^{(l)})^\dagger \underbrace{\Big(\sum_{p=1, \sigma }^{\rho_l} \lambda_p^{(l)} n_{p\sigma} \Big)^2}_{V^{(l)}} U^{(l)}
\end{align}
where for each $l$ we have a different orbital basis, $U^{(l)}$ is the unitary implementing the change of orbital basis. See \cref{sec:electronic structure problem} for details.
The $n_{p\sigma} = a_{p\sigma}^\dagger a_{p\sigma}$ are the number operators. After a Jordan-Wigner transformation, the Hamiltonian $V^{(l)}$ is diagonal in the standard basis.
This can be used for a deterministic product formula \cite{motta2021low}, since one can implement $\exp(-i \delta H_l)$ by
\begin{align}
    \exp(-i \delta H_l) = ( U^{(l)})^\dagger \exp(-i V^{(l)}) U^{(l)}.
\end{align}
Both the orbital basis changes and the diagonal Hamiltonian evolutions can be performed using $\bigO(N^2)$ gates using Givens rotations \cite{kivlichanQuantumSimulationElectronic2018}, so in total a single Trotter step needs $\bigO( L N^2)$ gates (both in terms of number of two-qubit gates and single-qubit rotations).

A decomposition of the form in \cref{eq:factorized hamiltonian} can be obtained through a \emph{double factorization}, which we review in \cref{sec:electronic structure problem}.
Here, the first factorization is into a sum of $L$ terms, the second factorization is the diagonalization of these in different orbital bases.
This can be an improvement over the Pauli representation if the ranks $L$ and $\rho_l$ are small.
Decompositions of the Hamiltonian as in \cref{eq:factorized hamiltonian}, or similar decompositions are also essential to qubitization and other LCU based approaches to Hamiltonian simulation \cite{berry2019qubitization,von2021quantum,lee2021even}.
Typically, in factorized representations, one truncates a tail of terms which are so small that they do significantly affect the ground state energy.
Partial randomization allows for much more flexible truncation schemes, and in that way can be used to improve product formulas based on Hamiltonian factorizations.

First of all, in \cref{eq:factorized hamiltonian}, some of the terms $H_l$ are rather small (but not so small that they can be completely ignored).
As an example, the distribution of the sizes of the $H_l$ for a double factorization of the FeMoco Hamiltonian is shown in \cref{fig:weight distribution}. This shows, similar to the Pauli representation, a dominant set of terms obeying powerlaw decay, and a tail of terms with exponential decay.
See \cite{motta2019efficient} for a discussion of the scaling of $\rho_l$ for spatially extended systems.
Note that the number of terms is much smaller than in the Pauli representation (but the cost of implementing any of the $\exp(-i \delta H_l)$ is larger).
We can apply partial randomization by keeping the $L_D$ largest terms in the deterministic part. The remainder term is now not given as a linear combination of Paulis, but we can easily re-express the remainder terms $\sum_{l > L_D} H_l$ in terms of Pauli operators using some choice of orbital basis.
In other words, we may use the partial randomization to truncate the rank $L$ of the first factorization.

One of the main bottlenecks of the factorization approach to Trotterization is the cost of implementing basis changes on the quantum computer, which scales quadratic in the number of orbitals.
This cost can be reduced if the diagonal Hamiltonian $\sum_{p\sigma} \lambda_p^{(l)} n_{p\sigma}$ has only $\rho_l < N$ nonzero terms. In that case, the basis transformation can be arbitrary on the remaining $N - \rho_l$ orbitals. Such basis change unitaries can be implemented using $\bigO(N \rho_l)$ gates \cite{kivlichanQuantumSimulationElectronic2018}.
It turns out that for large systems, for each $l$, the $\lambda_p$ can be truncated to $\rho_l$ terms, where $\rho_l$ is significantly smaller than $N$. Truncating small $\lambda_p$ leads to constant $\rho_l$ for large $N$; this effect, however, kicks in at rather large $N$ \cite{peng2017highly,motta2019efficient}.
We can again use partial randomization to truncate the second factorization.
Indeed, in this case for the diagonal Hamiltonian $V^{(l)}$ we only keep the terms for a subset of $\rho_l$ orbitals for which $\lambda_p^{(l)}$ is relatively large.
The remainder is assigned to the randomized part $H_R$.
In \cref{fig:gate counts factorized hchain} we show that this leads to a significant improvement for the hydrogen chain. Additionally, it resolves the issue that for the deterministic approach it is ambiguous what are good truncation thresholds which do not perturb the ground state energy too much.
We note that for the hydrogen chain, the overall gate count using the Pauli representation, as shown in \cref{fig:hchain cost} is lower.

\begin{figure}
    \centering
    \includegraphics[width=\linewidth]{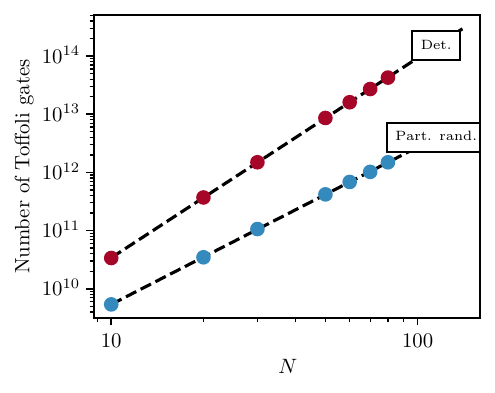}
    \caption{Toffoli gate counts for phase estimation of the hydrogen chain with $N$ atoms using a factorized representation of the Hamiltonian.}
    \label{fig:gate counts factorized hchain}
\end{figure}

As a final comment, we note that for the randomized methods we assume the terms in the Hamiltonian are Pauli operators. However, for the derivations the only relevant property is that $P_l^2 = \id$.
This means that one can also use a Hamiltonian of the form $H = \sum_m h_m \tilde P_m$ where each $\tilde P_m$ is a Pauli with respect to a different choice of orbital basis. For example, one obtains such representations from the double factorization representation of the Hamiltonian.
In that case, $V^{(l)}$ is a sum of terms proportional to $n_{p\sigma} n_{q\tau}$, which under the Jordan-Wigner mapping can be expanded as Pauli operators acting on only two qubits.
Therefore, the Hamiltonian can be expressed as a sum of terms of the form $\tilde P_{m} = (U^{(l)})^\dagger P_m U^{(l)}$, where $U^{(l)}$ is an orbital basis change unitary and $P_m$ is a two-qubit Pauli operator.
The cost of implementing $\exp(-i\varphi \tilde P_m) = (U^{(l)})^\dagger \exp(-i\varphi P_m) U^{(l)}$ is dominated by the cost of implementing the basis change, and since $P_m$ only acts on two qubits, this cost is $\bigO(N)$.
So, in cases where such a representation drastically reduces the weight $\lambda$ of $H$, this may be a beneficial approach for randomized product formulas.
Let $\lambda$ be the weight of $H$ in a Pauli representation, and let $\lambda_{\text{DF}}$ be the weight in a double factorization representation.
Since the cost scales quadratically in the value of $\lambda$, we have an improvement if $\lambda / \lambda_{\text{DF}} = o(\sqrt{N})$.
As an example, \cite{lee2021even} reports for an arrangement for 4 hydrogen atoms with an increasing number of orbitals (as a toy model for the continuum limit), scalings of $\lambda = \bigO(N^{3.1})$ for the Pauli representation and $\lambda_{\text{DF}} = \bigO(N^{2.3})$, so this approach would asymptotically improve the performance of random product formulas.

\section{Conclusion}
We performed a comprehensive analysis of quantum phase estimation using the Hadamard test and time evolution by product formulas for deterministic and (partially) randomized product formulas.
Deterministic product formulas are conceptually simple and their scaling with the approximation error $\eps$ is close to optimal when using higher order methods.
The main downside is their scaling with the number of terms in the Hamiltonian $L$, and a potentially complicated dependence on the Trotter error.
Random product formulas have no dependence on the number terms, but rather on the weight $\lambda$. However, the scaling with $\eps^{-1}$ and $\lambda$ is quadratic, which is suboptimal.
In this work, we analyze phase estimation and Hamiltonian simulation in an integrated way. In particular, this allows us to obtain a sharper analysis of phase estimation for random product formulas.
We also prove that the maximal depth can be reduced from $\eps^{-2}$ to $(1 - \eta)^2 \eps^{-2}$ in regimes with large ground state overlap.
These improvements do not increase the total number of gates, in contrast to other methods (which have better scaling with $\eps$).

Our other main contribution is a modification of the partially randomized scheme of \cite{hagan2023composite} and an analysis of its use in single-ancilla phase estimation. 
This approach combines advantages of deterministic and randomized product formulas. 
We demonstrated that this can lead to an improvement in the required resources for ground state energy estimation.

For a wide variety for benchmark molecular systems we compare the cost of different product formula Hamiltonian simulation methods for phase estimation.
For deterministic product formulas, it is well known that (computable) rigorous upper bounds can heavily overestimate the Trotter error, leading to much shorter Trotter steps (and hence deeper circuits) than necessary.
We numerically analyze the empirical Trotter error, and show that while this gives a large improvement in the required resources, random product formulas can still outperform deterministic product formulas. 
Partial randomization can lead to a further reduction in cost.

On a high level the decomposition strategies into deterministic and randomized parts $H_D$ and $H_R$ are guided by the following aim: Find an approximation $H_D$ to the full Hamiltonian $H$, such that the time evolution under $H_D$ is easier to implement than for $H$, while the remainder is small, $\lambda_R < \lambda$ (which therefore can be implemented cheaply in a randomized way).
This is a change of perspective compared to previous works on optimized Hamiltonian representations.
The possible approximations $H_D$ of the electronic structure Hamiltonian allow for a large degree of freedom in the algorithm design.
We proposed a decomposition strategy based on factorizations of the Coulomb tensor, which lead to improved resource estimates compared to the fully deterministic method.
One can further optimize the representation of $H$ by choosing different orbital bases for the deterministic and randomized Hamiltonians, or use the structure of the Coulomb tensor to optimize the cost. 
We think this is a promising direction for future explorations.

In this work we are counting both non-Clifford gates as well as two-qubit gates. It is unclear whether in practice the non-Clifford gates are necessarily the dominant cost in fault-tolerant quantum computation \cite{litinski2019magic,gidney2024magic}. 
In particular, it appears that magic state distillation may be cheaper for logical error rates around $10^{-9}$ \cite{gidney2024magic}, which means that trade-offs which reduce circuit depth may be extra useful.
Product formulas (using many very small angle single-qubit rotations) are also a promising candidate for approaches which are only partially fault-tolerant \cite{akahoshi2024partially}.
The quantum circuits deriving from product formulas have the advantage of being structurally simple (they consist only of Pauli rotations if the Hamiltonian is a linear combination of Paulis), so it is straightforward to count or optimize other cost measures.

Together with various other developments in the algorithmic theory of product formulas, such as multiproduct formulas and polynomial interpolation techniques \cite{low2019well,vazquez2023well,rendon2024improved,aftab2024multi}, isometries to enlarged basis sets \cite{luo2024efficient} this suggests that both for intermediate-term applications and fully fault-tolerant quantum computation, (partially randomized) product formulas remain a promising approach to ground state energy phase estimation.

\section*{Data availability}
Scripts and data reproducing the figures in this work are available at: \url{https://doi.org/10.5281/zenodo.15387186}.

\section*{Acknowledgements}
We thank Li Liu for helpful discussions.
We are grateful to Markus Reiher for providing access to Euler supercomputer at ETH Z\"urich.
We acknowledge funding from the Research Project ``Molecular Recognition from Quantum Computing.'' Work on ``Molecular Recognition from Quantum Computing'' was supported by Wellcome Leap as part of the Quantum for Bio (Q4Bio) Program.  We also acknowledge financial support from the Novo Nordisk Foundation (Grant No. NNF20OC0059939 `Quantum for Life').

\bibliographystyle{unsrtnat}

\bibliography{library}

\appendix

\onecolumngrid

\newpage

\section{Hadamard test using (partially) randomized time evolution}\label{sec:randomized hadamard}

In this appendix we derive in detail the result of performing a Hadamard test on (partially) randomized product formulas.
We briefly recall the setup and notation.
The Hamiltonian is given as
\begin{align}
    H = \sum_{l} H_l
\end{align}
where for the randomized product formulas we will assume more specifically that
\begin{align}
    H = \sum_l h_l P_l
\end{align}
where the $P_l$ are Pauli operators and $h_l \geq 0$ (using that signs can be absorbed in the Pauli operators). In fact, the only properties used in the derivations are the possibility to implement time evolution along $P_l$, and $P_l^2 = \id$.
We write $V(\varphi) = \exp(-i \varphi P_l)$.
We may rescale the Hamiltonian by a factor
\begin{align}
    \lambda = \sum_l h_l \qquad \text{ so } \qquad H = \lambda \sum_l p_l H_l
\end{align}
where the $p_l$ form a probability distribution $p$.
We may rescale time by a factor $\lambda$, and thereby assume that $H = \sum_l p_l P_l$ where the $p_l$ form a probability distribution $p$.
In the end, when determining the computational cost we have to account for the rescaling by rescaling the relevant parameters by $\lambda$.
The Hamiltonian has energies $E_k$ (which lie in $[-1,1]$ after rescaling with $\lambda$) with eigenprojectors $\Pi_k$.
We apply the Hadamard test with initial state $\ket{\psi}$, which has squared overlap $c_k = \bra{\psi} \Pi_k \ket{\psi}$ with the eigenspace of energy $E_k$.
If $E_k$ is nondegenerate, $\Pi_k = \proj{\psi_k}$ and $c_k = \abs{\braket{\psi | \psi_k}}^2$.

\subsection{Hadamard test using qDRIFT}

The qDRIFT method is to construct a unitary
\begin{align}
    U(\tau, r) = \prod_{i=1}^r \exp(-i \varphi H_{l_r})
\end{align}
where $l_i$ is sampled according to $p$ and $\varphi$ depends on a choice of time step $\tau$.
This gives a random unitary, but it turns out that for appropriately chosen $r$ and $\varphi$ the channel obtained by this process is close in diamond norm to the time evolution unitary $U(t) = \exp(-i Ht)$ \cite{campbell2019random,chen2021concentration}.
This may then be used in combination with phase estimation to compute ground state energies.

In this work we analyze in detail the performance of phase estimation methods based on the Hadamard test using qDRIFT.
The analysis is very similar to that of \cite{wan2022randomized}, which also uses Hadamard tests together with a randomized method for time evolution.
The main difference is that in their approach the randomized time evolution method is based on a Taylor series expansion of the time evolution operator.

The fact that $P_l$ squares to identity implies that for any $\varphi \in \RR$ the operator
\begin{align}
    V_l(\varphi) := \cos(\varphi) \id - i \sin(\varphi) P_l = \exp(-i \varphi P_l)
\end{align}
is unitary.
We now choose $\tau \in [0,1]$ and let $\varphi = \arctan(\tau)$.
Let $\rv{X}(\tau,r)$ and $\rv{Y}(\tau,r)$ be the $\pm 1$-valued random variables resulting from the following process:
\begin{enumerate}
    \item Sample $l_1, \dots, l_r$ from $p$.
    \item Apply a Hadamard test to the unitary
          \begin{align}
              U = V_{l_r}(\varphi) \cdots V_{l_1}(\varphi)
          \end{align}
\end{enumerate}
Let $\rv{Z}(\tau,r) = \rv{X}(\tau,r) + i\rv{Y}(\tau,r)$.
We may now compute the expectation value of $\rv{Z}(\tau,r)$, which we denote by $g(\tau,r)$ to be the following signal
\begin{align}
    %\label{eq:HadamardRandFxr}
    g(\tau,r) : & = \EE \rv{Z}(\tau,r)                                                                                                                                                \\
                & = \sum_{l_1,\dots,l_r} p_{l_1} \dots p_{l_r} \bra{\psi} V_{l_r}(\varphi) \cdots V_{l_1}(\varphi) \ket{\psi}                                                         \\
                & = \sum_{l_1,\dots,l_r} p_{l_1} \dots p_{l_r} \bra{\psi} \frac{\id - i\tau P_{l_r}}{\sqrt{1 + \tau^2}}\cdots \frac{\id - i\tau P_{l_1}}{\sqrt{1 + \tau^2}}\ket{\psi} \\
                & = (1 + \tau^2)^{-r/2} \bra{\psi} (\id - i\tau \sum_{l_r} p_{l_r} P_{l_r})\cdots (\id - i\tau \sum_{l_1} p_{l_1} P_{l_1})\ket{\psi}                                  \\
                & = (1 + \tau^2)^{-r/2} \bra{\psi} (\id - i\tau \sum_{l} p_{l} P_{l})^r \ket{\psi}                                                                                    \\
                & = (1 + \tau^2)^{-r/2} \bra{\psi} (\id - i\tau H)^r \ket{\psi}
\end{align}
which encodes the spectrum of the (normalized) Hamiltonian.
We summarize this conclusion in the following lemma.

\begin{lem}\label{lem:hadamard test qdrift}
    Let
    \begin{align}
        H = \sum_l p_l P_l
    \end{align}
    where $P_l^2 = \id$ and the $p_l$ are a probability distribution.
    Then the random variables $\rv{Z}$ from the measurement outcomes of a Hadamard test using qDRIFT with $r$ Pauli rotations and step size $\tau$ has expectation value
    \begin{align}
        \EE \rv{Z}(\tau,r) = g(\tau, r) =  (1 + \tau^2)^{-r/2} \bra{\psi} (\id - i\tau H)^r \ket{\psi}.
    \end{align}
\end{lem}

This can be expressed in terms of the eigenstate overlaps and energies as
\begin{align}\label{eq:qdrift signal appendix}
    \begin{split}
        g(\tau,r) & = (1 + \tau^2)^{-r/2} \sum_{k} c_k (1 - i\tau E_k)^r                                                                                                                                       \\
                  & = \sum_{k} c_k \underbrace{\left(\frac{1 + E_k^2\tau^2}{1 + \tau^2}\right)^{\frac{r}{2}}}_{\geq \exp(- \frac12 (1 - E_k^2)r\tau^2)} \exp\mleft[-ir \arctan\mleft(\tau E_k \mright)\mright]
    \end{split}
\end{align}
The goal is to extract the value of $E_0$ from this signal, choosing appropriate values of $\tau$ and $r$.
Recall the signal obtained in \cref{eq:signal hadamard test} where we perform time evolution for time $t$.
If we choose $r = \ceil{t^2}$ and $\tau = t/r \approx t^{-1}$ we find that
\begin{align}
    g(\tau,r) & = \sum_{k} d_k e^{-i t F_k}
\end{align}
with
\begin{align}\label{eq:qdrift signal optimal}
    F_k = \tau^{-1} \arctan\mleft(\tau E_k \mright) \approx t E_k, \qquad c_k \exp(- \frac12(1 - E_k^2)) \leq d_k \leq c_k
\end{align}
which approximates $g(t)$, with the difference that the amplitudes are modified by a factor at most $\sqrt{e}$.

\subsection{Randomized Taylor expansion}\label{sec:rte}
In this section we review a refinement of qDRIFT due to Berta, Campbell and Wan~\cite{wan2022randomized}, which we will call the \emph{Randomized Taylor Expansion} (RTE) method.
It proceeds in very similar fashion to qDRIFT, and in fact the above qDRIFT algorithm corresponds exactly to keeping only the first order term in the Taylor expansion.
We write a Taylor series expansion for $\exp(-i\tau)$ from which we get
\begin{align}
    \exp(-i\tau H) & = \sum_{n=0}^{\infty} \frac{(-i\tau )^n}{n!} \left(\sum_l p_l P_l\right)^n                                                               \\
                   & = \sum_{n \even} \frac{(-1)^{n/2}}{n!} \tau^n \left(\id - \frac{i\tau}{n}\sum_l p_l P_l\right) \left(\sum_l p_l P_l\right)^n             \\
                   & = \sum_{n \even} \frac{(-1)^{n/2}}{n!} \tau^n \sum_l p_l \left(\id - \frac{ i\tau}{n+1} P_l\right) \left(\sum_l p_l P_l\right)^n         \\
                   & = \sum_{n \even}\frac{(-1)^{n/2}}{n!} \tau^n \sqrt{1 + \frac{ \tau^2}{(n+1)^2}} \sum_l p_l  V_l(\varphi_n) \left(\sum_l p_l P_l\right)^n
\end{align}
where $\varphi_n = -\arctan( \tau / (n+1))$.
This gives a decomposition
\begin{align}
    \exp(-i\tau H) = \sum_m c_m U_m
\end{align}
where $m$ is a label consisting of an even integer $n$ and $l, l_1, \dots, l_n$
\begin{align}\label{eq:unitary rte}
    \begin{split}
        c_m & = \frac{1}{n!} \tau^n \sqrt{1 + \frac{\tau^2}{(n+1)^2}} p_l p_{l_1} \dots p_{l_n} \geq 0 \\
        U_m & = (-1)^{n/2} V_l(\varphi_n) P_{l_{n}} \cdots P_{l_1}.
    \end{split}
\end{align}
Note that in principle, $n$ can be arbitrary, but one can truncate to finite $n$ at small error, see \cite{wan2022randomized}.
Again, for simulation of time $t$, we can break up in $r$ steps, and for $\tau = t/r$
\begin{align}
    \exp(-i\tau H)^r & = \sum_{m_1, \dots, m_r} c_{m_1} \dots c_{m_r} U_{m_r} \cdots U_{m_1} \\
                     & = \sum b_k W_k
\end{align}
where $k$ is a label consisting of $m_1, \dots, m_r$ and $b_k = c_{m_1} \dots c_{m_r}$, and the unitary $W_k$ is given by $U_{m_r} \cdots U_{m_1}$.
If we now let
\begin{align}\label{eq:prob dist rte}
    B = \sum_{k} b_k \qquad q_k = \frac{b_k}{B}
\end{align}
then sampling according to $q_k$ and applying a Hadamard test to $W_k$ with initial state $\ket{\psi}$ gives random variables $\rv{X}(\tau,r)$ and $\rv{Y}(\tau,r)$ such that for $\rv{Z}(\tau,r) = \rv{X}(\tau,r) + i\rv{Y}(\tau,r)$
\begin{align}
    \EE \rv{Z}(\tau,r) = \frac{1}{B} \bra{\psi} \exp(-i t H) \ket{\psi}.
\end{align}
This means that we can estimate the value of
\begin{align}
    g(\tau r) = g(t) = \bra{\psi} \exp(-i t H) \ket{\psi}
\end{align}
to precision $\gamma$ by estimating the expectation value of $\rv{Z}(\tau,r)$ to precision $\gamma/B$.
In other words, if we want to compute $g(\tau,r)$ using this method we can do so at a sampling overhead of $B^2 / \gamma^2$.
One can bound \cite{campbell2019random}
\begin{align}\label{eq:damping factor rte}
    B = \sum_k b_k = \left( \sum_m c_m\right)^r
    \leq \exp(\tau^2 r)
\end{align}
This means that one can estimate $g(t)$ with constant sampling overhead by taking the same choice of parameters as for qDRIFT with depth $r = \ceil{t^2}$ and step size $\tau = t/r \approx t^{-1}$.

Given a sample $k = (m_1,\dots, m_r)$ the resulting unitary $W_k$ consists of the composition of $r$ unitaries $U_{m_i}$, which itself consists of a single Pauli rotation together with a product of Pauli operators (which makes up a new Pauli operator), so we get a circuit with $r$ Pauli rotations.

\begin{lem}
    Let
    \begin{align}
        H = \sum_l p_l P_l
    \end{align}
    where $P_l^2 = \id$ and the $p_l$ are a probability distribution.
    Then the random variable $\rv{Z}$ resulting from Hadamard tests on RTE with $r$ steps and step size $\tau$ has expectation value
    \begin{align}
        \EE \, B \rv{Z}(\tau,r) = g(\tau r) =  \bra{\psi} \exp(-i\tau r H) \ket{\psi}.
    \end{align}
    Here $B \leq \exp(\tau^2 r)$.
\end{lem}

We now briefly comment on the difference between RTE and qDRIFT.
First of all, if one chooses $\tau$ and $r$ with $r = \Omega(t^2)$, they are very similar. At every step, the probability of sampling order $n \geq 2$ in RTE scales is $\bigO(\tau^2)$, so the expected number of steps where you sample $n \geq 2$ is $\bigO(r \tau^2) = \bigO(1)$ out of $r = \bigO(t^2)$ steps.
The main difference is that RTE, with a Hadamard test, gives an \emph{unbiased} estimator of the signal $g(t)$. This is convenient for analysis. On the other hand, we have also computed a simple expression for the exact signal resulting from qDRIFT, from which one can extract the same information about the energies.
For qDRIFT, the damping depends on the energy as well; if $E_0$ is close to its minimal value of $-1$ (so the Hamiltonian is close to frustration free) this means that the damping factor is smaller.

\subsection{Partially randomized product formulas}
In \cref{sec:composite} we have discussed a partially randomized scheme, which treats certain terms in the Hamiltonian deterministically, and the remainder randomly, and here we explain and analyze this scheme in more detail.
The error analysis is particularly simple, where the signal from the partially randomized product formula has a bias from the deterministic decomposition, and a decreased amplitude due to the randomization. This intuition is illustrated in \cref{fig:partial random intuition}.
We now recall the set-up.
We write
\begin{align}\label{eq:hamiltonian composite}
    H = \underbrace{\sum_{l = 1}^{L_D} H_l}_{= H_D} + \underbrace{\sum_{m=1}^{M} h_{m} P_{m}}_{= H_R}, \qquad \lambda_R = \sum_{m} \abs{h_m}.
\end{align}
Here, the $P_m$ should square to identity, but the $H_l$ need not.
As before, we let
\begin{align}
    H_R = \lambda_R \sum_m p_m P_m
\end{align}
where the $p_m$ form a probability distribution (and we have absorbed signs into the $P_m$).

We let $\trot_{p}(\delta)$ be a deterministic product formula for $H$, written as a sum of the $H_l$ and treating $H_R$ as a single term.
This means we can write
\begin{align}
    \trot_{p}(\delta) = U_I \cdots U_1
\end{align}
where each $U_i$ is either of the form $\exp(-\delta_i H_l)$ for some real $\delta_i$, or $U_i = \exp(-\delta_i H_R)$.
Here the number of terms is $I \leq (L_D + 1)N_{\stage}$, where $N_{\stage}$ is the number of stages of the product formula. For the standard Trotter formula of order $p$, we have $N_{\stage} = 2 \cdot 5^{k-1}$.
For the second order Suzuki-Trotter formula, we get
\begin{align}
    \trot_{2}(\delta) = e^{- \frac12 i \delta H_1} \cdots e^{- \frac12 i \delta H_{L_D}} e^{-i\delta H_R} e^{- \frac12 i\delta H_{L_D}} \cdots e^{- \frac12 i \delta H_1}.
\end{align}
Let $U_{i_1}, \dots, U_{i_{N}}$ for $N \leq N_{\stage}$ denote the unitaries where we time evolve along $H_R$, for times $\delta_{i_1}, \dots, \delta_{i_{N}}$.
We have $\delta_{i_1} + \dots + \delta_{i_N} = \delta$, and $\tilde{\delta} := \abs{\delta_{i_1}} + \dots + \abs{\delta_{i_{N}}} = \bigO(\delta)$ (higher order methods require evolving backwards in time as well).
Next, we will apply RTE to each of the $U_{i_p}$.
To this end, we choose a number of steps $r_{i_p}$ to simulate $\exp(-\delta_{i_p}  H_R)$.
We may write
\begin{align}\label{eq:probabilities composite}
    \begin{split}
        U_{i_p} & = \exp(-i\delta_{i_p} H_R) = B_{i_p} \sum_k q_{i_p,k} W_{i_p,k} \\
        B_{i_p} & \leq \exp(\delta_{i_p}^2 \lambda_R^2 / r_{i_p})
    \end{split}
\end{align}
as in \cref{sec:rte}, where each $W_{i_p,k}$ is a product of Paulis $P_m$ and $r_{i_p}$ Pauli rotations.
Expanding each $U_{i_p}$ in this fashion, and taking the $s$-th power of the result, gives a linear combination of unitaries
\begin{align}
    \trot_{p}(\delta)^s = B \sum_j q_j W_j
\end{align}
where
\begin{align}\label{eq:normalization composite}
    B = (B_{i_1} \dots B_{i_{N}})^s \leq \exp\mleft(s \sum_{p} \delta_{i_p}^2 \lambda_R^2 / r_{i_p} \mright).
\end{align}
Each $W_j$ consists of $s(r_{i_1} + \dots + r_{i_{N}})$ Pauli rotations from the expansion of the evolution along $H_R$, and at most $N_{\stage} L_D s$ applications of $\exp(-\delta H_l)$ for some values of $\delta$.
Now, given $\delta$ and $s$, in order to keep the sample overhead constant, we choose $r_{i_p} = \ceil{\lambda_R^2 \abs{\delta_{i_p}} \tilde{\delta} s}$, which gives
\begin{align}
    s \sum_{p} \delta_{i_p}^2 \lambda_R^2 / r_{i_p} \leq 1.
\end{align}
The total number of Pauli rotations applied in the randomized part of the resulting circuit is
\begin{align}
    r = s\sum_p r_{i_p} = s \sum_{p} \ceil{\lambda_R^2 \delta_{i_p} \delta s} \leq \lambda_R^2 \tilde{\delta}^2 s^2 + sN_{\stage} = \bigO(\lambda_R^2 \delta^2 s^2).
\end{align}
The term $s N_{\stage}$ comes from having to round the values of the $r_{i_p}$ to integers and is a small subleading contribution.
We summarize this result as follows.

\begin{lem}\label{lem:composite simulation}
    Let $H$ be given as in \cref{eq:hamiltonian composite} and let $\trot_{p}(\delta)$ be a deterministic product formula for $H$, written as a sum of the $H_l$ and treating $H_R$ as a single term.
    Then the random variables $\rv{Z}$ which results from the measurement outcomes of a Hadamard test, using the partially randomized product formula described above, has expectation value
    \begin{align}
        \EE \, B \rv{Z}(\delta,s) =  \bra{\psi} \trot_{p}(\delta)^s \ket{\psi}.
    \end{align}
    Here $B = \bigO(1)$, and the circuit consists of at most $N_{\stage} L_Ds$ time evolution unitaries along one of the $H_l$, and $r = \bigO(\lambda_R^2 \delta^2 s^2)$ time evolution unitaries along one of the $P_m$.
\end{lem}

\begin{figure}
    \centering
    \includegraphics[width=0.5\linewidth]{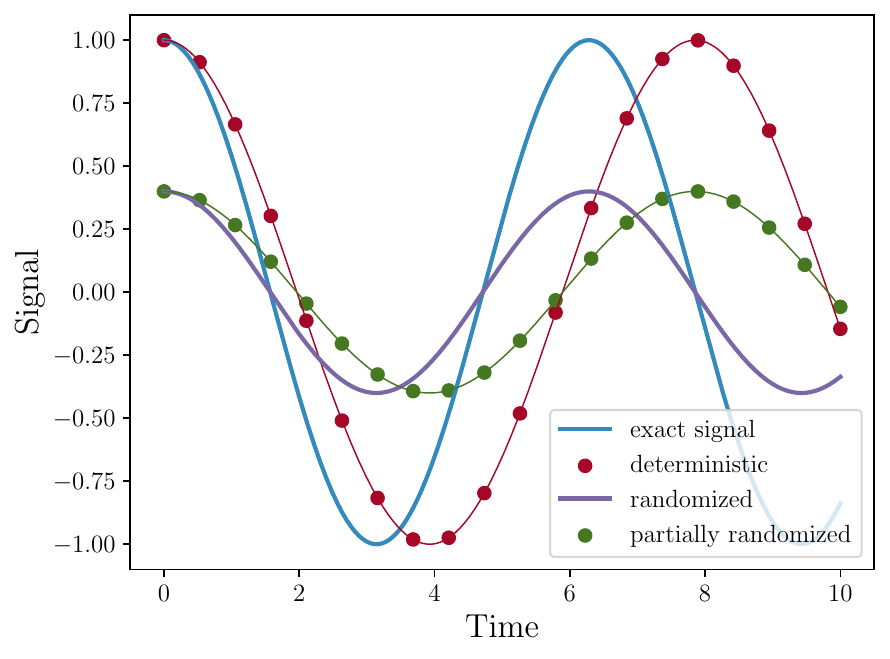}
    \caption{Intuition for the error analysis: deterministic product formulas give a discrete approximation which has a bias in the frequency, randomized product formulas give a damping of the amplitude (and no bias when using RTE); the partially randomized method has both.
    }
    \label{fig:partial random intuition}
\end{figure}

We note that the second order Trotter-Suzuki formula has the benefit in this context that it only evolves in the positive time direction, so we have $\tilde{\delta} = \delta$.
The conclusion of \cref{lem:composite simulation} is that from the Hadamard test with the partially randomized product formula, we get a signal which, up to a damping, implements the deterministic product formula which treats $H_R$ as a single term.
The resulting Trotter error can be bounded using any standard technique (such as commutator bounds).

Different schemes for composite random and deterministic product formulas have been proposed \cite{ouyang2020compilation,hagan2023composite,pocrnic2024composite}, treating a relatively small number of important terms using deterministic Trotterization and a tail of terms with small weight randomly.
The scheme described above is the same as in \cite{hagan2023composite}, where we have replaced qDRIFT by RTE.
Compared to the diamond norm bounds of \cite{hagan2023composite}, the randomized part of the product formula does not introduce any error. Additionally, in our analysis it is clear that if we are only interested in ground state energies, $\trot_p(\delta)$ only needs to be accurate on the ground state, a significantly weaker requirement on the Trotter error.
This allows for a sharper analysis of the performance of partially randomized product formulas for single-ancilla phase estimation.

A final comment is that a general potential disadvantage of (partially) randomized product formulas is that in principle for each Hadamard test one has to sample a new random circuit. In practice, it may be more convenient (e.g. for compilation purposes) to repeat the same circuit multiple times. For randomized product formulas \cite{chen2021concentration} as well as partially randomized product formulas \cite{kiss2023importance} it has been shown that this is in fact possible. These works show that with increasing circuit depth, the expectation value of observables (such as Hadamard test outcomes) rapidly concentrate with respect to the randomness of the sampled circuit.

\subsection{Halving the required time evolution}\label{sec:halving time}
There is a basic trick to reduce the maximum required evolution time when using the Hadamard test for time evolution \cite{reiher2017elucidating}.
This assumes we have a symmetric Trotter formula (meaning that $\trot_{p}(-t) = \trot_{p}(t)^\dagger$), which is the case for the even order Suzuki-Trotter formulas.
We write out the Trotter unitary in terms of the individual Hamiltonian time evolutions
\begin{align}
    \trot_p(\delta) = U_I(\delta_I) \cdots U_1(\delta_1)
\end{align}
where each unitary is time evolution along a local term $H_l$ of the Hamiltonian.
Now, for some arbitary unitary $U$, let $C_{\pm}U$ denote a controlled version
\begin{align}
    C_{\pm}U = \proj{0} \ot U^\dagger + \proj{1} \ot U
\end{align}
so the product formula, we have
\begin{align}
    C_{\pm}U_i(\delta_i) = \proj{0} \ot e^{i\delta_i H_{l_i}} + \proj{1} \ot e^{-i\delta_i H_{l_i}}.
\end{align}
By the symmetry property of the product formula, we have
\begin{align}
    C_{\pm} \trot_p(\delta) = C_{\pm} U_I(\delta_I) \cdots C_{\pm} U_1(\delta_1).
\end{align}
It is now easy to see that applying the Hadamard test to the unitary $(C_{\pm} \trot_p(\delta))^s$ gives expectation value $\bra{\psi} \trot_p(\delta)^{2s} \ket{\psi}$.
This is implemented by the circuit in \cref{fig:time halving trick}.
If the $U_i(\delta_i)$ are Pauli rotations, then it requires only Cliffords and one single-qubit rotation to implement $C_{\pm} U_i(\delta_i)$.
Therefore, this trick reduces the required number of single-qubit rotations by a factor of 2.

This trick does \emph{not} apply for randomized product formulas (which do not have the right symmetry property). However, it can be modified to apply for partially randomized product formulas.
In that case, one proceeds as above, thinking of $H$ as consisting of $H_1, \dots, H_{L_D}$ and $H_R$ as in \cref{eq:hamiltonian composite}.
This gives a circuit with controlled applications of unitaries of the form $U = \exp(\pm i\delta H_R)$.
Finally, we expand
\begin{align}
    C_{\pm} U = (\proj{0} \ot U^\dagger + \proj{1} \ot I)(\proj{0} \ot I + \proj{1} \ot U)
\end{align}
and separately apply RTE to $U$ and $U^\dagger$ (with independent randomness).
Altogether, this halves the gate count for the deterministic part of the partially randomized product formula (but not for the randomized part).

\begin{figure}[t]
    \centering
    \begin{quantikz}
        \lstick{$\ket{0}$} & \gate{H} & \ctrl{1}\gategroup[2,steps=3,style={dashed, color=gray, inner sep=6pt}]{$C_{\pm} \trot_p(\delta)$} & \qw \ \ldots\ & \ctrl{1} & \qw \ \ldots \ & \gate{R(\theta)} & \gate{H} & \meter{}\\
        \lstick{$\ket{\psi}$} & \qw & \gate{e^{\pm i \delta_1 H_{l_1}}} & \qw \ \dots\ & \gate{e^{\pm i \delta_I H_{l_I}}} & \qw \ \ldots \ & \qw & \qw & \qw
    \end{quantikz}
    \vspace{1cm}

    \centering
    \begin{quantikz}
        & \ctrl{1} & \qw \\
        & \gate[3]{e^{\pm i\delta P}} & \qw \\
        && \qw \\
        && \qw
    \end{quantikz}
    =\begin{quantikz}
        & \qw & \ctrl{1} & \qw & \ctrl{1} & \qw & \qw \\
        & \gate[3]{C} & \gate{X} & \gate[1]{e^{-i\delta Z}} & \gate{X} & \gate[3]{C^\dagger} & \qw \\
        & \qw & \qw & \qw & \qw & \qw & \qw \\
        & \qw & \qw & \qw & \qw & \qw & \qw
    \end{quantikz}

    \caption{Performing the  Hadamard test on $C_{\pm} \trot_p(\delta)^s$ allows one to estimate $\bra{\psi} \trot_p(\delta)^{2s} \ket{\psi}$.
        If $U = e^{i\delta P}$ is a Pauli rotation, we can let $C$ be a Clifford which is such that $C P C^\dagger$ acts as a Pauli $Z$ on the first qubit, and use this to implement $C_{\pm}U$ with only one single-qubit rotation.}
    \label{fig:time halving trick}
\end{figure}

\section{Robust phase estimation}\label{sec:rpe}
As a phase estimation algorithm we consider \emph{robust phase estimation} \cite{higgins2009demonstrating,kimmel2015robust,belliardo2020achieving,ni2023low}. This algorithm is particularly simple, and requires shorter depth if the guiding state is close to the ground state.
Its main disadvantage is that it requires the ground state overlap $\eta$ to be at least some constant.

The idea of robust phase estimation is to estimate the time evolution signal at times $t = 2^m$ for $m = 0, 1, \dots, M$ with $M = \ceil{\log \eps^{-1}}$.
For each $m$ the outcome of the Hadamard test $\rv{Z}(2^m)$ has expectation value
\begin{align}
    g(2^m) = \sum_k c_k \exp(-i 2^m E_k).
\end{align}
We repeat $N_m$ times and let $\overline{\rv{Z}}(2^m)$ denote the average.
Given some outcome we compute an angle $\phi_m = \arg(\overline{Z}(2^m))$
Each $m$ corresponds to one bit of precision for the estimate of $E_0$, i.e. $\phi_m$ is an approximation of $2^m E_0$ modulo $2\pi$.
The set of estimates of $E_0$ compatible with $\phi_m$ is $\{2^{-m}(2\pi k + \phi_m) : k =0,\dots,2^m-1\}$.
Given a guess $\theta_{m-1}$ for $E_0$ after $m-1$ rounds, we let $\theta_{m}$ be given by
\begin{align}
    \theta_m = 2^{-m}\left(2\pi k + \phi_m)\right)
\end{align}
for the integer $k = 0, \dots, 2^m - 1$ such that the angles $\theta_{m-1}$ an $\theta_m$ are as close as possible.
For two angles $\theta, \phi$ we let $d(\theta, \phi)$ denote the distance between the angles, so
\begin{align}
    d(\theta, \phi) = \min_{k \in \ZZ} \abs{\theta - \phi + 2k\pi}.
\end{align}
The robust phase estimation algorithm is prescribed in detail in \cref{algo:rpe vanilla}. There, the angles $\phi_j$ are treated as input (and can, for example, be obtained by sampling Hadamard test measurement outcomes).

\begin{algorithm}
    \caption{Robust phase estimation}\label{algo:rpe vanilla}
    \DontPrintSemicolon
    \SetAlgoLined
    \Input{A sequence of angles $\phi_m$, $m = 0, 1, \dots, M$.}
    \Output{A phase estimate $\theta$.}
    \BlankLine
    $\theta_{-1} \leftarrow 0$\;
    \For{$m=0,1,\dots, M$}{
        $S_m = \{2^{-m}\left(2\pi k + \phi_m\right) : k = 0, \dots, 2^{m} - 1\}$ \;
        $\theta_m \leftarrow \argmin\limits_{\theta \in S_m} d(\theta, \theta_{m-1})$, $-\pi < \theta_m \leq \pi $ \;
    }
    $\theta \leftarrow \theta_{M}$
\end{algorithm}

\subsection{Rigorous bounds for robust phase estimation}\label{sec:rigorous bounds rpe}
We will now review the argument for the scaling of robust phase estimation, and argue for what happens when using randomized simulation methods.
The algorithm gives an accurate estimate of a target $E$ if the estimate $\phi_m$ in round $m$ is within distance at most $\frac{\pi}{3}$ of $2^m E$. If this is not the case, we will say that there is an error in round $m$.
The Heisenberg scaling of robust phase estimation is achieved by making sure that the probability of error in early rounds is sufficiently small.
The following results are derived in \cite{higgins2009demonstrating,kimmel2015robust,belliardo2020achieving} and are the key to proving Heisenberg limited scaling as well as robustness of the phase estimation algorithm.

\begin{lem}\label{lem:heisenberg scaling}
    Let $\xi \leq 1$.
    \begin{enumerate}
        \item If in \cref{algo:rpe vanilla}, $d(\phi_m, 2^m E) < \frac{\pi}{3}$ for $m = 0, 1, \dots, m$, then $\theta_m$ in \cref{algo:rpe vanilla} is such that $d(\theta_m, E) \leq 2^{-m} \frac{\pi}{3}$.
              If $d(\phi_m, 2^m E) < \frac{\pi}{3}$ for $m = 0,\dots,M-1$ and $d(\phi_M, 2^M E) \leq \xi$, then $d(\theta, E) \leq 2^{-M}\xi$.
        \item Suppose that $\{\rv{\phi}_m\}_{m=0}^M$ are independent real random variables.
              such that for $m = 1, \dots, M$ we have
              \begin{align}
                  \prob(d(\rv{\phi}_m, 2^m E) \geq \frac{\pi}{3}) \leq \xi^2 4^{- \alpha(M - m)}
              \end{align}
              for some $\alpha > 1$, and
              \begin{align}
                  \EE \, d(\rv{\phi}_M, 2^M E)^2 \leq \xi^2.
              \end{align}
              Then applying \cref{algo:rpe vanilla} to $\{\rv{\phi}_m\}_{m=0}^M$ gives an estimate $\theta$ of $E$ with
              \begin{align}
                  \EE \, d(\theta, E)^2 \leq C\xi^2 4^{-M}
              \end{align}
    \end{enumerate}
\end{lem}

\begin{proof}
    The argument for the first claim is by induction. For $m = 0$, the result holds by assumption. Suppose that $d(\phi_m, 2^m E) < \frac{\pi}{3}$ and $d(\theta_{m-1}, E) \leq 2^{-m + 1} \frac{\pi}{3}$.
    Then the candidate set $S_m$ contains an estimate which is $(2^{-m} \frac{\pi}{3})$-close to $E$ given by $2^{-m}(2\pi k + \phi_m)$. Let $k'$ be the integer such that $\theta_m = 2^{-m}(2\pi k' + \phi_m)$, then we would like to show that $k = k'$. The distance between candidates in $S_m$ is at least $2^{-m +1} \pi$, so it suffices to show that $d(\theta_{m - 1}, 2^{-m}(2\pi k + \phi_m)) < 2^{-m} \pi$.
    We use the triangle inequality, and the induction hypothesis that $d(\theta_{m-1}, E) \leq 2^{-m + 1} \frac{\pi}{3}$ to confirm this:
    \begin{align}
        d(\theta_{m - 1}, 2^{-m}(2\pi k + \phi_m)) \leq d(\theta_{m-1}, E) + d(E, 2^{-m}(2\pi k + \phi_m)) < 2^{-m + 1} \frac{\pi}{3} + 2^{-m} \frac{\pi}{3} = 2^{-m} \pi.
    \end{align}
    In the last round, if $d(\phi_M, 2^M E) \leq \xi$, then $d(\theta, E) \leq 2^{-M}\xi$.
    Next, if for $l \leq m$ we have $d(\phi_l, 2^l E) < \frac{\pi}{3}$ (but possibly not for $l > m$), then the output $\theta = \theta_M$ satisfies $d(\theta_M, E) \leq 2^{-m + 3} \frac{\pi}{3}$.
    This follows from the observation that the algorithm is such that $d(\theta_l, \theta_{l+1}) \leq 2^{-l} \pi$, so
    \begin{align}
        d(\theta_M, E) \leq d(\theta_m, E) + \sum_{l = m}^{M - 1} d(\theta_l, \theta_{l+1}) \leq 2^{-m} \frac{\pi}{3} + \underbrace{\sum_{l = m}^{M - 1} 2^{-l} \pi}_{\leq 2^{-m + 1} \pi} \leq 2^{-m + 3} \frac{\pi}{3}.
    \end{align}
    We now consider the case where the angles $\rv{\phi_m}$ are random variables.
    Let $p_m$ denote the probability that $m = 0, \dots, M-1$ is the smallest value for which we have an error (so $d(\phi_m, 2^m E) \geq \frac{\pi}{3}$), and let $p_M$ denote the probability there are no errors for $m = 0, \dots, M-1$.
    Then we can bound the mean square error
    \begin{align}
        \EE d(\theta_M, E)^2 \leq \sum_{m = 0}^{M-1} p_m \left(\frac{8 \pi}{3 \cdot 2^{m-1}}\right)^2 + p_M 4^{-M}\EE \, d(\rv{\phi}_M, 2^M E)^2.
    \end{align}
    We may bound $p_m$ by the probability of having an error in round $m$, so $p_m \leq \xi^2 4^{- \alpha(M - m)}$, and $p_M \leq 1$ which gives
    \begin{align}
        \EE \, d(\theta_M, E)^2 & \leq \sum_{m = 0}^{M-1} \xi^2 4^{- \alpha(M - m)} \left(\frac{8 \pi}{3 \cdot 2^{m-1}}\right)^2 + \xi^2 4^{-M}          \\
                                & = \left(1 + \left(\frac{16\pi}{3}\right)^2 \sum_{m=0}^{M-1} 4^{- (\alpha - 1)(M-m)} \right) \xi^2 4^{-M}               \\
                                & \leq \left(1 + \frac{1}{4^{\alpha - 1} - 1} \left(\frac{16\pi}{3} \right)^2 \right) \xi^2 4^{-M}      = C \xi^2 4^{-M}
    \end{align}
    for a constant $C$ which only depends on $\alpha$.
\end{proof}

In \cref{sec:rpe empirical} we will give numerical estimates for the relevant constant factors.
The angles $\phi_m$ are obtained from a Hadamard test, simulating time evolution for time $2^m$.
We will assume that we have a guiding state $\ket{\psi}$ with sufficiently high ground state overlap.
We first repeat the following fact from \cite{ni2023low}:

Suppose $Z \in \CC$, $d_k \geq 0$ and $E_k \in \RR$.
If
\begin{align}\label{eq:bound on z}
    \abs{Z - \sum_{k} d_k \exp(i E_k)} \leq \beta, \qquad \frac{\beta + \sum_{k \geq 1} d_k}{d_0} \leq 1
\end{align}
then we can bound
\begin{align}
    \abs{Z - d_0  \exp(i E_0)} \leq \beta + \sum_{k \geq 1} d_k.
\end{align}
which implies that the sine of the angle between $Z$ and $d_0 \exp(iE_0)$ satisfies $\abs{\sin(d(E_0, \arg(Z)))} \leq (\beta + \sum_{k \geq 1} d_k)/ d_0$ and hence
\begin{align}\label{eq:angle rpe}
    d(E_0, \arg(Z)) \leq \arcsin\mleft(\frac{\beta + \sum_{k \geq 1} d_k}{d_0} \mright).
\end{align}

% \begin{proof}
%     We have
%     \begin{align}
%         \abs{Z - \sum_{k} d_k \exp(i E_k)} \geq \abs{Z - d_0  \exp(i E_0)} - \sum_{k \geq 1} d_k
%     \end{align}
%     Together with \cref{eq:bound on z} we get that
%     \begin{align}
%         \abs{Z - d_0  \exp(i E_0)} \leq \beta + \sum_{k \geq 1} d_k.
%     \end{align}
%     This implies that the sine of the angle between $Z$ and $c_0 \exp(iE_0)$ satisfies
%     \begin{align}
%         \abs{\sin(d(E_0, \arg(Z)))} \leq \frac{\beta + \sum_{k \geq 1} d_k}{d_0}.
%     \end{align}
% \end{proof}

We run robust phase estimation on a signal estimated by a Hadamard test, and estimate the energy $E_0$ based on the outcome $\theta = \theta_M$.
We consider two variants. For the randomized methods, we assume we have a normalized Hamiltonian $H = \sum_l p_l P_l$ (in general that means we have to rescale $\eps$ with $\lambda$).

\begin{itemize}
    \item We estimate $\phi_j$ from applying a Hadamard test to $e^{-iHt}$, i.e. from the exact time evolution, for time $t = 2^m$.
          We use $2N_m$ samples to get an estimate $Z_m$ of $g(2^m) = \bra{\psi} e^{-iH 2^m} \ket{\psi}$, and let $\phi_m = -\arg(Z_m)$.
          We estimate $E_0$ by $\theta_M$.
          The total amount of time evolution we use is
          \begin{align}
              t_{\tot} = \sum_{m=0}^M 2N_m 2^{m}.
          \end{align}
    \item We estimate $\phi_j$ from applying Hadamard tests using a randomized product formula (either qDRIFT or RTE) for time evolution $t = 2^m$ with $r_m = \Omega(2^{2m})$ Pauli rotations, and time step $\tau_m = t / r_m$.
          We use $2N_m$ samples to get an estimate $Z_m$ of of the signal, and let $\phi_m = -\arg(Z_m)$.
          For RTE, we estimate $E_0$ by $\theta_M$, and in the case of qDRIFT, we estimate $E_0$ by $\tau_M^{-1} \tan (\tau_M \theta_M)$.
          In both cases, the total number of Pauli rotations used is
          \begin{align}\label{eq:total rotations randomized}
              r_{\tot} = \sum_{m=0}^M 2N_m r_m.
          \end{align}
\end{itemize}

The argument for the exact time evolution model in the following result follows \cite{higgins2009demonstrating,kimmel2015robust,belliardo2020achieving} for \ref{it:rpe time evolution} and \cite{ni2023low} for \ref{it:rpe time depth}. A difference compared to \cite{ni2023low} is that for \ref{it:rpe random product} we bound the mean square error, instead of bounding a success probability for an error guarantee.
The application to randomized product formulas in \ref{it:rpe random product} and \ref{it:rpe rotation depth} is a straightforward consequence of the analysis.

\begin{thm}\label{thm:rpe}
    Suppose that we have a guiding state $\ket{\psi}$ with $c_0 \geq \eta > 4 - 2\sqrt{3} \approx 0.53$.
    \begin{enumerate}
        \item\label{it:rpe time evolution} In the exact time evolution model, we obtain an estimate with root mean square error $\eps$ using $t_{\max} = \bigO(\frac{1}{\eps})$, and $t_{\tot} = \bigO(\frac{1}{\eps})$.
        \item\label{it:rpe random product} When using RTE or qDRIFT, we obtain an estimate with root mean square error $\eps$ using maximal number of rotations per circuit $r_{\max} = \bigO(\frac{1}{\eps^2})$ and total number of rotations $r_{\tot} = \bigO(\frac{1}{\eps^2})$.
    \end{enumerate}
    In the limit $\eta \to 1$, for $\xi > \arcsin(\frac{1 - \eta}{\eta})$, we can reduce the maximal depth.
    \begin{enumerate}[resume]
        \item\label{it:rpe time depth} In the exact time evolution model, we obtain an estimate with root mean square error $\eps$ using $t_{\max} = \bigO(\frac{\xi}{\eps})$, and $t_{\tot} = \bigO(\frac{1}{\xi\eps})$.
        \item\label{it:rpe rotation depth} When using RTE or qDRIFT, we obtain an estimate with root mean square error $\eps$ using maximal number of rotations per circuit $r_{\max} = \bigO(\frac{\xi^2}{\eps^2})$ and total number of rotations $r_{\tot} = \bigO(\frac{1}{\eps^2})$.
    \end{enumerate}

\end{thm}

\begin{proof}
    We estimate $X_m$ and $Y_m$ using $N_m$ Hadamard tests, and thus estimate $Z_m = X_m + iY_m$ using $2N_m$ Hadamard tests.
    For the case where we use a randomized product formula, we take $r_m = 2^{m + M}$.
    The expectation value of $\rv{Z}_m$ is given by
    \begin{align}
        \EE \rv{Z}_m = \sum_{k} d_k e^{-i2^m F_k}
    \end{align}
    where $F_k = E_k$ for the case of exact time evolution and RTE, and $F_k = 2^{M} \arctan(2^{-M} E_k)$ for qDRIFT.
    Note that if $\abs{F_0 - \theta_M} \leq \eps$, then
    \begin{align}
        \abs{E_0 - 2^M \tan(2^{-M} \theta_M)} & \leq 2^M\abs{\tan(2^{-M}\theta_M) - \tan(2^{-M}F_0)} \\
                                              & \leq \eps + \bigO(2^{-2M} \eps)
    \end{align}
    so it suffices to get an accurate estimate for $F_0$.
    For qDRIFT, we assume for convenience that $\abs{E_0} \geq \abs{E_k}$ for $k \geq 1$ (which can always be achieved by shifting $H$ with a constant factor).

    For exact time evolution, $d_k = c_k$, for the randomized methods $e^{-1}c_k \leq d_k \leq c_k$ (while $d_k$ also depends on $m$, this bound holds for all $m$).
    In particular, $d_0 \geq e^{-1}c_0 \geq e^{-1}\eta$.
    We let $0 < \beta = \eta(1 + \sin(\frac{\pi}{3})) - 1$ for the case with exact time evolution and $0 < \beta = e^{-\frac12}(\eta(1 + \sin(\frac{\pi}{3})) - 1)$ for the randomized cases.
    This choice of $\beta$ is such that \cref{eq:bound on z} is satisfied
    \begin{align}
        \arcsin \mleft(\frac{\beta + \sum_{k \geq 1} d_k}{d_0} \mright) \leq \frac{\pi}{3}.
    \end{align}
    Now, by \cref{eq:angle rpe} and Hoeffding's inequality we see that for $\rv{\phi}_m = -\arg(\rv{Z}_m)$
    \begin{align}\label{eq:bound prob error}
        \begin{split}
            \prob\left(d(\rv{\phi}_m, F_0) \geq \frac{\pi}{3}\right) & \leq \prob\left(\abs{\rv{Z}_m - \EE \rv{Z}_m} \geq \beta\right) \leq \prob\left(\abs{\rv{X}_m - \EE \rv{X}_m} \geq \frac{\beta}{\sqrt2}\right) + \prob\left(\abs{\rv{Y}_m - \EE \rv{Y}_m} \geq \frac{\beta}{\sqrt2}\right) \\
                                                                     & \leq 4 \exp\mleft( -\frac{N_m \beta^2}{4} \mright).
        \end{split}
    \end{align}
    In particular, we find that by choosing $N_m = \beta^{-2}(2\log(\xi^{-1}) +\bigO(M - m))$ we can achieve
    \begin{align}\label{eq:correctnes rpe condition}
        \prob(d(\rv{\phi}_m, 2^m F_0) \geq \frac{\pi}{3}) \leq \xi^2 4^{-\alpha (M - m)},
    \end{align}
    which is a necessary condition for the RPE algorithm to succeed (see \cref{lem:heisenberg scaling}).
    We now discuss the first part of the Theorem and take $\xi$ constant.
    By \cref{lem:heisenberg scaling} and \cref{eq:correctnes rpe condition} $\eps^2$ is bounded as $\bigO(4^{-M})$.
    The total required evolution time is given by
    \begin{align}
        t_{\tot} = \sum_{m=0}^M 2N_m 2^m = \bigO(2^M)
    \end{align}
    and similarly $r_{\tot} = \bigO(2^{2M})$.
    This results in $t_{\tot} = \bigO(\eps^{-1})$ and $r_{\tot} = \bigO(\eps^{-2})$ and proves statements \ref{it:rpe time evolution} and \ref{it:rpe random product} of the Theorem.

    For \ref{it:rpe time depth} and \ref{it:rpe rotation depth}, we choose $N_M = \bigO(\beta^{-2} \xi^{-2})$ samples in the last round.
    To bound the error, we bound by
    \begin{align}\label{eq:expected error last round}
        \EE \, d(\rv{\phi}_M, 2^M F_0)^2 \leq \prob\mleft(\abs{\rv{Z}_M - \EE \rv{Z}_M} \geq \beta\mright)\pi^2 + \EE \left[ d(\rv{\phi}_M, 2^M F_0)^2 \bigg\vert \, \abs{\rv{Z}_M - \EE \rv{Z}_M} < \beta \right].
    \end{align}
    The first term is bounded by \cref{eq:bound prob error}, so $\prob\mleft(\abs{\rv{Z}_M - \EE \rv{Z}_M} \geq \beta\mright) = 4\exp(-\Omega(\xi^{-2})) = \bigO(\xi^{-2})$.
    For the second term, we use that for $\delta = \abs{{Z}_M - \EE \rv{Z}_M} < \beta$ we have
    \begin{align}
        d(\phi_M, 2^M F_0) \leq \arcsin\mleft(\frac{\delta + (1 - \eta)}{\eta}\mright) = \bigO(\delta + (1 - \eta))
    \end{align}
    (which for \cref{lem:heisenberg scaling} we would like to be $\bigO(\xi^2)$).
    Here we use \cref{eq:angle rpe}, $\eta^{-1}(\delta + (1 - \eta)) \leq \frac{\pi}{3}$ and $\eta$ is close to 1.
    This can be used to see that the second term in \cref{eq:expected error last round} can be bounded by
    \begin{align}
        \EE \, \arcsin\mleft(\frac{\abs{{Z}_M - \EE \rv{Z}_M} + (1 - \eta)}{\eta}\mright)^2 = \bigO\mleft(\EE \, \abs{\rv{Z}_M - \EE \rv{Z}_M}^2 + (1 - \eta)^2 \mright) = \bigO(\xi^2)
    \end{align}
    if $N_M = \Omega(\xi^{-2})$.
    By \cref{lem:heisenberg scaling} and \cref{eq:correctnes rpe condition}, this implies root mean square error $\eps = \bigO(\xi 2^{-M})$.
    The maximal evolution time is $t_{\max} = 2^M = \bigO(\xi \eps^{-1})$, and the randomized methods use $r_{\max} = r_M = \bigO(\xi^2 \eps^{-2})$, and
    \begin{align}
        t_{\tot} = \sum_{m=0}^M 2N_m 2^m = \sum_{m=0}^{M}\bigO(\beta^{-2} (\log(\xi^{-1}) + (M - m)) 2^m) + \bigO(\xi^{-2} 2^M) = \bigO(\xi^{-2} 2^M) = \bigO(\xi^{-1} \eps^{-1}).
    \end{align}
    With the randomized methods, the total number of rotations is similarly given by
    \begin{align}\label{eq:total rotations random}
        r_{\tot} = \sum_{m=0}^M 2N_m r_m = \bigO(\xi^{-2} 2^{2M}) = \bigO(\eps^{-2}).
    \end{align}
\end{proof}

We briefly comment on the choice of depth $r_m$ for the randomized product formulas, what choice minimizes the total cost, and what is the resulting difference between RTE and qDRIFT.
Recall that the depth chosen determines how much the amplitude of the signal is damped, as in \cref{eq:qdrift signal appendix} and \cref{eq:damping factor rte}. If the relevant component of the signal is damped by a factor $B$, we need to scale the number of samples by $\bigO(B^2)$.
For RTE for time $t$ and $r$ rotations, we get a damping factor $B \leq \exp(t^2 / r)$, which means that we have to take a number of samples scaling with $B^2 \leq \exp(2t^2/r)$ to estimate the signal to constant precision. The total cost scales as $B^2 r$, which is optimized by $r = 2t^2$.
In particular, for $t = 2^m$, we should choose $r_m = 2^{2m+1}$.
For qDRIFT, for our proof analysis above it was convenient to choose $r_m = 2^{M + m}$, since that means that the we have a constant value of $F_k$ over different rounds $m$.
Alternatively, one may take $r_m = 2^{2m}$, in which case in round $m$ we have $2^m F_0 = 2^{2m} \arctan(2^{-m} E_0) = 2^m E_0 + \bigO(2^{-m})$.
Since we only need to estimate $2^m E_0$ to constant precision in the $m$-th round, this does not make a significant difference (one has to be careful for small $m$, but in that case it is anyways cheap to take $r_m$ larger).
By \cref{eq:qdrift signal appendix}, we have $c_0 \geq \exp(- t^2 (1 - E_0^2)/ 2r) \geq \exp(- t^2 / 2r)$.
In this case, optimizing the total resources means one should take $r = t^2$ (so this is a factor of 2 more efficient than RTE).
One caveat is that the damping depends on $E_k$ for qDRIFT, which makes a difference if we do not start with the ground state. This is beneficial if $\abs{E_0} \leq \abs{E_k}$ for all $k \neq 0$ with $c_k > 0$ (the excited states are damped more than the ground state).
% However, if this is not the case, one should be careful that the ratio $\sum_{k \geq 1} d_k / d_0$ is not too large.

\subsection{Performance of robust phase estimation}\label{sec:rpe empirical}
In the previous section we saw that robust phase estimation achieves Heisenberg limited scaling, when choosing an appropriate number of samples for the Hadamard test.
There, we did not attempt to optimize constant factors.
In the special case where we start with an exact eigenstate, \cite{belliardo2020achieving} shows a bound for the root mean square error of $\eps \leq C_{\tot} \pi/ t_{\tot}$ with $C_{\tot} \approx 25$ using $N_m = \ceil{N_M + (M - m)\sloperpe}$ samples with $N_M = 11$ and $\sloperpe\approx 4.11$, based on the analysis in \cref{lem:heisenberg scaling}.
This is in contrast to the optimal performance of QPE in general, which has optimal constant factor $C_{\tot} = 1$ \cite{berry2000optimal}.
However, the rigorous bound of \cite{belliardo2020achieving} is still loose, and for our resource estimates we use an empirically estimated constant.
Numerical simulation shows that using the choice of parameters from \cite{belliardo2020achieving}, we achieve a scaling with a constant factor of $C_{\tot} \approx 5$, see \cref{fig:rpe empirical scaling} (a).
The empirical scaling with the \emph{maximal} time evolution $t_{\max} = 2^M$ is $\eps \approx C_{\max}\pi/t_{\max}$ for $C_{\max} \approx 0.08$.
These results corroborate previous numerical estimates, see e.g. \cite{dutkiewicz2024error}.
The total time evolution is given by
\begin{align}\label{eq:upper bound time evolution}
    t_{\tot} = \sum_{m=0}^M 2N_m 2^m = \sum_{m=0}^M 2(N_M + \sloperpe(M-m))2^m \leq 4(N_M + \sloperpe)2^M = 4(N_M + \sloperpe)t_{\max}.
\end{align}
% using that $\sum_{m=0}^M 2^m \approx 2^{M+1}$ and $\sum_{m=0}^M (M-m)2^m = \sum_{k=0}^{M-1} \sum_{m=0}^k 2^{m} \approx \sum_{k=0}^{M-1} 2^{k+1} \approx 2^{M+1}$
This estimate only ignores a small subleading term, and gives a result which is consistent with the numerical relation between $C_{\tot}$ and $C_{\max}$ from \cref{fig:rpe empirical scaling}.
These estimates assume a perfectly accurate quantum computer. By design, the robust phase estimation procedure can tolerate a small amount of error. See \cite{dutkiewicz2024error} for a discussion of the error budget in the context of quantum chemistry computations.

\begin{figure}
    \includegraphics[width=.75\linewidth]{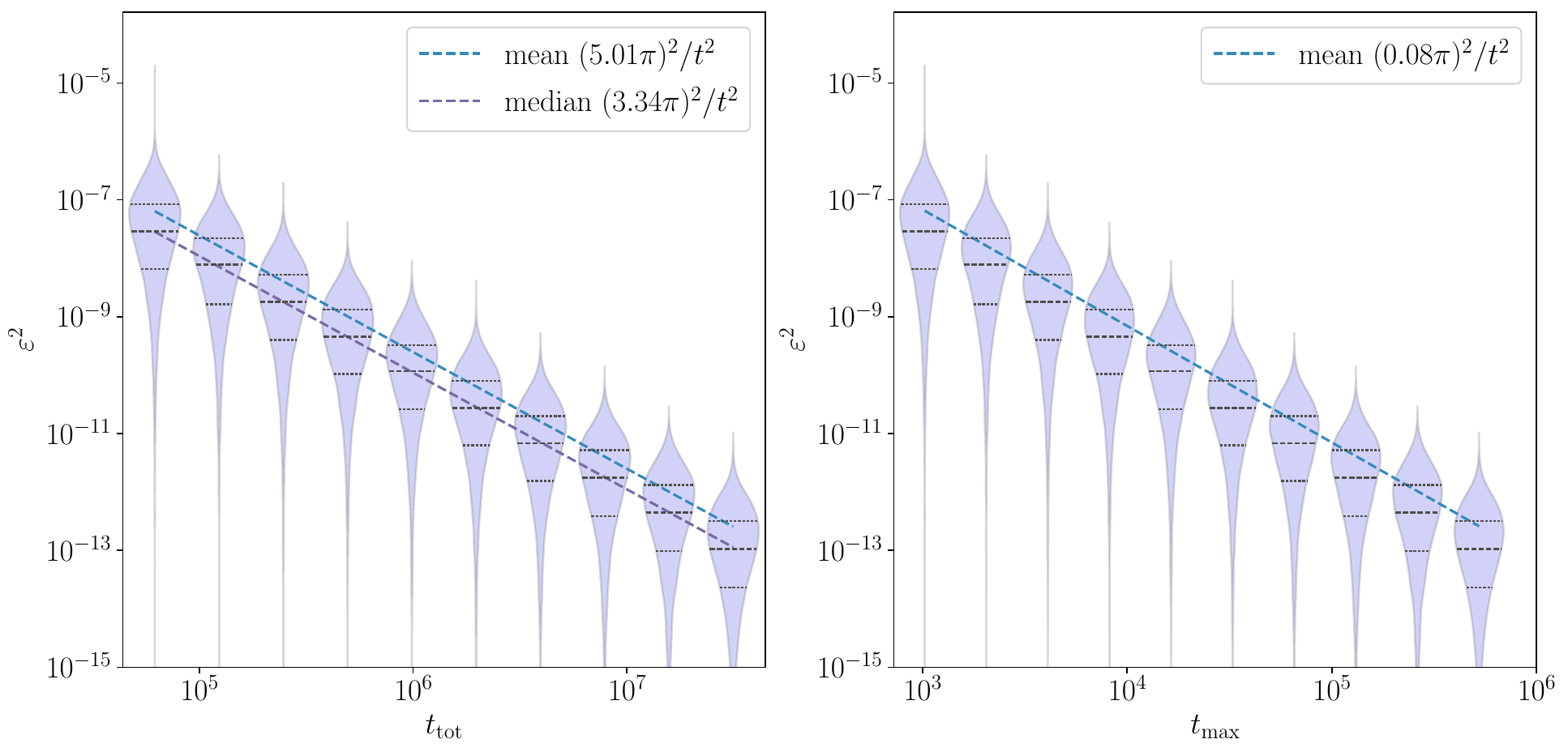}\\

    \includegraphics[width=.75\linewidth]{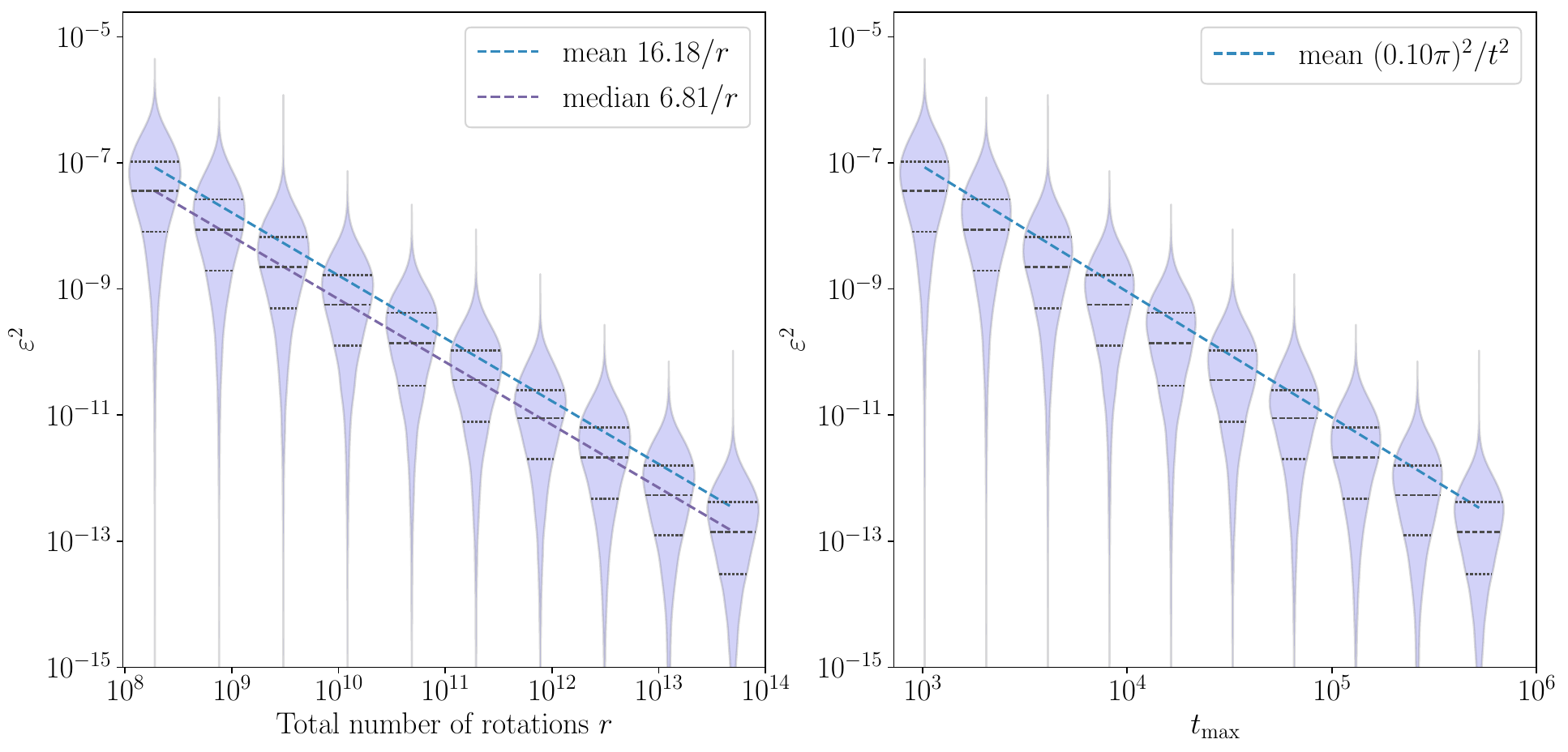}

    \caption{Estimation of the overhead of robust phase estimation. Distribution of the mean square error, $10^4$ data points for each value of $M$. The top panels shows scaling of mean square error with total time evolution and maximal time evolution $t_{\max} = 2^M$, using $N_m = 11 + 4(M - m)$ samples in round $m$.
        The bottom panels show scaling of mean square error with total number of rotations according to \cref{eq:total rotations random} and maximal time evolution $t_{\max} = 2^M$ simulated, using depth $r_m = 2^{2m+1}$ and $N_m = \ceil{e(11 + 4(M - n))}$ samples in round $m$.}
    \label{fig:rpe empirical scaling}
\end{figure}

We now discuss the total cost when using randomized product formulas.
For RTE, using $\tau = 2^{-m-1}$ and $r = 2^{2m+1}$ in the $m$-th round requires a sampling overhead of $e$.
We should now take $N_m$ a factor $e$ larger than in the deterministic case.
For our numerical estimates, we take $N_m = e(11 + 4(M-m))$.
The total number of rotations is given to good approximation by
\begin{align}\label{eq:total cost rpe randomized}
    \begin{split}
        r_{\tot} & = \sum_{m=0}^M 2N_m 2^{2m+1} = \sum_{m=0}^M 2(N_M + \sloperpe(M-m))2^{2m+1} \leq 16\left(\frac{N_M}{3} + \frac{\sloperpe}{9}\right) 2^{2M} \\
                 & = 16\left(\frac{N_M}{3} + \frac{\sloperpe}{9}\right)t_{\max}^2 = 8\left(\frac{N_M}{3} + \frac{\sloperpe}{9}\right)r_{\max} \, .
    \end{split}
\end{align}
In \cref{fig:rpe empirical scaling} (b) we numerically show that this indeed leads to a scaling $\eps^2 \approx C_{\tot}' / r_{\tot}$ and $\eps^2 \approx (C_{\max}' \pi)^2 / r_{\max}$ where $C_{\tot}' \approx 16$ and $C_{\max}' \approx 0.2$ (note that the maximal number of rotations is $2t_{\max}^2$).
As discussed at the end of the previous subsection, for qDRIFT we save an additional factor of 2 for both $r_{\max}$ and $r_{\tot}$, giving a total cost of approximately $r_{\tot} = 8 \eps^{-2}$ for performing phase estimation with root mean square error $\eps$.

Concretely, for the example of FeMoco discussed in the introduction, we have $\lambda = 405$ and $\eps = 0.0016$. We should estimate the ground state energy of the normalized Hamiltonian to precision $\lambda^{-1}\eps$.
We choose $\xi = 0.1$, giving $M = \ceil{\log_2(\xi \lambda \eps^{-1})} = 15$. For the last round, we do not necessarily need to take a power of two, but we can choose $2^{M-1} < K_M \leq 2^M$ and estimate $K_M E_0$. We let $K_M = \ceil{\xi \lambda \eps^{-1}}$, then estimating $K_M E_0$ to precision $\xi$ gives an estimate of $E_0$ to precision $\eps$.
This gives circuits with at most $K_M^2 = 6.4 \times 10^8$ Pauli rotations.

\section{The electronic structure problem}\label{sec:electronic structure problem}
The motivation of this work is to use quantum computing for Hamiltonians which arise in electronic structure problems in the Born-Oppenheimer approximation, meaning that we fix the location of nuclei, and treat the electron wavefunctions as the quantum degrees of freedom.
For completeness, we describe the set-up, relevant to the numerical examples, in some detail here.
% This is the most relevant quantum mechanical problem in quantum chemistry and material sciences.
We work in second quantization and we assume that we have chosen a set of orthonormal spatial orbitals labeled by an index $p \in [N]$ corresponding to a function $\phi_p \in L^2(\RR^3)$.
Additionally, there is a spin degree of freedom, so we end up with a set of spin-orbitals $\phi_{p,\alpha}$ for $p \in [N]$ and $\alpha = \uparrow, \downarrow$.
For convenience, we will assume that the spatial orbitals are real-valued functions (but this is not essential to the arguments).
The spin-orbitals span a single-particle space $\HH_{1} \cong \CC^{2N}$.
The Hilbert space we work with is the corresponding Fock space
\begin{align}
    \HH = \bigoplus_{n=0}^{2N} \HH_n \qquad \HH_n = \HH_{1}^{\wedge n}.
\end{align}
We denote the fermionic annihilation operator associated to orbital $p$ and spin $\alpha$ by $a_{p\sigma}$.
The size of the problem is determined by $N$, the number of orbitals, and $n$, the number of electrons.

Given a choice of orbitals, the Hamiltonian of the electronic structure problem is given by
\begin{align}\label{eq:electronic hamiltonian appendix}
    H = T + V = \sum_{pq}\sum_{\sigma} h_{pq} a_{p\sigma}^\dagger a_{q\sigma} + \frac12 \sum_{pqrs} \sum_{\sigma\tau} h_{pqrs} a_{p\sigma}^\dagger a_{r\tau}^\dagger a_{s\tau} a_{q\sigma}
\end{align}
where the coefficients $h_{pq}$ and $h_{pqrs}$ result from performing appropriate integrals involving the orbital functions.
Here and in the remainder of this work we use the convention that the spatial orbitals are labeled by $p,q,r,s,\dots$ and spins by $\sigma,\tau,\dots$.
In \cref{eq:electronic hamiltonian appendix}, the first term $T$ represents the kinetic energy and the interaction between the nuclei and the electrons
\begin{align}\label{eq:1-electron integral}
    h_{pq} = \int_{\RR^3} \phi_p(\vec{r})\left( -\frac12 \nabla^2 - \sum_{i = 1}^m \frac{Z_i}{\norm{\vec{r} - \vec{R}_i}}\right) \phi_q(\vec{r}) \der \vec{r}.
\end{align}
Here the $\vec{R}_i$ and $Z_i$ are the (fixed) positions and the nuclear charge of the $m$ nuclei.
The last term corresponds to the electron-electron Coulomb interaction $V$
\begin{align}\label{eq:2-electron integral}
    h_{pqrs} = \int_{\RR^3} \int_{\RR^3} \frac{\phi_p(\vec{r}_1) \phi_q(\vec{r}_1) \phi_r(\vec{r}_2) \phi_s(\vec{r}_2)}{\norm{\vec{r}_1 - \vec{r}_2}} \der \vec{r}_1 \der \vec{r}_2.
\end{align}
The integrals depend on the choice of the basis of orbitals. Any unitary $u \in U(N)$ defines a basis change on the spatial part of the single-particle space, giving rise to a new set of orbitals by $\tilde \phi_p = \sum_q u_{pq} \phi_{q}$.
It defines a particle-number and spin conserving unitary $U$ on the Hilbert space $\HH$ by defining a new set of fermionic creation operators $\tilde a_p^\dagger = \sum_q u_{pq} a_q^\dagger$.
Typically, the set of orbitals used to define \cref{eq:electronic hamiltonian appendix} is derived from a mean-field calculation, giving rise to a set of \emph{canonical orbitals}. Often it is useful to perform a basis change so the orbitals are \emph{localized} functions with exponentially decaying tails.
In that case, from \cref{eq:2-electron integral} we see that the density-density terms with $p=q$ and $r=s$ lead to contributions which decay polynomially in the distance between the spatial locations of the orbitals $p$ and $r$, whereas for $p \neq q$ we have an exponential decay in the distance between the spatial locations of orbitals $p$ and $q$.
In \cref{fig:weight distribution} we indeed see a regime of terms with power law decay, and a regime of terms with exponential decay.
For large spatially extended systems, one finds that many of the $h_{pqrs}$ are extremely small (say, less than $10^{-10}$) and can be safely ignored. In that case, the number of relevant terms scales as $\bigO(N^2)$.
Details on electronic structure formalism and computations can be found in \cite{helgakerMolecularElectronicStructureTheory2014}; see in particular section 9.12 for a detailed discussion of the scaling of the size of integrals $h_{pqrs}$.

The Hamiltonian $H$ is mapped to $2N$ qubits by a fermion-to-qubit mapping.
Various options are possible, but the key feature is that the terms in $H$ get mapped to Pauli operators.
Generally, one defines the Majorana fermionic operators by
\begin{align}
    \gamma_{p\sigma,0} = a_{p\sigma}^\dagger + a_{p\sigma}, \qquad \gamma_{p\sigma,1} = i(a_{p\sigma}^\dagger - a_{p\sigma})
\end{align}
and each $\gamma_{p\sigma,\alpha}$ should be mapped to a Pauli operator in a way that preserves the anticommutation relation $\{\gamma_{p\sigma,\alpha}, \gamma_{q\tau,\beta}\} = 0$ if $p\sigma,\alpha \neq q\tau,\beta$ (for example, using the Jordan-Wigner mapping, or Bravyi-Kitaev mapping).
Ignoring a component which is proportional to the identity operator, the Hamiltonian can be written as
\begin{align}\label{eq:hamiltonian majorana}
    H = \frac{i}{2} \sum_{pq} \sum_{\sigma} \left(h_{pq} + \sum_r h_{pqrr} - \frac12 \sum_r h_{prrq} \right) \gamma_{p\sigma,0} \gamma_{q\sigma,1} - \frac18 \sum_{pqrs} \sum_{\sigma\tau}h_{pqrs} \gamma_{p\sigma,0} \gamma_{q\sigma,1}\gamma_{r\tau,0} \gamma_{s\tau,1}
\end{align}
A minimal representation in which each Majorana operator products appears at most once has been derived in \cite{koridonOrbitalTransformationsReduce2021, mitaraiPerturbationTheoryQuantum2023},
which gives a representation of $H$
\begin{align}\label{eq:pauli rep appendix}
    H = \sum_{l=1}^L h_l P_l
\end{align}
where the $P_l$ are Pauli operators and the $h_l$ are real coefficients.

\subsection{Factorization of electronic structure Hamiltonians} \label{sec: hamiltonian factorization}
The Coulomb interaction (i.e. the two-body term) is the part of the Hamiltonian that is the main contribution to the complexity of quantum simulation of $H$ in \cref{eq:electronic hamiltonian appendix}. For example, simply Trotterizing over the sum as given in the description of the Hamiltonian, the one-body interactions give $\bigO(N^2)$ terms, while the two-body interaction contributes $\bigO(N^4)$ terms.
Moreover, if the Hamiltonian \emph{only} has one-body terms, the Hamiltonian is quadratic in the creation and annihilation operators, which may be simulated efficiently classically.
Employing anticommutation relations the electronic structure Hamiltonian \ref{eq:electronic hamiltonian appendix} can be written as
\begin{align}\label{eq:electronic hamiltonian reordered}
    H = \sum_{pq}\sum_{\sigma} \underbrace{(h_{pq} - \frac12 \sum_{r} h_{prrq})}_{h'_{pq}} a_{p\sigma}^\dagger a_{s\sigma} + \frac12 \sum_{pqrs}\sum_{\sigma\tau} h_{pqrs} a_{p\sigma}^\dagger a_{q\sigma} a_{r\tau}^\dagger a_{s\tau} = T + V.
\end{align}
It is useful to look for simpler representations of the Coulomb interaction $V$ in order to simplify both product formulas and block encodings of the Hamiltonian. A prominent example is \emph{double factorization} \cite{motta2021low,berry2019qubitization,huggins2021efficient,von2021quantum,cohn2021quantum,kim2022fault,rocca2024reducing}.

The idea is to first consider $h$ as a positive $N^2 \times N^2$ real matrix (the positivity follows from \cref{eq:2-electron integral}), and decompose it as
\begin{align}\label{eq:first factorization}
    h_{pqrs} = \sum_{j=1}^{L} (\mathcal{L}_{pq}^{(j)})^T \mathcal{L}_{rs}^{(j)}
\end{align}
where $\mathcal{L}^{(j)}$ is a real vector of dimension $N^2$, with entries labeled by orbital pairs $pq$. This is the `first factorization.'
The Coulomb tensor satisfies the symmetries $h_{pqrs} = h_{pqsr} = h_{qprs} = h_{qpsr} = h_{rspq} = h_{rsqp} = h_{srpq} = h_{srqp}$.
These symmetries imply that we may take the $\mathcal{L}^{(j)}$ to be real and such that $\mathcal{L}^{(j)}_{pq} = \mathcal{L}^{(j)}_{qp}$.
Therefore, we can reinterpret the $\mathcal{L}^{(j)}$ as \emph{symmetric} $N \times N$ matrices.
One may now apply a `second factorization' and diagonalize each $\mathcal{L}^{(j)}$ using an orthogonal transformation.
That is, we find $u^{(j)} \in O(N)$ and diagonal $N \times N$ matrices $\Lambda^{(j)}$ with real numbers $\lambda^{(j)}_k$ on the diagonal, such that
\begin{align}\label{eq:eigenvalues L}
    \mathcal{L}^{(j)} = (u^{(j)})^T \Lambda^{(j)} u^{(j)}.
\end{align}
% Together, this yields
% \begin{align}
%     h_{pqrs} = \sum_{j=1}^L \mathcal{L}^{(j)}_{pq} \mathcal{L}^{(j)}_{rs} = \sum_{j=1}^L \sum_{k,l=1}^{\rho_j} u^{(j)}_{kp} \lambda^{(j)}_k u^{(j)}_{kq} \, u^{(j)}_{lr} \lambda^{(j)}_l u^{(j)}_{ls}
% \end{align}

Inserting this back into the Coulomb interaction term of the Hamiltonian, one finds
\begin{align}
    V & = \frac12 \sum_{j=1}^L \left(\sum_{pq} \sum_{\sigma} \mathcal{L}_{pq}^{(j)} a_{p\sigma}^\dagger a_{q\sigma} \right)^2 = \frac12 \sum_{j=1}^L \left(\sum_{k=1}^{\rho_j} \lambda^{(j)}_k \sum_{pq} \sum_{\sigma} u^{(j)}_{kp} a_{p\sigma}^\dagger a_{q\sigma} u^{(j)}_{kq} \right)^2
\end{align}
The $N\times N$ matrices $u^{(j)}$ define a change of orbital basis, and they give rise to a  corresponding orbital rotation $U^{(j)}$ on the full Hilbert space.
The fermionic annihilation operators in the changed basis $\phi^{(j)}_k = \sum_p u^{(j)}_{kp} \phi_p$ are given by
\begin{align}
    a^{(j)}_{k\sigma} = \sum_p u^{(j)}_{kp} a_{p\sigma}
\end{align}
and the associated number operators are $n^{(j)}_{k\sigma} = (a^{(j)}_{k\sigma})^\dagger a^{(j)}_{k\sigma}$.
Then,
\begin{align}\label{eq:double fact V}
    V & = \frac12 \sum_{j=1}^L \left(\sum_{k=1}^{\rho_j} \sum_{\sigma}\lambda^{(j)}_k n^{(j)}_{k\sigma} \right)^2 = \sum_{j=1}^L (U^{(j)})^T V^{(j)} U^{(j)}, \qquad \text{ for } \quad  V^{(j)} = \frac12 \left(\sum_{k,\sigma} \lambda^{(j)}_k n_{k\sigma} \right)^2.
\end{align}
The main reason that this decomposition is useful is that $h$ is well-approximated by a \emph{low-rank} matrix.
In the first factorization in \cref{eq:first factorization} we may truncate the sum to contain only $L = \tilde \bigO(N)$ terms, instead of the maximal number of $N^2$.
Secondly, numerical results suggest that in the limit of large $N$, for spatially extended systems, one can truncate $\rho_j$, the number of terms in the second factorization, to only $\tilde \bigO(1)$ terms instead of $N$ \cite{peng2017highly,motta2019efficient}, although this only becomes relevant for large systems.
Note that in this construction the choice of the $\mathcal{L}^{(j)}$ is not unique. Common choices are an eigenvalue decomposition or a Cholesky decomposition, but there is significant flexibility.
There are various methods which aim to optimize the factorization in order to reduce the total circuit depth \cite{cohn2021quantum,oumarou2022accelerating,rocca2024reducing}.
We discuss the associated gate counts for Trotterization using this low rank factorization in \cref{sec:gate count}.

\subsection{Benchmark systems}

For FeMoco we use the model Hamiltonian first described by Reiher et al. \cite{reiher2017elucidating}. There is an alternative widely used choice of active space given by \cite{li2019electronic} that is believed to capture the electronic structure more accurately, although even larger active spaces may be needed \cite{morchen_classification_2024}. For our resource estimates we stick to the Hamiltonian from \cite{reiher2017elucidating}, for facilitating comparison with deterministic product formulas (which are only available from \cite{reiher2017elucidating}). For the other methods, the active space from \cite{li2019electronic} shows a qualitatively similar picture.
The data specifying the Hamiltonian coefficients for both the hydrogen chain and for FeMoco is taken from \cite{lee2021even}. For the hydrogen chain we use a STO-6G minimal basis set and set the distance between the hydrogen atoms at 1.4 Bohr. The number of electrons is $n = N_H$, the number of spin-orbitals is $2N_H$.
Additionally, we use a collection of small molecules with varying active space sizes to estimate the Trotter error.
These molecules are alanine \ch{C3H7NO2}, the chromium dimer \ch{Cr2}, an iron-sulfur cluster \ch{Fe2S2}, hydrogen peroxide \ch{H2O2}, imidazole \ch{C3H4N2}, naphtalene \ch{C10H8}, nitrous acid \ch{HNO2} and sulfurous acid \ch{H2SO3} at equilibrium geometries.
The active space sizes vary from four to nine orbitals, and the active orbitals are chosen to be the closest ones to the HOMO-LUMO gap.
Furthermore, for each active space Hamiltonian generated this way, the same $\lambda$ optimization as described below is performed.

\subsection{Optimizing the \texorpdfstring{weight $\lambda$}{weight}}
As discussed in \cref{sec:optimize hamiltonian}, it is useful to optimize the choice of representation to have small weight $\lambda$.
Firstly, it is clear that the coefficients $h_{pq}$ and $h_{pqrs}$ defined by \cref{eq:1-electron integral} and \cref{eq:2-electron integral} depend on the choice of orbital basis.
Typically, one starts from a set of canonical orbitals from a mean-field calculation, which determine the active space for the problem.
However, given the subspace, one is free to choose a different basis, for which is the representation of $H$ has desirable properties, such as a small weight $\lambda$.
In \cite{koridonOrbitalTransformationsReduce2021} the effect of different basis changes on the value of $\lambda$ was explored, with as conclusion that orbital localization (i.e. choosing an orbital basis for which the orbital functions are spatially localized) decreases the value of $\lambda$ and that further reduction is possible by minimizing over the basis change.
We perform such optimization for all benchmark systems, using gradient descent over the basis change unitaries to minimize the weight $\lambda$.
The problem is not convex and needs a good starting point. One option is to use localized orbitals.
From our numerics, we observe that this choice can be a local minimum for $\lambda$.
A good alternative is to take the largest term in a Cholesky decomposition of the Coulomb tensor (cf. \cref{eq:first factorization}) $\mathcal L_{pq}^{(1)}$.
This defines an $N \times N$ symmetric matrix, and we can consider the orbital basis which diagonalizes this matrix, as in \cref{eq:eigenvalues L}. This is cheap to compute, already leads to significantly reduced weight $\lambda$, and forms a good Ansatz to start the optimization.

\begin{figure}[ht]
    \centering
    \includegraphics[width=0.4\linewidth]{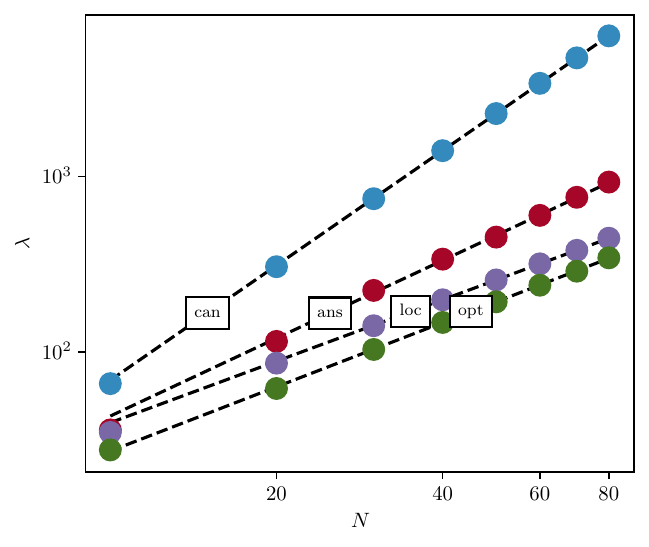}
    \caption{The weight $\lambda$ of the Hamiltonian depends on its decomposition. Here we show the scaling of $\lambda$ for the hydrogen chain with system size $N_{\text{H}}$, for canonical orbitals, localized orbitals, an ansatz based on the eigenvectors of the largest Cholesky vector, and an optimized representation from gradient descent over the space of orbital basis changes.}
    \label{fig:one-norm-H-chain}
\end{figure}

The Hamiltonian has a particle-number symmetry. That is, $H$ commutes with the electron number operator $\hat n = \sum_{p,\sigma} a_{p\sigma}^\dagger a_{p\sigma}$, and we work in a subspace of fixed electron number $n$, which is set by the initial state.
This feature can be exploited to define a different Hamiltonian, which equals $H$ on the subspace of $n$-particles (possibly up to an additive constant factor) but with lower weight.
This has been used in various ways, in particular to optimize the cost of LCU implementations of $H$ \cite{loaiza2023reducing,rocca2024reducing,caesura2025faster}.
For this work we use a variant of the method of Ref.~\cite{loaiza2023block}.
Let $\hat{n} = \sum_{p\sigma} n_{p\sigma}$ be the number operator, and let $n$ be the number of electrons for the ground state of $H$.
Then we can also apply phase estimation to $H + F(\hat{n} - n) + (\hat{n} - n)F$ for any self-adjoint operator $F$, since this has the same spectrum when restricted to the $n$-particle subspace.
In particular, one can take $F$ to be a one-body operator, and use this to reduce the weight of the two-body terms in $H$.
We perform a joint numerical optimization, minimizing $\lambda$, over the choice of the symmetry shift and the choice of orbital basis, using the L-BFGS-B quasi-Newton optimization method.
For the optimization, we restrict to an Ansatz where we decompose the Coulomb operator as in \cref{eq:double fact V} and shift
\begin{align}
    V^{(j)} \to \frac12 \left(\sum_{k\sigma} \lambda_k^{(j)} n_{k\sigma} + f_j(\hat{n} - n)\right)^2
\end{align}
and optimize the parameters $f_j$.
We show the resulting scaling of $\lambda$ with $N$ for the hydrogen chain in \cref{fig:one-norm-H-chain}. Localized and optimized orbitals bases give a weight scaling as $\lambda = \bigO(N^{1.2})$.

\section{Trotter error for electronic structure Hamiltonians}\label{sec:trotter appendix}

In this section we report in detail on our numerical analysis of Trotter error.
We briefly recall our motivation for doing so.
There are many ways to analytically bound the Trotter error. The tightest bounds are based on nested commutators.
Nevertheless, even if these have the correct scaling, they are still loose by some constant factor \cite{childsTheoryTrotterError2021}.
Additionally, if the Hamiltonian has $L$ terms, the commutator bound for the $p$-th order in principle may require as many as $\bigO(L^{p+1})$ operations, which is not feasible to compute exactly in practice, even for $p=2$, for molecular electronic structure Hamiltonians with many orbitals.

For these reasons, we take a numerical approach also taken in \cite{babbush2015chemical,reiher2017elucidating} and numerically compute the Trotter error for small systems, to then extrapolate to larger systems.
Recall that if $H$ is the Hamiltonian with ground state $E_0$, $\tilde{H}$ the effective Hamiltonian of the Trotterized time evolution with ground state $\tilde{E}_0$, $U(\delta)$ is the exact time evolution for time $\delta$ and $\trot_p(\delta)$ is a Trotter step with time $\delta$, then we have
\begin{align}
    \abs{E_0 - \tilde{E_0}} \leq C_{\gs} \delta^p \quad \text{ and } \qquad \norm{U(\delta) - \trot_p(\delta)}_{\op} \leq C_{\op} \delta^{p+1}
\end{align}
for an order $p$ product formula. Here, we let $C_{\gs}$ and $C_{\op}$ be the smallest constants for which the bound holds.
Crucially, these values depend on the Hamiltonian, as well as the system size.
The cost of the product formula scales with $C_{\gs}^{1/p}$.

We estimate $C_{\gs}$ and $C_{\op}$ by numerically computing $U(\delta)$ and $\trot_p(\delta)$ for a range of values of $\delta$, and fit a power law to the resulting ground state energy and operator norm errors.
We find that for a wide range of $\delta$ we get an excellent fit (even though in principle for large $\delta$ there will be higher order corrections which are dominant), and use this to determine $C_{\gs}$ and $C_{\op}$, leading to \cref{fig:trotter-error-proxy} in the main text.
For the numerical examples, we use a symmetry-preserving Bravyi-Kitaev mapping \cite{bravyi2017tapering}. The Trotter error depends on the ordering of the terms; we choose a lexicographic ordering. See \cite{tranter2018comparison,tranter2019ordering} for a numerical study of the effect of the choice of fermion-to-qubit mapping and ordering on the Trotter error.

In \cref{fig:gs op comparison} we compare $C_{\gs}$ and $C_{\op}$ for the small molecule benchmark set and the hydrogen chain for $p = 2$.
We observe both correlate well (especially $C_{\op}$) with $\lambda$.
Additionally, we also check the correlation with the commutator
\begin{align}
    \alpha = \sum_{l_1,l_2} \norm{[H_{l_1}, H_{l_2}]}_{\op} = \sum_{l_1, l_2} \abs{h_{l_1} h_{l_2}} \, \norm{[P_{l_1}, P_{l_2}]}
    \label{eq:alpha-def}
\end{align}
and find a comparable correlation.

\begin{figure}[ht]
    \centering
    \includegraphics[width=0.45\linewidth]{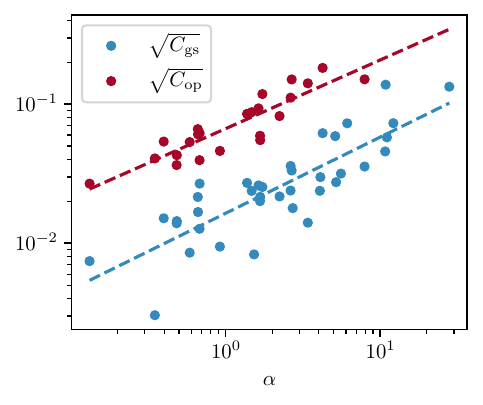} \hspace{0.5cm}
    \includegraphics[width=0.45\linewidth]{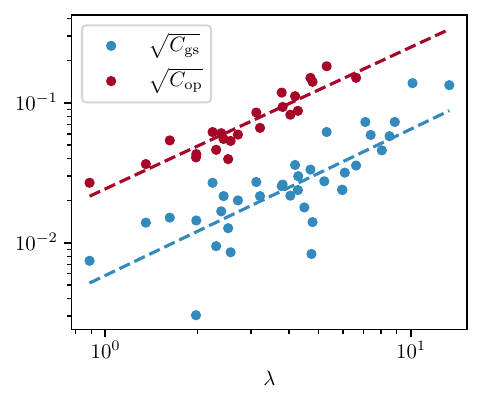}
    \caption{The left plot shows the Trotter constants for the operator norm $C_{\op}$ and the ground state energy $C_{\gs}$, for the benchmark set of small molecules and its correlation with $\alpha$ (defined in \cref{eq:alpha-def}.
        The right plot shows the same data, compared to the weights $\lambda$.}
    \label{fig:gs op comparison}
\end{figure}

One can also use higher-order product formulas, which leads to a better scaling in $t$, the time to be simulated.
However, at fixed finite time $t$ it is not clear whether this gives an improvement, since increasing the order also increases the number of stages $N_{\stage}$ of the product formula, and $N_{\stage}$ increases exponentially with the order $p$ for the Trotter-Suzuki product formulas.
In \cref{fig:order comparison} we compare the second and fourth order product formulas for the same benchmark systems.
Here we plot the factor that determines the difference in cost, which is $N_{\stage} C_{\gs}^{1/p}(p, \{H_l\}) \eps^{-1/p}$. The number of stages is $N_{\stage} = 2, 10$ for $p=2,4$ and we take $\eps = 0.0016$.
We do not observe a consistent improvement from using the fourth order method.
Of course, for sufficiently small accuracy $\eps$, the higher order method will eventually be superior.
However, for this work we conclude that it suffices to use second order (which has the added benefit that it does need to evolve backwards in time, which is beneficial for the partially randomized product formula).

\begin{figure}[ht]
    \centering
    \includegraphics[width=0.45\linewidth]{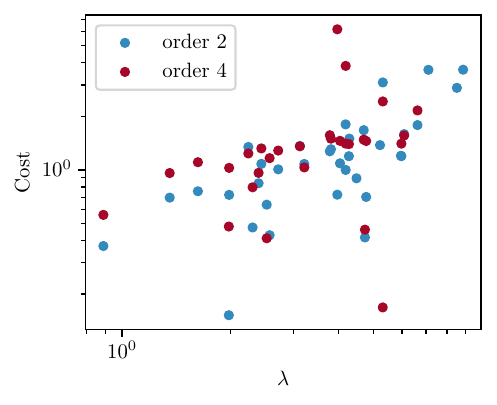}
    \caption{This figure shows the effect of the choice of order $p$ of the product formula. We plot $\text{Cost} = P C_{\gs}^{1/p} \eps^{-1/p}$ for $p=2$ and $p=4$ Trotter-Suzuki product formulas and $\eps = 0.0016$ for the benchmark set of small molecules and for the hydrogen chain.}
    \label{fig:order comparison}
\end{figure}

We also briefly comment on the first order method: for the ground state energy it has \emph{second order scaling}. This is an easy consequence of the fact that in perturbation theory, the first order correction $\bra{\psi_0} [H_{l_1}, H_{l_2}] \ket{\psi}$ vanishes \cite{yi2022spectral}. However, since the second order method only has two stages, the additional cost is not very high.

\begin{figure}[ht]
    \centering
    \includegraphics[width=0.45\linewidth]{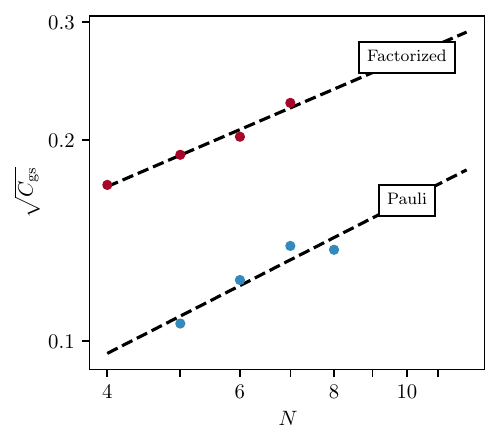} \hspace{0.5cm}
    \caption{For order $p=2$ the Trotter ground state error constant $\sqrt{C_{\gs}}$ (defined in \cref{eq:Cgs-def}) for the hydrogen chain with $N$ atoms, both for the Pauli representation and the factorized representation.}
    \label{fig:trotter error hchain}
\end{figure}

\subsection{Trotter error for partial randomization}
For partially randomized product formulas, given
\begin{align}
    H = \sum_l h_l P_l
\end{align}
one groups together a subset of terms into $H_R$.
The Trotterization error for the partially randomized method is the Trotterization error of dividing the Hamiltonian into terms $h_l P_l$ and keeping $H_R$ as a single term.
For the performance of the partially random method it is clearly of interest how the Trotter error depends on the choice of $H_R$.
There are two intuitions of what could happen to the Trotter error. Firstly, if many terms are in $H_R$, we are discretizing less, and the Trotter error should become smaller. In particular, if all terms are in $H_R$, there is no Trotter error whatsoever.
On the other hand, the Trotter error decreases if one takes very small steps. Grouping together many small terms into one larger term $H_R$ means performing a relatively large step in the direction of $H_R$, which could \emph{increase} the Trotter error.

In \cref{fig:trotter error partially random} we show numerically that for $\lambda_R$ relatively small (the regime which is of most interest for partial randomization), the Trotter error remains essentially unchanged. For intermediate values of $\lambda_R$, the Trotter error in fact increases, but only moderately so. When $\lambda_R$ tends to $\lambda$, the Trotter error decreases significantly.

These findings justify our choice to optimize the partially randomized methods using the Trotter error estimate for the fully deterministic method.
Indeed, the regime where $\lambda_R$ is a large fraction of $\lambda$ will not lead to significant speed-up over the fully randomized method (so it does not help much that the Trotter error is smaller), and when $\lambda_R$ is reasonably small, the Trotter error of the fully deterministic product formula is a very good approximation.

\begin{figure}
    \centering
    \includegraphics[width=0.5\linewidth]{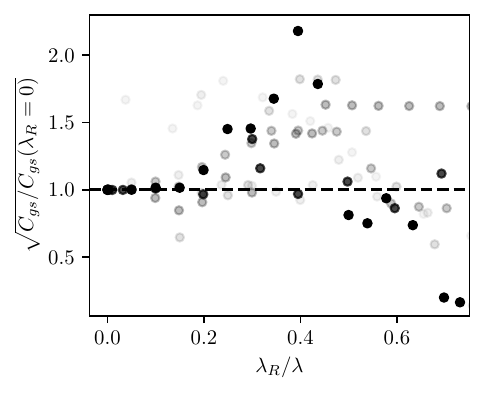}
    \caption{Trotter error as a function of the fraction of the randomized part for different active spaces of alanine and \ch{Fe2S2}. The dashed line indicates the fully deterministic Trotter error, while the darkness of the data points corresponds to the size of the active space, with darker points corresponding to larger active space size.
The error is written in terms of $C_{\gs}$, defined in \cref{eq:Cgs-def}.
    }
    \label{fig:trotter error partially random}
\end{figure}

\subsection{Truncation error}
In the partially randomized method the small terms are implemented in a randomized fashion.
An alternative way to deal with small terms is to simply neglect them altogether, provided the resulting truncation error is smaller than the desired precision.
In \cref{fig:truncation error} we show the dependence of truncation errors on the fraction of remaining terms for the same set of model systems as used in the Trotter error analysis in \cref{sec:trotter appendix}.
For the Pauli decomposition %(without the bitwise Hamming weight phasing of the deterministic part)
we computed the optimal partition into randomized and deterministic parts as described in \cref{sec:gate counts partrand}.
We then computed the error in the ground state energy when keeping only the terms in the deterministic part.
Note that we do not consider errors due to Trotterization here.
\cref{fig:truncation error} shows that the deterministic terms typically make up only a fraction of 0.05 or less of the total number of terms.
Simply neglecting the other terms results in an error that is one to two orders of magnitude higher than the target precision.

In \cref{fig:truncation error scan} we compute exact truncation errors over the whole range of terms for various active space Hamiltonians of alanine.
Here we see that the minimal number of terms one could reach by neglecting small terms (such that the truncation error is not more than 0.0016) is about an order of magnitude higher than the number of deterministic terms in the partially randomized method.
However, still surprisingly many terms can be truncated (if the exact truncation error is known), here to a fraction of about 0.2. 
Note that in practice some of the error budget also needs to be assigned to the Trotterization and phase estimation error.
%Note that there is a large range where the error is only slightly smaller, and some of the error budget should be assigned to the Trotterization error.
%However, this is still about an order of magnitudes more terms than the number of deterministic terms, see \cref{fig:truncation error}.

Consequently, the partially randomized method achieves a reduction of the deterministically implemented terms that would not be possible by a simple truncation approach.
%However, it does mean that in practice one can make the (partially) random methods faster by truncating part of the tail, reducing the value of $\lambda_R$.
In general, a downside of trying to (deterministically) truncate as many terms as possible from the Hamiltonian, is that it is hard to know a priori to what point one can truncate without incurring significant error, compromising the accuracy guarantees of the energy approximation procedure.

Yet another alternative to deal with small terms in the Hamiltonian is to treat them perturbatively.
Quantum algorithms for second-order perturbation theory corrections have been proposed previously \cite{mitaraiPerturbationTheoryQuantum2023, guntherMoreQuantumChemistry2024}.
However, such an approach requires many samples, and for each sample the ground state of the truncated Hamiltonian needs to be prepared, and the inverse of the truncated Hamiltonian needs to be implemented.
This potentially negates any cost-advantage from truncating the Hamiltonian, which was the outcome of \cite{mitaraiPerturbationTheoryQuantum2023}.
A full analysis of a perturbation theory approach in our context is beyond the scope of this work.

\begin{figure}
    \centering
    \includegraphics[width=0.5\linewidth]{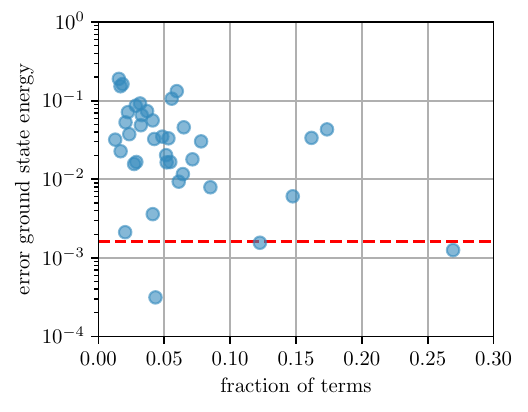}
    \caption{For the Hamiltonians from the benchmark set, we plot the error in the ground state energy as a function of the fraction of remaining terms after truncating the ones with smallest weight. We use the truncation which is optimal for the cost of the partially random product formula, so we compare the energies of $H$ and $H_D$.
    The red dashed line indicates chemical accuracy, $0.0016$.}
    \label{fig:truncation error}
\end{figure}

\begin{figure}
    \centering
    \includegraphics[width=0.5\linewidth]{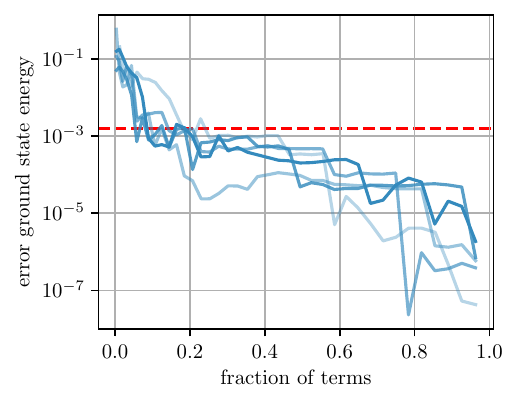}
    \caption{Error in the ground state energy as a function of the fraction of remaining terms after truncating the ones with smallest weight.
    The five blue lines correspond to the active space Hamiltonians of alanine with 4,5,6,7 and 8 orbitals (darker lines for larger active spaces).
    The red dashed line indicates chemical accuracy, $0.0016$.}
    \label{fig:truncation error scan}
\end{figure}

\section{Detailed cost estimates and compilation}\label{sec:gate count}

In this part of the Appendix we are giving a detailed account of how the gate counts are estimated for ground state energy estimation using the deterministic, randomized, or partially randomized simulation methods.
For each method we are estimating the number of gate $G$, where we can either measure the number of (non-Clifford) Toffoli gates, or the number of two-qubit gates.
We are given a Hamiltonian $H$, and wish to estimate the ground state energy to target precision $\eps$ (root mean square error).
In these estimates we assume access to an exact eigenstate of $H$, but as discussed this assumption can easily be relaxed.
The total cost depends on various parameters, some of which depend on $H$ and some of which can be optimized.
We give an overview of the relevant parameters in \cref{table:parameters qpe}.

\begin{table}
    \centering
    \begin{tabular}{|c|l|}
        \hline
        Parameter                      & Notation                                                                             \\ \hline
        \multicolumn{2}{|c|}{Phase estimation parameters}                                                                     \\ \hline
        $\eps$                         & Target precision for the energy.                                                     \\
        \hline
        $M$                            & Number of rounds in phase estimation.                                                \\
        $N_m$                          & Number of samples in round $m$, $N_m = N_M + D(M - m)$.                              \\
        \hline
        \multicolumn{2}{|c|}{Deterministic product formulas}                                                                  \\ \hline
        $L$                            & Number of terms in the Hamiltonian.                                                  \\
        $\eps_{\qpe}, \eps_{\trotter}$ & Dividing error budget between phase estimation and Trotter error.                    \\
        $\delta$                       & Trotter step size.                                                                   \\
        $p$                            & Order product formula.                                                               \\
        $N_{\stage}$                   & Number of stages product formula.                                                    \\
        $C_{\gs}$                      & Trotter error constant for the ground state energy.                                  \\
        $G_{\det}$                     & Number of elementary gates to perform one elementary evolution $\exp(-i\delta H_l)$. \\
        \hline
        \multicolumn{2}{|c|}{Randomized product formulas}                                                                     \\ \hline

        $\lambda$                      & Weight of the Hamiltonian.                                                           \\
        $\tau$                         & Step size                                                                            \\
        $\varphi$                      & Rotation angle $\varphi = \arctan(\tau)$.                                            \\
        $G_{\rand}$                    & Number of elementary gates to time evolve a single Pauli operator $P_l$.             \\
        $B$                            & Damping factor from the randomized product formula.                                  \\
        \hline
        \multicolumn{2}{|c|}{Partially randomized product formulas}                                                           \\ \hline

        $L_D$                          & Number of terms treated deterministically.                                           \\
        $\lambda_R$                    & Weight of the terms treated randomly.                                                \\
        $\kappa$                       & Scaling factor for the number of rotations for randomized product formula.           \\
        \hline
    \end{tabular}
    \caption{Relevant parameters for gate counts for of phase estimation of ground state energies using deterministic and (partially) randomized product formulas.}
    \label{table:parameters qpe}
\end{table}

\subsection{Gate counts for Trotterization}
The Hamiltonian has a decomposition $H = \sum_{l = 1}^L H_l$.
We use a product formula of order $p$.
For most of the resource estimates we use second order ($p = 2$).
As explained in \cref{sec:trotter}, we perform phase estimation on the Trotter unitary $\trot_{p}(\delta) = e^{-i\delta \tilde H} \approx e^{-i\delta H}$ for this task, computing the ground  energy $\tilde E_0$ of the effective Hamiltonian $\tilde H$.
The error between the ground states of $H$ and $\tilde H$ behaves as $\eps_{\trotter} = \abs{E_0 - \tilde E_0} \leq C_{\gs} \delta^{p}$.
Additionally, when applying phase estimation, we obtain $\tilde E_0$ with mean square error $\eps_{\qpe}$.
This error is unbiased and independent of the Trotter error, which means that in order get an estimate of precision $\eps$ for $E_0$ we should have
\begin{align}
    \eps^2 = \eps_{\qpe}^2 + \eps_{\trotter}^2.
    \label{eq:incoherent error trotter}
\end{align}
To reach precision $\eps_{\qpe}$ for $\tilde E_0$, we need to estimate the phase of $\trot_{p}(\delta)$ to precision $\eps_{\qpe} \delta$. We estimated in \cref{sec:rpe} that robust phase estimation needs on average $5\pi/(\eps_{\qpe}\delta)$ implementations of $\trot_{p}(\delta)$ to do so. % (see Appendix \ref{sec:appendix trotter}).
The number of rounds in the phase estimation is
\begin{align}\label{eq:max rounds}
    M = \left\lceil\log_2\left(\frac{0.08\pi}{\eps_{\qpe}\delta}\right)\right\rceil.
\end{align}

Let $G_{\det}$ be the average gate cost to implement a single evolution along the $H_l$, then we see that the total gate cost is given by
\begin{align}
    G = 5\pi N_{\stage}L G_{\det}\frac{C_{\gs}^{1/p}}{\eps_{\qpe}\eps_{\trotter}^{1/p}}.
\end{align}
Minimizing the cost under the error budget in \cref{eq:incoherent error trotter} results in $\eps_{\qpe} = \eps \sqrt{\frac{p}{p+1}}$, and trotter step size and optimal total cost
\begin{align}
    \delta & = \left(\frac{\eps}{C_{\gs}\sqrt{p+1}}\right)^{1/p}\, ,    \label{eq:trotter step size}      \\
    G      & = 5\pi N_{\stage}L G_{\det}\sqrt{\frac{(p+1)^{1+1/p}}{p}}\frac{C_{\gs}^{1/p}}{\eps^{1+1/p}}.
    \label{eq:trotter cost}
\end{align}
We can reduce the cost by factor 2 using the method in \cref{sec:halving time}.
At this point, the cost is agnostic of the gate set used.
For the Toffoli count we additionally need to approximate with a gate set of Cliffords and Toffolis, and consider the resulting synthesis error.
Here we follow the analysis of \cite{reiher2017elucidating,kivlichan2020improved}.
We decompose circuits into Cliffords plus arbitrary single-qubit rotations, which can be synthesized to error $\eps'$ at a cost of
\begin{align}
    G_{\text{rot}} = 1.14 \log_2\left(\frac{1}{\eps'}\right) + 9.2
\end{align}
T gates \cite{bocharov2015efficient}.
We convert from Toffoli to T gates using 1 Toffoli gate for every 2 T gates \cite{gidney2019efficient}.
Let $n_{\text{rot}}$ be the number of single qubit rotations per trotter stage, each approximated to error $\eps'$.
We then want to bound the synthesis error $\eps_{\synth}$, i.e. the difference in the ground state energies of the exact effective Hamiltonian and the effective Hamiltonian with imperfectly synthesized single-qubit rotations, denoted by $H'$.
One can bound \cite{reiher2017elucidating}
\begin{align}
    \norm{e^{-i\delta \tilde H} - e^{-i\delta H'}} \leq N_{\stage} n_{\rot}\eps',
\end{align}
so to first order in $\delta$ the ground state error is bounded by
\begin{align}
    \abs{\tilde E_0 - E_0'} \leq \norm{\tilde H - H'} \leq N_{\stage} n_{\rot}\eps'\delta^{-1}.
\end{align}
Consequently, to keep the synthesis error below $\eps_{\synth}$, we choose
\begin{align}
    \eps' = \frac{\delta \eps_{\synth}}{N_{\stage}n_{\rot}} = \frac{\eps_{\synth}}{N_{\stage}n_{\rot}}\left(\frac{\eps}{C_{\gs}}\sqrt{\frac{p}{p+1}}\right)^{1/p}.
\end{align}
We do not optimize the total error budget with respect to $\eps_{\synth}$, instead, when computing the Toffoli count, we allocate $\eps = 1.5\,$mH and $\eps_{\synth} = 0.1\,$mH.
Alternatively, if we measure two-qubit gates, we ignore the synthesis error, we let $\eps=1.6\,$mH and set $G_{\det}$ to be the number of 2-qubit gates.

So far, we have assumed a general Hamiltonian $H = \sum_l H_l$.
We now specialize to specific representations.

\subsubsection{Pauli Hamiltonian}
We start with the case where the Hamiltonian has a representation as a linear combination of Pauli operators,
where $H_l= h_l P_l$, as in \cref{eq:pauli rep appendix}.
The time evolution of a single term is then a Pauli rotation, $V_l(\delta h_l) = \exp(-i\delta h_l P_l)$.
This can be implemented using a single-qubit rotation and $2(\abs{P_l}-1)$ \cnot gates to implement.
Here we use $\abs{P_l}$ to denote the number of qubits on which $P_l$ acts nontrivially.

The number of Pauli rotations per stage simply equals the number of terms $L$.
The worst-case scaling of the number of terms with the number of orbitals is $L = \bigO(N^4)$, but in practice many terms are sufficiently small to be truncated (as discussed in \cref{sec:electronic structure problem}) and $H$ is much more sparse.
This depends on the choice of single orbital basis as well.
We find that minimizing the weight $\lambda$ through orbital rotations and symmetry shifts typically also leads to more sparse representations of $H$.
For our estimates, we truncate the smallest coefficients of $H$, such that the total weight $\sum_{l > L'} \abs{h_l}$ of truncated terms is at most a small threshold $10^{-3}$.

When counting two-qubit gates we can contract the innermost \cnot gates together with the single qubit rotation into a single 2-qubit gate, so the number of two-qubit gates per Pauli rotation is $2\abs{P_l}-3$.
Therefore the number of 2-qubit gates per operation is $G_{\det} = 2\abs{\bar{P}}-3$, where $\abs{\bar{P}}$ denotes the average Pauli weight.
Reductions in two-qubit gate cost are achieved by reducing $L$ or by reducing the average Pauli weight.
Again, minimizing $\lambda$ also reduces $L$.
The average Pauli weight depends on the choice of fermion-to-qubit mapping.
On $N$ spin orbitals the Bravyi-Kitaev mapping gives an average Pauli weight of $\bigO(\log(N))$, compared to $\bigO(N)$ for the Jordan-Wigner and parity mappings.
We use a symmetry-conserving Bravyi-Kitaev mapping, which has the additional advantage of reducing the number of qubits by 2 by taking into account spin and particle number conservation \cite{bravyi2017tapering}.
Yet another degree of freedom is the choice of term ordering the ordering in a Trotter stage, which can lead to significant cancellations of 2-qubit gates \cite{hastingsImprovingQuantumAlgorithms2014}.
We use the lexicographic ordering, which has a relatively large number of cancellations \cite{tranterComparisonBravyiKitaev2018}.

\subsubsection{Factorized Hamiltonian}\label{sec:gate count factorized}
In this case, we have a Hamiltonian of the form
\begin{align}
    H = T + \sum_{l=1}^L (U^{(l)} )^\dagger V^{(l)} U^{(l)}, \qquad V^{(l)} = \sum_{\sigma} \sum_{p = 1}^{\rho_l} \lambda_p^{(l)} n_{p\sigma}.
\end{align}
We first discuss the two-body terms.
In order to implement one stage in a Trotter step, we need to implement $L$ basis changes $U^{(l)} (U^{(l-1)})^\dagger$ (which can be combined into a single orbital basis change) and evolutions of the form $\exp(-i\delta V^{(l)})$.
This can be done \cite{motta2021low,kivlichanQuantumSimulationElectronic2018} using fermionic swap networks. The orbital basis rotation $U^{(l)}$ for both spin values can be implemented with $2(\binom{N}{2} - \binom{N-\rho_l}{2}) = \bigO(\rho_l N)$ Givens rotations (which are two-qubit gates, and can be implemented using two arbitrary single-qubit gates), and $\exp(-i\delta V^{(l)})$ can be implemented using $\binom{2\rho_l}{2}$ two-qubit gates (requiring one arbitrary single-qubit gate).
This gives a total number of single-qubit rotations
\begin{align}
    \sum_{l=1}^L 4\left(\binom{N}{2} -  \binom{N-\rho_l}{2}\right) + \binom{2\rho_l}{2} = \sum_{l=1}^N 4N(\rho_l - 1) + \rho_l.
\end{align}
The evolution along $T$ is a single orbital basis change (requiring $2 \binom{N}{2}$ Givens rotations).
For this approach to be useful, one has to truncate small $\lambda_p^{(l)}$. For \cref{fig:gate counts factorized hchain}, where we compute the cost of this approach for the hydrogen chain, we truncate such that for each $l$, if we truncate $\lambda_p^{(l)}$ for $p > p_l$, we have $\sum_{p > p_l} \abs{\lambda_p^{(l)}} \leq \eps'$, where we choose $\eps' = 10^{-3}$.

\subsection{Gate counts for randomized product formulas}\label{sec:gate count randomized}
For randomized product formulas, the Hamiltonian is assumed to be given in a Pauli representation. The cost is determined by the weight $\lambda$ and the target precision $\eps$.
We are given a Pauli Hamiltonian with some value of $\lambda$, and target precision $\eps$.
As discussed in \cref{sec:rpe}, we have to perform time simulation for $t = 2^m$ and take $N_m$ samples from a Hadamard test, with $m = 1, \dots, M$ for $M = \ceil{\log_2(\lambda \eps^{1})}$, which gives a number of Pauli rotations as in \cref{eq:total rotations randomized} as
\begin{align}
    r = \frac{16.3\lambda^2}{\eps^2}
\end{align}
Note that for qDRIFT the number of rotations may be reduced by a factor of two (see the discussion at the end of \cref{sec:rigorous bounds rpe}).
Counting two-qubit gates is similar to the previous section.
When counting Toffoli gates it turns out that it is helpful that all (or most) rotations are over the same angle, leading to a significantly smaller overhead in converting the number of rotations to the number of Toffoli gates.

The below approach combines ideas from \cite{gidney2018halving,kivlichan2020improved,sanders2020compilation,wan2022randomized} in the analysis of the Toffoli gate requirements.
To begin with, for qDRIFT, in every circuit \emph{all} rotations are over the same angle.
For RTE, the angle $\varphi_n$ depends on the order $n$ sampled of the term in the Taylor series.
For small $\tau$, the probability of sampling a term which is not of the lowest order scales with $\bigO(\tau^2)$.
When performing time simulation for time $t$ with $\tau = \Omega(t^{-1})$ and number of steps $r = \bigO(t^2)$ this means that we only expect $\bigO(1)$ Pauli rotations with an angle different than $\varphi_0$.
This means that we expect long sequences of equiangular Pauli rotations which dominate the cost.

Additionally, any individual Pauli operator $P_l$ with nonzero $h_l$ in the Hamiltonian commutes with \emph{most} of the other terms; it only anticommutes with $\bigO(N^3)$ out of $\bigO(N^4)$ terms.
For this reason, one expects relatively long sequences of commuting Pauli rotations in both qDRIFT and RTE, as observed in \cite{wan2022randomized} (if the terms were selected uniformly, one would have average length $\bigO(\sqrt{N})$).
Given such a sequence, $e^{i \varphi P_{l_1}} \cdots e^{i \varphi P_{l_K}}$ of length $K < 2N$, we may apply a Clifford to map
\begin{align}
    \{P_{l_1}, \dots, P_{l_K}\} \to \{Z_1, \dots, Z_K\}
\end{align}
to single-qubit $Z$-rotations on the first $K$ qubits.
We may now use Hamming weight phasing \cite{gidney2018halving,kivlichan2020improved,campbell2021early} to implement these rotations, as suggested for RTE in \cite{wan2022randomized}.

To reduce the cost even further, we can use a single catalyst phase state to implement the rotations \cite{kitaev1995quantum,sanders2020compilation,lee2021even}.
We may choose the step size in the randomized product formula to our liking; in particular we may choose such that the rotations are of angle $2^{-J} \pi$ for integer $J$.
Since we also want to choose the angle such that $\tau = \bigO(\eps \lambda^{-1})$, we take $J = \ceil{\log_2 (\lambda \pi^{-1}\eps^{-1})}$.
Let
\begin{align}
    \ket{\phi} = \frac{1}{2^{J/2}} (\ket{0} + e^{2 \pi i 2^{-J}} \ket{1})^{\ot J} = \frac{1}{2^{J/2}} \sum_{j=0}^{2^J - 1} e^{2\pi i j 2^{-J}} \ket{j}
\end{align}
be a phase state on $J$ qubits.
At the start of the computation, we initialize the ancilla phase estimation qubit, together with $J$ qubits, in the state
\begin{align}
    \frac{1}{\sqrt{2}} \left(\ket{0} \ket{+}^{\ot J} + \ket{1} \ket{\phi} \right) = (\id \ot \mathrm{QFT}) \frac{1}{\sqrt2}(\ket{0} \ket{0} + \ket{1} \ket{1})
\end{align}
where we label the basis for the $J$ qubits by $\{0,1,\dots,2^J -1\}$.
It (once) requires T gates or Toffoli gates to prepare this state (a constant cost which is negligible compared to the other contributions).
The phase state $\ket{\phi}$ is such that if we add another register with $W \leq J$, and implement addition $\ket{w} \ket{j} \mapsto \ket{w}\ket{j + w}$, then
\begin{align}
    \ket{w} \ket{\phi} \mapsto e^{-2\pi i w 2^{-J}} \ket{w} \ket{\phi}
\end{align}
We can use $J$ Toffoli gates, together with $J$ ancilla qubits to perform the addition.
The state $\ket{+}^{\ot J}$ is invariant under addition, so we get that
\begin{align}
    \alpha \ket{0} \ket{w} \ket{+}^{\ot J} + \beta \ket{1} \ket{w} \ket{\phi} & \mapsto \alpha \ket{0} \ket{w} \ket{+}^{\ot J} + \beta e^{-2\pi i w 2^{-J}} \ket{1} \ket{w} \ket{\phi}
\end{align}
To implement $\prod_{k=1}^K e^{i \pi 2^{-J} Z_k}$, on the first $K$ qubits we first compute the Hamming weight on a register of $ W = \ceil{\log_2(K + 1)}$ qubits.
We then apply the above trick to apply (controlled on the QPE ancilla) $e^{i w \pi 2^{-J}}$ to Hamming weight $w$, and we uncompute the Hamming weight.
Note that if $\ket{k}$ has Hamming weight $w$, then
\begin{align}
    \left(\prod_{k=1}^K e^{i \pi 2^{-J} Z_k}\right) \ket{k} = e^{- i \pi(2w - K)2^{-J}} \ket{k}
\end{align}
so we have implemented the correct transformation up to a global phase $e^{-i\pi K 2^{-J}}$ which we can account for in classical postprocessing of the Hadamard test outcomes.

Computing and uncomputing the Hamming weight requires at most $K - 1$ Toffoli gates applying the phase to the Hamming weight requires $J - 1$ Toffoli gates.
This means we use $1 + (J - 2)/K$ Toffoli gates per controlled Pauli rotation.
The adders require respectively $(K-1)$ and $(J-1)$ ancilla qubits, giving in total $K + 2J - 2$ ancilla qubits.

For example, for FeMoco, we may take $K = 10$ (choosing $K$ not too large, in order to make sure that with high probability we get sequences of commuting terms of length $K$) and $J = 17$, giving an average of $2.5$ Toffoli gates per controlled Pauli rotation, and using $42$ ancilla qubits.

\subsection{Gate counts for partially randomized product formulas}\label{sec:gate counts partrand}

We assume a given decomposition of $H$ in $L_D$ deterministic terms, and $H_R$ with weight $\lambda_R$
\begin{align}
    H = \sum_{l=1}^{L_D} H_l + H_R.
\end{align}
The Trotter error concerns the decomposition into $H$ as a sum of the $H_l$ and treating $H_R$ as a single term.
For our estimates we use the second order Trotter-Suzuki product formula (so $p = 2$ with $N_{\stage} = 2$).
It is straightforward to extend to higher orders.
Let $G_{\det}$ and $G_{\rand}$ be the (average) number of Toffolis used per unitary $\exp(-i\delta H_l)$ in the deterministic part, and per Pauli rotation for the randomized part of the Hamiltonian simulation.
These can be different.

We choose the parameters $\eps_{\qpe}$ and $\eps_{\trotter}$, the time step $\delta$ for the deterministic product formula, and the number of rounds $M$ for phase estimation as in \cref{eq:incoherent error trotter}, \cref{eq:max rounds} and \cref{eq:trotter step size}.
For the compilation of the randomized part, it is convenient to choose time step $\tau$ such that we can use Pauli rotations over angles $\pi 2^{-J}$.
For $m = 1,\dots, M$ we use $2N_m$ times the Hadamard tests with the partially randomized product formula to estimate (see \cref{lem:composite simulation})
\begin{align}
    B^{-1} \bra{\psi} \trot_{2}(\delta)^{2^m} \ket{\psi}.
\end{align}
As explained at the end of Section \ref{sec:rigorous bounds rpe}, for RTE alone it is optimal to choose a total of $r = 2t^2 = 2\lambda_R^2 \delta^2 2^{2m}$ Pauli rotations.
For the partially randomized method it is advantageous to allow for $r$ to vary, so we introduce a factor $\kappa\geq 0$ and set $r = \kappa\lambda_R^2 \delta^2 2^{2m}$, which gives a damping factor of $B \leq e^{1/\kappa}$.
The number of samples $N_m$ scales with the damping factor $B^2$, as discussed in \cref{sec:rpe}.
We use the observation in \cref{sec:halving time} to reduce the evolution time of the deterministic part by a factor of 2, so the deterministic product formula needs $2^{m-1}$ rather than $2^m$ applications of $\trot_2(\delta)$.
This gives a total number of gates
\begin{align}
    G & = \sum_{m=0}^M 2N_m \left(G_{\det} N_{\stage} L_D 2^{m-1} + G_{\rand} \kappa\lambda_R^2 \delta^2 2^{2m}\right).
\end{align}
Note that if we take $L_D = 0$, this reduces (up to possibly different rounding) to the cost of phase estimation to precision $\eps_{\qpe}$ using a randomized product formula, independent of the choice of $\delta$.
If we choose $N_m =  N_M + \sloperpe (M - m) = B^2(11 + 4(M - m)) = e^{2/\kappa} (11 + 4(M - m))$, the total number gates is bounded as in \cref{eq:upper bound time evolution} and \cref{eq:total cost rpe randomized} by
\begin{align}\label{eq:total gates partially random}
    \begin{split}
        G & \leq 2 G_{\det} N_{\stage} L_D (N_M + \sloperpe) 2^M + 8 G_{\rand} \left(\frac{N_M}{3} + \frac{\sloperpe}{9}\right) \kappa \lambda_R^2 \delta^2 2^{2M}                \\
          & = 30 G_{\det} N_{\stage} L_D e^{2/\kappa} \frac{0.1 \pi}{\delta \eps_{\qpe}} + \frac{296}{9} G_{\rand} \kappa e^{2/\kappa} \frac{(0.1 \pi\lambda_R)^2}{\eps_{\qpe}^2}
    \end{split}
\end{align}
The optimal choice of $\kappa$ solves a quadratic equation.
Again, we require the total error to be bounded by
\begin{align}
    \eps^2 \leq \eps_{\qpe}^2 + \eps_{\trotter}^2 = \eps_{\qpe}^2 + C_{\gs}^2 \delta^4
\end{align}
such that we substitute $\eps_{\qpe} = \sqrt{\eps^2 - C_{\gs}^2 \delta^4}$ in \cref{eq:total gates partially random}.
Finally, we find the minimal cost by optimizing for $L_D$ (which determines $\lambda_R$).
In the following we discuss two strategies for decomposing the Hamiltonian into deterministic and randomized parts, and the associated gate counts.

\subsubsection{Pauli decomposition}
A natural scheme for the partially randomized method is to start from the Hamiltonian in the Pauli sum decomposition after the fermion-to-qubit mapping, $H=\sum_{l=1}^L h_l P_l$.
We decompose into $H_D$ and $H_R$ by assigning the $L_D$ terms with largest weight $h_l$ to $H_D$, which makes it straightforward to find the optimal choice of $L_D$ based on \cref{eq:total gates partially random}.

If we count two-qubit gates, as discussed for deterministic product formulas, the terms in the deterministic part can be ordered lexicographically in order to achieve cancellations of \cnot gates.

If we count Toffoli gates we can again apply the modified Hamming weight phasing to the randomized part, as presented in \cref{sec:gate count randomized}.
We can additionally leverage this to reduce the gate synthesis cost of the deterministic part.
The idea is that while in principle the deterministic part consists of terms $h_l P_l$ with arbitrary real coefficients $h_l$, we may round $h_l$ to a desired value and assign the remainder to $H_R$.
This slightly increases the value of $\lambda_R$, but may significantly reduce the synthesis cost of time evolving along $P_l$ in the deterministic part.
This approach is similar to previous methods using randomization for compiling arbitrary-angle rotations \cite{campbell2017shorter,kliuchnikov2023shorter,koczor2024probabilistic}.
We obtain further reductions by expanding the rounded $h_l$ bitwise, and reorder terms so we can apply Hamming weight phasing.

We now explain this idea in detail.
Consider the binary expansion of the rotation angle of the $l$-th term
\begin{align}
    \frac{h_l\delta}{\pi} = \sum_{J \geq 1} x_{l_i} 2^{-J} \quad \text{with} \quad x_{l_J} \in \{0,1\}.
\end{align}
Grouping bits of the same significance and truncating to $n_{\max}$ bits of precision gives
\begin{align}
    \delta H = \sum_{J=1}^{n_{bit}} \underbrace{\sum_{l=1}^{L_J}\pi 2^{-J}P_{l}}_{\delta H_J} + \delta H_{\text{rem}},
\end{align}
where $H_{\text{rem}}$ contains the remaining higher bits of precision.
For the deterministic product formula, we now apply all the terms in $H_J$ consecutively.
This consists of $L_J$ Pauli rotations with the same angle, so we can use Hamming weight phasing.
Moreover, for the deterministic part we are free to choose the order of the Pauli rotations, and since any individual term anticommutes with $\bigO(N^3)$ out of the total $\bigO(N^4)$ terms, one can group into long sequences of commuting terms.
While decomposing into the $H_J$ increases the number of terms, they are cheaper to implement, reducing the total cost.

We now comment on two subtleties which influence this approach.
Firstly, one needs to group the terms in each $H_J$ into commuting groups.
Such groupings have been studied extensively for the sake of reducing measurement overheads when measuring energies \cite{yen2020measuring,yen2023deterministic}.
While one can typically find quite large groupings, this can be computationally intensive. For our compilation purposes, if we have groups of $K$ commuting Pauli operators, implementing a rotation over angle $\pi 2^{-J}$ requires $1 + (J-2)/K$ Toffoli gates per Pauli. This means that groups of moderate size suffice to have a low overhead, which can be constructed using simple heuristics.
As a lower bound, which we use for our numerical estimates, when considering \cref{eq:hamiltonian majorana}, it is not hard to see that it is possible to group into groups of 12 commuting Pauli operators.
The second subtlety is that the procedure of decomposing each term into several bits and reordering the terms will in principle affect the Trotter error.
The effect of ordering on Trotter error has been studied numerically for molecular electronic structure Hamiltonians \cite{tranter2019ordering}.
We expect that decomposing by bits of precision will not increase the Trotter error (since it decomposes into shorter time evolutions, so this should be beneficial).
For simplicity we assume that the Trotter error is the same as for the fully deterministic product formula and leave a detailed numerical exploration to future work.
For FeMoco, we find a total cost of $1.5 \times 10^{12}$ Toffoli gates using this strategy, as shown in \cref{fig:femoco intro}.
Note that this approach increases the number of ancilla qubits, since we need to prepare phase states for different bits of precision $J$.

\subsubsection{Scaling for the hydrogen chain}
As a benchmark system we use the hydrogen chain. The result for $\eps = 0.0016$ is shown in the main text in \cref{fig:hchain cost}.
We can also vary over $\eps$ as well, we can fit an expression of the form $C N^a \eps^{-b}$ to the resulting cost.
In \cref{fig:asymptotics hchain} we show scalings with respect to $N$ and $\eps$ (keeping the other variable fixed).
We find a fit for the Toffoli cost of $\bigO(N^2 \eps^{-1.7})$.

\begin{figure}[ht]
    \centering
    \includegraphics[width=0.45\linewidth]{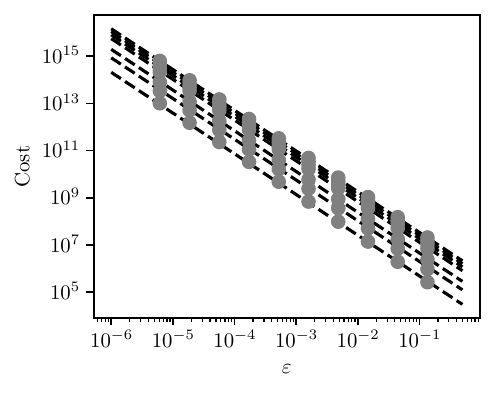} \hspace{0.5cm}
    \includegraphics[width=0.45\linewidth]{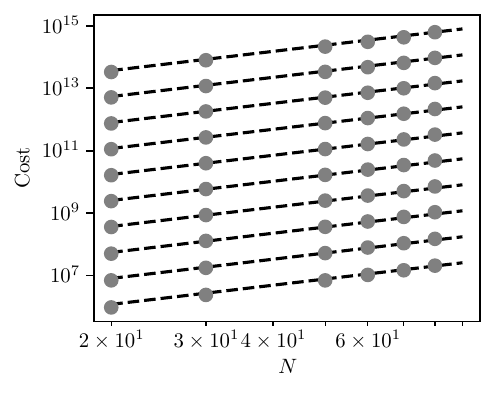}
    \caption{Toffoli cost for partially randomized product formulas with respect to varying $N$ and $\eps$, with a fit to a scaling $\bigO(N^2 \eps^{-1.7})$.}
    \label{fig:asymptotics hchain}
\end{figure}

In \cref{fig:hchain 2qubit} we show the number of \cnot gates required for phase estimation using deterministic, randomized and partially randomized product formulas. Note that this has a different scaling than the Toffoli count, since the number of two-qubit gates required per Pauli rotation scales with $N$. As for the Toffoli gate count, we observe a significant improvement using partial randomization.

\begin{figure}[ht]
    \centering
    \includegraphics[width=0.45\linewidth]{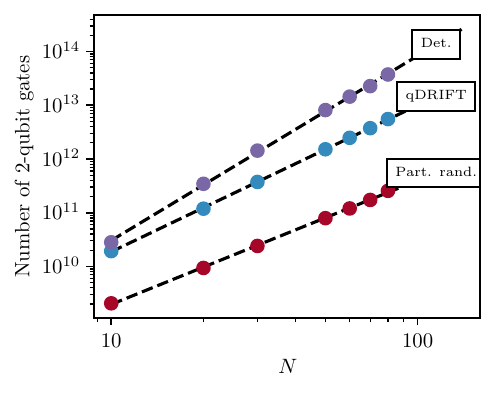}
    \caption{Two-qubit gate count in terms of \cnot required for the hydrogen chain for precision $\eps = 0.0016$.}
    \label{fig:hchain 2qubit}
\end{figure}

\subsubsection{Double factorized decomposition}
If we truncate to $L_D$ terms in the factorization, where the $l$-th term has rank at most $\rho_l$, then for the deterministic part we have cost as described in \cref{sec:gate count factorized}.
The cost of the randomized part remains the same.
The performance of this approach for the hydrogen chain is shown in \cref{fig:gate counts factorized hchain}.

\end{document}